\documentclass[%
 reprint,
 raggedbottom,
 amsmath,amssymb,
 aps,prx,
 longbibliography,
 floatfix,
 superscriptaddress,
]{revtex4-2}

\usepackage{graphicx}
\usepackage{dcolumn}
\usepackage{amsmath,amssymb,array}
\usepackage{bm}
\usepackage{quantikz}
\usepackage{amsthm}
\usepackage{mathtools}
\usepackage{braket}
\usepackage{placeins}
\usepackage{float}
\usepackage{mathrsfs}
\usepackage{enumitem}
\usepackage{tikz}
\usetikzlibrary{arrows.meta,positioning,calc,fit,backgrounds}

\usepackage[colorlinks=true,citecolor=blue,linkcolor=blue,urlcolor=blue]{hyperref}

\DeclareUnicodeCharacter{FF1C}{<}

\newtheorem{theorem}{Theorem}
\newtheorem{lemma}{Lemma}
\newtheorem{proposition}{Proposition}
\newtheorem{corollary}{Corollary}
\newtheorem{definition}{Definition}

\theoremstyle{remark}

\newcommand{\tr}{\mathrm{tr}}
\newcommand{\Tr}{\mathrm{Tr}}
\newcommand{\rank}{\mathrm{rank}}

\providecommand{\PSL}{\operatorname{PSL}}

\providecommand{\SL}{\operatorname{SL}}
\providecommand{\Sp}{\operatorname{Sp}}

\newcommand{\Loc}{L}

\newcommand{\Ksc}{K_{\mathrm{sc}}}

\begin{document}

\preprint{APS/PRXQ-DRAFT}

\title{Local Universality and Structural Certificates for Minimal Fixed-Depth Two-Qutrit Gate Decomposition}

\author{Yurui Liu}
\thanks{These authors contributed equally to this work.}
\affiliation{National Laboratory of Solid State Microstructures and School of Physics, Collaborative Innovation Center of Advanced Microstructures, Nanjing University, Nanjing 210093, China}
\affiliation{Jiangsu Key Laboratory of Quantum Information Science and Technology, Nanjing University, Suzhou 215163, China}
\author{Ruoting Dou}
\thanks{These authors contributed equally to this work.}
\affiliation{National Laboratory of Solid State Microstructures and School of Physics, Collaborative Innovation Center of Advanced Microstructures, Nanjing University, Nanjing 210093, China}
\affiliation{Jiangsu Key Laboratory of Quantum Information Science and Technology, Nanjing University, Suzhou 215163, China}
\author{Peng Xu}
\affiliation{National Laboratory of Solid State Microstructures and School of Physics, Collaborative Innovation Center of Advanced Microstructures, Nanjing University, Nanjing 210093, China}
\affiliation{Institute of Quantum Information and Technology, Nanjing University of Posts and Telecommunications, Nanjing 210003, China}
\affiliation{State Key Laboratory of Quantum Optics and Devices, Shanxi University, Taiyuan 030006, China}
\author{Xinsheng Tan}
\affiliation{National Laboratory of Solid State Microstructures and School of Physics, Collaborative Innovation Center of Advanced Microstructures, Nanjing University, Nanjing 210093, China}
\affiliation{Jiangsu Key Laboratory of Quantum Information Science and Technology, Nanjing University, Suzhou 215163, China}
\author{Shengjun Wu}
\email{sjwu@nju.edu.cn}
\affiliation{National Laboratory of Solid State Microstructures and School of Physics, Collaborative Innovation Center of Advanced Microstructures, Nanjing University, Nanjing 210093, China}
\affiliation{Jiangsu Key Laboratory of Quantum Information Science and Technology, Nanjing University, Suzhou 215163, China}
\author{Yang Yu}
\email{yuyang@nju.edu.cn}
\affiliation{National Laboratory of Solid State Microstructures and School of Physics, Collaborative Innovation Center of Advanced Microstructures, Nanjing University, Nanjing 210093, China}
\affiliation{Jiangsu Key Laboratory of Quantum Information Science and Technology, Nanjing University, Suzhou 215163, China}
\author{Zeng-Bing Chen}
\email{zbchen@nju.edu.cn}
\affiliation{National Laboratory of Solid State Microstructures and School of Physics, Collaborative Innovation Center of Advanced Microstructures, Nanjing University, Nanjing 210093, China}

\date{\today}
\begin{abstract}
We study a dimension-saturating fixed-core ansatz in which four copies of a fixed, non-tunable two-qutrit core \(K\in SU(9)\) are interleaved with five adjustable local layers from \(L=SU(3)\otimes SU(3)\).
Since \(\dim SU(9)=80\) and \(5\dim L=80\), this is the shortest fixed-core architecture not excluded by parameter counting.
We formulate the smooth map \(\Phi_K:L^5\to SU(9)\) and use its right-trivialized differential to give verifiable certificates for local universality.
We construct an explicit Clifford-word core whose Pauli-label splitting makes the identity-point differential an exact isometry, and we classify all 2304 symplectic actions satisfying the same splitting criterion.
We also prove a structural obstruction for an important symmetry class: every complex-symmetric core \(K=K^{T}\), including every core generated by a time-independent real-symmetric Hamiltonian in the chosen computational basis, has identity-point differential rank at most 78; hence any full-rank certificate for such a core must occur away from that point.
We then assess a hardware-motivated superconducting core generated by a noncommuting, temporally asymmetric drive. 
Direct calculation verifies \(K_{\rm sc}\neq K_{\rm sc}^{\mathsf T}\), and the core achieves \(F_{\rm avg}\ge0.999\) for all 1000 Haar-random targets tested under the stated restart protocol. 
We also report favorable sampled Jacobian-rank, structured-target, and robustness diagnostics.
These results establish local universality at the parameter-counting-minimal, dimension-saturating depth, with an exact Clifford certificate complemented by a hardware-motivated numerical case study.
Throughout, we separate exact local certificates from numerical evidence for broader synthesis performance.
\end{abstract}

\maketitle

\section{Introduction}
\label{sec:introduction}

Quantum gate decomposition expresses a target unitary as a finite sequence of implementable elementary operations.
For qubits, universality, Solovay--Kitaev approximation, Cartan/KAK decompositions, local invariants, and optimal CNOT-based two-qubit constructions give a mature synthesis theory~\cite{Barenco1995,NielsenChuang2010,DawsonNielsen2006,Khaneja2001,Makhlin2002,Zhang2003,VidalDawson2004,VatanWilliams2004,Mottonen2004,Shende2006}.
Three-level qutrits can reduce circuit complexity~\cite{WangHuSandersKais2020}.
For qutrits and higher-dimensional qudits, general universality and synthesis results are known~\cite{Brylinski2002,MuthukrishnanStroud2000,Bullock2005,Brennen2006,DiWei2013,NikolaevaKiktenkoFedorov2024,KiktenkoNikolaevaFedorov2025}, but explicit, short, and hardware-aware decompositions applicable to Haar-almost every two-qutrit gate remain less developed~\cite{DiZhangWei2008,Fischer2023,JiangLiuWei2024}.
This gap is increasingly relevant because superconducting platforms~\cite{Krantz2019,Kjaergaard2020,Blais2021}, particularly transmon circuits~\cite{Koch2007}, provide experimentally accessible higher levels that can encode qutrits, and recent experiments have demonstrated single-qutrit gates~\cite{Yu2026SingleQutrit}, qutrit randomized benchmarking, programmable two-qutrit processing, and two-qutrit entangling gates~\cite{Yurtalan2020,Morvan2021,Goss2022,Luo2023,Roy2023,Goss2024QutritProcessor,SubramanianLupascu2023}.

The present work is motivated by fixed-entangling-gate synthesis in superconducting qutrit circuits~\cite{Goss2022,Fischer2023,SubramanianLupascu2023,Roy2023}.
In particular, high-fidelity qutrit entangling gates have been used as fixed nonlocal elements interleaved with arbitrary local qutrit operations, and a six-layer fixed-entangling-gate network was shown numerically to synthesize \(1000\) Haar-random two-qutrit targets~\cite{Goss2022}.
This raises the minimal-depth question of whether the nonlocal depth can be reduced to the dimension-saturating four-core architecture, in which four identical two-qutrit cores are interleaved with five adjustable local layers.
Because \(SU(9)\) has dimension \(80\) and each local layer \(SU(3)\otimes SU(3)\) has dimension \(16\), the shallowest fixed-core architecture not ruled out by parameter counting contains four cores.
At this depth, however, parameter counting is only a necessary condition.
The five local layers provide exactly the right number of real parameters, but these parameters may still fail to open all directions in \(SU(9)\) if the fixed core has hidden symmetries or transports the local tangent directions into overlapping subspaces.
Thus the central issue is not simply how many local parameters are present, but whether a given core converts those parameters into independent nonlocal directions.
This viewpoint leads naturally to a differential certificate: evaluate the rank of the fixed-core synthesis map after trivializing the tangent space of \(SU(9)\).
If the rank is full at one point, the image contains an open neighborhood of the corresponding gate.
If the rank is deficient everywhere, the architecture cannot be locally universal.

Our main contributions are fourfold.
First, we formulate the four-core, five-local-layer ansatz as a smooth map \(\Phi_K:L^5\to SU(9)\) and derive a full-rank differential certificate for local universality.
Second, we construct an explicit Clifford-word core \(K_{\rm Cl}\) whose Pauli-label dynamics gives an exact full-rank certificate at the identity local point; with respect to the product Hilbert--Schmidt metric, the certified differential is an exact isometry and has condition number one.
Third, we prove that every complex-symmetric core \(K=K^{T}\)---including every core generated by a time-independent real Hamiltonian and every core locally conjugate to a complex-symmetric core---has differential rank at most \(78\) at the identity local point. 
Thus no core in this class admits an identity-local-point full-rank certificate, although the theorem does not constrain other local points.
Fourth, we test a hardware-motivated superconducting Hamiltonian core \(K_{\rm sc}\) generated by a temporally asymmetric, time-dependent drive. 
It is not complex symmetric, so the rank-\(78\) theorem does not apply to it. 
In the reported numerical tests, no sampled Jacobian falls below the stated rank tolerance, and the same fixed core achieves \(F_{\rm avg}\ge0.999\) for all 1000 Haar-random targets under the prescribed restart protocol; we also evaluate structured targets and shared static coefficient offsets.
We therefore pair an exact structural certificate for a minimal tight architecture with a hardware-motivated numerical core that probes how much of this tangent-spreading behavior survives in a compact physical model.
Our results have three distinct logical statuses.
The Clifford Pauli-label splitting, the identity-point isometry, and the symmetric-core rank bound are exact statements.
The pointwise, synthesis, and regular-pair computations are finite-precision diagnostics.
Global surjectivity of the four-core map remains open.
From a practical perspective, the framework provides geometric criteria for selecting fixed entangling cores and designing compact, hardware-motivated two-qutrit synthesis protocols.

\section{FIXED-CORE MAP AND LOCAL CERTIFICATE}
\label{sec:geometry}

\subsection{Dimension-saturating ansatz}
\label{subsec:minimal-ansatz}

Let
\begin{equation}
    G:=SU(9),\qquad \mathfrak g:=\mathfrak{su}(9),
\end{equation}
and let
\begin{equation}
    L:=SU(3)\otimes SU(3),\quad
    \mathfrak l:=\mathrm{Lie}(L)\cong \mathfrak{su}(3)\oplus\mathfrak{su}(3).
\end{equation}
The fixed-core map studied here is
\begin{equation}
\begin{split}
  \Phi_K:L^5&\longrightarrow SU(9),\\
  (L_1,\ldots,L_5)&\longmapsto L_5KL_4KL_3KL_2KL_1 .
\end{split}
\label{eq:fixed-core-map}
\end{equation}
It contains four applications of a parameterless nonlocal core \(K\in SU(9)\) and five adjustable local layers \(L_i\in L\); see Fig.~\ref{fig:four-core-ansatz}.
Four cores constitute the smallest fixed-core depth not excluded by dimension counting.
Since \(\dim SU(9)=80\) and \(\dim L=16\), a depth-\(\ell\) fixed-core architecture has \(16(\ell+1)\) tunable local parameters, and \(\ell=4\) is the first depth for which this number can reach \(80\).
Thus Eq.~\eqref{eq:fixed-core-map} is the minimal dimension-saturating fixed-core ansatz for two-qutrit synthesis.

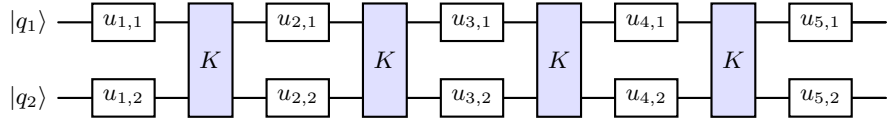
\begin{figure*}[t]
\centering
\begin{quantikz}[column sep=0.45cm]
\lstick{$\ket{q_1}$} & \gate{u_{1,1}} & \gate[2,style={fill=blue!12}]{K} & \gate{u_{2,1}} & \gate[2,style={fill=blue!12}]{K} & \gate{u_{3,1}} & \gate[2,style={fill=blue!12}]{K} & \gate{u_{4,1}} & \gate[2,style={fill=blue!12}]{K} & \gate{u_{5,1}} & \qw \\
\lstick{$\ket{q_2}$} & \gate{u_{1,2}} &                                      & \gate{u_{2,2}} &                                      & \gate{u_{3,2}} &                                      & \gate{u_{4,2}} &                                      & \gate{u_{5,2}} & \qw
\end{quantikz}
\caption{Dimension-saturating four-core ansatz. Each local layer is \(L_i=u_{i,1}\otimes u_{i,2}\in SU(3)\otimes SU(3)\), and the nonlocal core \(K\in SU(9)\) is fixed.}
\label{fig:four-core-ansatz}
\end{figure*}

\subsection{Differential certificate}
\label{subsec:right-trivialized-differential}

For \(U\in SU(9)\), we use right trivialization of the target tangent space,
\begin{equation}
    \tau_U^R(\dot U)=\dot U U^{-1}.
    \label{eq:right-trivialization-target}
\end{equation}
For the local layers, write \(X_i=L_i^{-1}\dot L_i\in\mathfrak l\).
Let the prefix products be
\begin{align}
    P_5(\mathbf L)&:=L_5,\nonumber\\
    P_4(\mathbf L)&:=L_5KL_4,\nonumber\\
    P_3(\mathbf L)&:=L_5KL_4KL_3,\nonumber\\
    P_2(\mathbf L)&:=L_5KL_4KL_3KL_2,\nonumber\\
    P_1(\mathbf L)&:=L_5KL_4KL_3KL_2KL_1=\Phi_K(\mathbf L).
    \label{eq:prefix-products}
\end{align}
The right-trivialized differential is the linear map
\begin{equation}
    A_K(\mathbf L)(X_1,\ldots,X_5)
    =
    \sum_{i=1}^{5}
    \operatorname{Ad}_{P_i(\mathbf L)}(X_i),
    \label{eq:AK-definition}
\end{equation}
and \(\operatorname{rank}D\Phi_K(\mathbf L)=\operatorname{rank}A_K(\mathbf L)\).
At the identity local point \(\boldsymbol I=(I,I,I,I,I)\), this reduces to
\begin{equation}
    A_K(\boldsymbol I)(X_1,\ldots,X_5)
    =
    \sum_{i=1}^{5} \operatorname{Ad}_{K^{5-i}}(X_i).
    \label{eq:AK-identity}
\end{equation}
Thus the fixed-core circuit is locally expressive exactly when five transported copies of the local Lie algebra span \(\mathfrak{su}(9)\).
Equivalently,
\begin{equation}
    \operatorname{im}A_K(\mathbf L)
    =
    \sum_{i=1}^{5}
    \operatorname{Ad}_{P_i(\mathbf L)}(\mathfrak l)
    \subseteq \mathfrak g .
    \label{eq:transported-local-sum}
\end{equation}
The role of \(K\) is therefore to rotate the five local tangent spaces into sufficiently independent directions.
This is the geometric object tested in all later examples.

\begin{proposition}[Right-trivialized differential]
\label{prop:right-trivialized-differential}
For \(U=\Phi_K(\mathbf L)\) and \(X_i=L_i^{-1}\dot L_i\in\mathfrak l\),
\[
    \tau_U^R\!\left(D\Phi_K(\mathbf L)[\dot L_1,\ldots,\dot L_5]\right)
    =
    \sum_{i=1}^{5}\operatorname{Ad}_{P_i(\mathbf L)}(X_i).
\]
Hence \(D\Phi_K(\mathbf L)\) has full rank if and only if \(A_K(\mathbf L)\) has full rank.
\end{proposition}
The proof is a direct product differentiation and is given in Appendix~\ref{app:differential-proofs}.

Choose orthonormal bases of \(\mathfrak l^5\) and \(\mathfrak g\) with respect to the real Hilbert--Schmidt form \(\langle X,Y\rangle=-\operatorname{Tr}(XY)\), and let \(M_K(\mathbf L)\in\mathbb R^{80\times80}\) represent \(A_K(\mathbf L)\).
The local certificate is
\begin{equation}
    \sigma_{\min}\!\left(M_K(\mathbf L_\ast)\right)>0
    \Rightarrow
    \Phi_K \text{ is locally surjective near } \mathbf L_\ast .
    \label{eq:smin-certificate}
\end{equation}
At any fixed point \(\mathbf L_\ast\), an exact positive value \(\sigma_{\min}(M_K(\mathbf L_\ast))>0\) is equivalent to full rank of \(D\Phi_K(\mathbf L_\ast)\).
The inverse function theorem then implies that the image contains an open neighborhood of \(\Phi_K(\mathbf L_\ast)\).
We record this as a theorem because it is the rigorous local-universality statement used below.
\begin{theorem}[Full rank implies local universality]
\label{thm:full-rank-local-universality}
If there exists \(\mathbf L_\ast\in L^5\) such that
\(\sigma_{\min}(M_K(\mathbf L_\ast))>0\), then \(\Phi_K\) is locally
surjective near \(\mathbf L_\ast\).  In particular,
\(\operatorname{Im}\Phi_K\) contains an open subset of \(SU(9)\) and has
positive Haar measure.
\end{theorem}
The theorem follows from the inverse function theorem; see Appendix~\ref{app:differential-proofs}.

\paragraph*{Regularity at the identity local point for Haar-almost every core.}
The full-rank condition at the identity local point holds for Haar-almost every core once it is satisfied by one core.
Indeed, after choosing bases, the determinant of the differential matrix at the identity local point, \(M_K(\boldsymbol I)\), is a real-analytic function of \(K\in SU(9)\).
Since the Clifford construction below gives one core with \(\det M_K(\boldsymbol I)\ne0\), this analytic function is not identically zero.
Hence the set
\[
\{K\in SU(9):\operatorname{rank}M_K(\boldsymbol I)<80\}
\]
has Haar measure zero.
Thus Haar-almost every core is locally regular at the identity local point.
The Clifford construction is valuable not because identity-point regularity is rare, but because it provides an exact finite-field certificate and a classifiable reference family. 
The measure-zero exceptional set nevertheless contains structurally significant core families.
As shown in Appendix~\ref{app:differential-proofs} (Proposition~\ref{prop:symmetric-core-identity-critical}), it contains every complex-symmetric core \(K=K^{T}\): for such cores \(\operatorname{rank}M_K(\boldsymbol I)\le78\).
Cores generated by real time-independent Hamiltonians are complex symmetric, and the bound extends to every core that is conjugate to a complex-symmetric core by a local unitary.
The hardware-motivated core of Sec.~\ref{sec:superconducting} is designed to leave this class: its evolution is noncommuting and temporally asymmetric, and direct calculation verifies \(K_{\rm sc}\ne K_{\rm sc}^{T}\). Consequently, Proposition~\ref{prop:symmetric-core-identity-critical} does not determine its identity-local-point rank, which is instead evaluated directly as a numerical diagnostic.

We also use two pointwise diagnostics in the numerical sections: the smallest singular value \(\sigma_{\min}(A_{\rm pt})\) of the sampled pointwise differential matrix and the effective tangent-volume indicator \(v_{5,\mathrm{eff}}\), defined below as the product of all \(80\) singular values of \(A_{\rm pt}\).
These diagnostics measure the weakest opened tangent direction and the local tangent-volume scale, respectively.
The numerical implementation uses the pointwise convention
\begin{equation}
    \Phi^{\rm pt}_{5,K}(L_1,\ldots,L_5)
    =
    L_5KL_4KL_3KL_2KL_1K ,
    \label{eq:pointwise-map}
\end{equation}
which differs from Eq.~\eqref{eq:fixed-core-map} by a fixed right multiplication by \(K\).
This does not change the intrinsic rank interpretation because right multiplication by a fixed unitary is a diffeomorphism of \(SU(9)\).
At a sampled point, the corresponding right-trivialized differential is represented by
\[
    A_{\rm pt}(K,\mathbf L)\in\mathbb R^{80\times80}.
\]
The Gram matrix
\[
    \gamma_{\rm pt}(K,\mathbf L)=A_{\rm pt}(K,\mathbf L)A_{\rm pt}(K,\mathbf L)^T
\]
has eigenvalues equal to the squared singular values of \(A_{\rm pt}\).
We report \(\sigma_{\min}(A_{\rm pt})\) as the weakest opened direction and
\[
    v_{5,\mathrm{eff}}(K,\mathbf L)
    :=
    \prod_{r=1}^{80}\sigma_r(A_{\rm pt}(K,\mathbf L))
\]
as the local tangent-volume indicator.
This product vanishes if any singular value vanishes.
It is not a new mathematical certificate beyond full rank, but it is useful for comparing sampled tangent maps that show no rank collapse yet are differently conditioned.

\subsection{Local universality versus global surjectivity}
\label{subsec:local-global-critical}

The certificate is local.
Since \(L^5\) is compact, \(\operatorname{Im}\Phi_K\) is closed, but a full-rank point only proves that the image has nonempty interior.
A global-surjectivity proof requires additional control of the global image and its critical fibers; one sufficient route is to show that every reachable value has at least one regular preimage.
This distinction is used below: the Clifford-word construction gives an exact local certificate, whereas the superconducting core is supported by pointwise numerical and synthesis diagnostics.
We use the notation
\begin{align}
    \Sigma_K
    &:=
    \{\mathbf L\in L^5:\operatorname{rank}D\Phi_K(\mathbf L)<80\},
    \label{eq:critical-set}\\
    C_K
    &:=
    \Phi_K(\Sigma_K)
    \label{eq:critical-values}
\end{align}
for the critical set and critical value set.
Two elementary facts organize the global question.
First, \(\operatorname{Im}\Phi_K\) is compact and therefore closed in \(SU(9)\).
Second, if every reachable value has at least one regular preimage, then the image is also open; since \(SU(9)\) is connected, this would imply \(\operatorname{Im}\Phi_K=SU(9)\).
Thus boundary points, if any, must be critical values whose entire fiber is critical.
This is the reason regular partners matter: a critical value that also has a regular preimage cannot be a local image wall.
The closed-image and regular-preimage criteria, together with the reason a nonzero-degree argument does not directly settle the global question for this family of maps, are recorded in Appendix~\ref{app:differential-proofs}.

\section{CLIFFORD-WORD CORE WITH AN EXACT LOCAL CERTIFICATE}
\label{sec:clifford}

\subsection{Pauli-label splitting certificate}
\label{subsec:pauli-label-splitting}

The Clifford construction reduces the differential-rank problem to finite arithmetic.
Let \(V=\mathbb F_3^4\), with labels \(v=(a_1,a_2\mid b_1,b_2)\) for two-qutrit Pauli directions.
Define the one-qutrit label planes
\begin{equation}
    \begin{aligned}
    V_A &=\{(a,0\mid b,0):a,b\in\mathbb F_3\}, \\
    V_B &=\{(0,a\mid 0,b):a,b\in\mathbb F_3\},
    \end{aligned}
\end{equation}
and the local label block
\begin{equation}
    \Lambda_{\rm loc}:=(V_A\setminus\{0\})\cup(V_B\setminus\{0\}).
    \label{eq:local-label-block}
\end{equation}
For a Clifford core \(K\), conjugation induces a symplectic action \(S_K\in Sp(4,\mathbb F_3)\) on Pauli labels~\cite{Gottesman1999,HostensDehaeneDeMoor2005}.
If
\begin{equation}
    V\setminus\{0\}=\bigsqcup_{t=0}^4S_K^t(\Lambda_{\rm loc}),
    \label{eq:five-block-splitting-condition}
\end{equation}
then the five transported local Pauli-label blocks form an orthogonal decomposition of the \(80\)-dimensional traceless Pauli label space.
Consequently,
\begin{equation}
    \operatorname{rank}A_K(\boldsymbol I)=80.
    \label{eq:identity-local-point-full-rank}
\end{equation}
This Pauli-label splitting criterion is an exact finite certificate for local universality at the identity local point.
It is stronger than a numerical singular-value calculation: once the finite partition in Eq.~\eqref{eq:five-block-splitting-condition} is verified, the transported Pauli subspaces are disjoint by construction and their dimensions add to \(5\cdot16=80\).
The only continuous input is the fact that a determinant-one Clifford lift implements the corresponding symplectic label action by conjugation on Pauli directions.
Thus the Jacobian-rank statement is reduced to a finite-field partition problem over \(V\setminus\{0\}\).
The sign-pair real-basis convention and the full five-block label table are given in Appendix~\ref{app:explicit-label-splitting}.

\subsection{Explicit short-word representative}
\label{subsec:short-clifford-word}

Let \(\omega=e^{2\pi i/3}\), let \(F\) be the one-qutrit Fourier transform, let \(P=\operatorname{diag}(1,1,\omega)\), and let
\(\operatorname{SUM}|x,y\rangle=|x,x+y\rangle\) over \(\mathbb F_3\).
With the rightmost factor acting first, define
\begin{equation}
    K_{\rm Cl}
    =
    e^{5i\pi/6}
    \mathrm{SUM}^{\dagger}P_2\mathrm{SUM}F_1\mathrm{SUM}.
    \label{eq:KCl-definition}
\end{equation}
The phase fixes the determinant-one representative and does not affect the adjoint action.
In the label order \((a_1,a_2\mid b_1,b_2)\), its symplectic matrix is
\begin{equation}
    S_{K_{\rm Cl}}
    =
    \begin{pmatrix}
        0&0&2&1\\
        1&1&0&0\\
        2&1&2&1\\
        1&1&2&2
    \end{pmatrix}
    \in Sp(4,\mathbb F_3),
    \label{eq:SKCl-matrix}
\end{equation}
with \(S_{K_{\rm Cl}}^5=I_4\).
A finite-field check gives the disjoint splitting
\begin{equation}
    V\setminus\{0\}
    =
    \Lambda_0\sqcup\Lambda_1\sqcup\Lambda_2\sqcup\Lambda_3\sqcup\Lambda_4,
    \quad
    \Lambda_t=S_{K_{\rm Cl}}^t(\Lambda_{\rm loc}).
    \label{eq:KCl-five-block-splitting}
\end{equation}

\begin{theorem}[Short Clifford-word local certificate]
\label{thm:short-clifford-word-local-certificate}
The four-core map generated by \(K_{\rm Cl}\),
\[
    \Phi_{K_{\rm Cl}}(L_1,\ldots,L_5)
    =
    L_5K_{\rm Cl}L_4K_{\rm Cl}L_3K_{\rm Cl}L_2K_{\rm Cl}L_1,
\]
has full-rank differential at the identity local point:
\begin{equation}
    \operatorname{rank}A_{K_{\rm Cl}}(\boldsymbol I)=80.
\end{equation}
Therefore \(\Phi_{K_{\rm Cl}}\) is locally surjective near \(\boldsymbol I\).
\end{theorem}
The theorem follows immediately from the splitting in Eq.~\eqref{eq:KCl-five-block-splitting} and the Pauli-label criterion above; the full finite verification is in Appendix~\ref{app:explicit-label-splitting}.
The splitting in fact yields strictly more than full rank.
Distinct Pauli labels index Hilbert--Schmidt-orthogonal directions, so the five transported blocks span mutually orthogonal subspaces, and each block is reached from \(\mathfrak l\) by the isometry \(\operatorname{Ad}_{K_{\rm Cl}^t}\).
Consequently \(A_{K_{\rm Cl}}(\boldsymbol I)\) is an exact isometry of \(\mathfrak l^5\) onto \(\mathfrak{su}(9)\): all \(80\) singular values of \(M_{K_{\rm Cl}}(\boldsymbol I)\) equal \(1\), and \(\sigma_{\min}(M_{K_{\rm Cl}}(\boldsymbol I))=1\) exactly (Corollary~\ref{cor:identity-point-isometry} in Appendix~\ref{app:explicit-label-splitting}).
The identity local point is therefore not merely regular but perfectly conditioned, which makes the broad conditioning tails observed at sampled configurations (Appendix~\ref{app:clifford-diagnostics}) a genuinely pointwise phenomenon rather than a weakness of the certificate.
Because \(K_{\rm Cl}^5=I_9\), the identity local point maps to \(K_{\rm Cl}^{-1}\).
The theorem therefore proves that the four-core reachable set contains an open neighborhood of this explicit two-qutrit Clifford gate.
It does not prove that every element of \(SU(9)\) is reachable; the obstruction, if present, would have to arise from global fiber structure rather than from the local tangent rank at the certified point.
The finite classification yields \(2304\) good symplectic actions, divided equally between elements of orders \(5\) and \(10\).
After choosing compatible determinant-one representatives, all good Clifford actions admit lifts in a single two-sided local double coset and therefore define the same reachable set after reparameterization of the local layers.
Arbitrary determinant-one lifts may additionally differ by a central ninth root of unity.
Concretely, if two good lifts satisfy \(K_j=AK_iB\) with \(A,B\in L\), then
\[
\begin{aligned}
L_5K_jL_4K_jL_3K_jL_2K_jL_1
&=
(L_5A)K_i(BL_4A)K_i(BL_3A)\\
&\quad \times K_i(BL_2A)K_i(BL_1),
\end{aligned}
\]
and every parenthesized factor is again local.
The inverse equivalence gives the reverse inclusion of reachable sets.
Thus the short word \(K_{\rm Cl}\) can be used as a canonical representative of the full good Clifford-word class.
The classification, order split, and local-equivalence proof are given in Appendix~\ref{app:clifford-family}.

The same inversion relation also gives a small exact global reachability statement.
As shown in Appendix~\ref{subsec:inversion-symmetry-reachable-set}, one has
\[
L\subseteq \operatorname{Im}\Phi_{K_{\rm Cl}}.
\]
Thus all purely local two-qutrit gates are exactly reachable by the same four-core Clifford architecture.
This does not imply global surjectivity, since \(L\) is a lower-dimensional subgroup of \(SU(9)\), but it provides an exact global consistency check complementary to the local full-rank certificate.

\subsection{Numerical Clifford-core diagnostics}
\label{subsec:clifford-pointwise-diagnostics}

We also tested the numerical conditioning of the same core away from the identity local point.
For each of the \(2304\) certifying Clifford cores, we evaluated a shared pool of \(100\) independently Haar-sampled configurations in \(L^5\), giving \(230{,}400\) core--configuration pairs in total.
No sampled rank collapse is observed in this data set; the smallest recorded value of \(\sigma_{\min}(A_{\rm pt})\) is \(4.08\times10^{-9}\).
This finite-precision pointwise diagnostic should be interpreted as sampled conditioning evidence rather than a uniform analytic lower bound over \(L^5\).
For the synthesis benchmark, gate quality is reported using the standard average gate fidelity~\cite{PedersenMollerMolmer2007}
\begin{equation}
F_{\rm avg}(U,V)=
\frac{|\operatorname{Tr}(U^\dagger V)|^2+d}{d(d+1)},
\label{eq:favg-definition}
\end{equation}
with \(d=9\) for the two-qutrit case.
In a Haar-random synthesis benchmark, one randomly selected good Clifford-word core reached \(F_{\rm avg}\ge1-10^{-10}\) for all tested targets under the stopping protocol.
Because the optimizer terminates at the threshold, the histogram in Fig.~\ref{fig:clifford-haar-1em10} reports threshold-reaching behavior rather than maximum attainable fidelity.

\begin{figure*}[!tbp]
    \centering
    \begin{minipage}{0.48\textwidth}
        \centering
        \includegraphics[width=\linewidth]{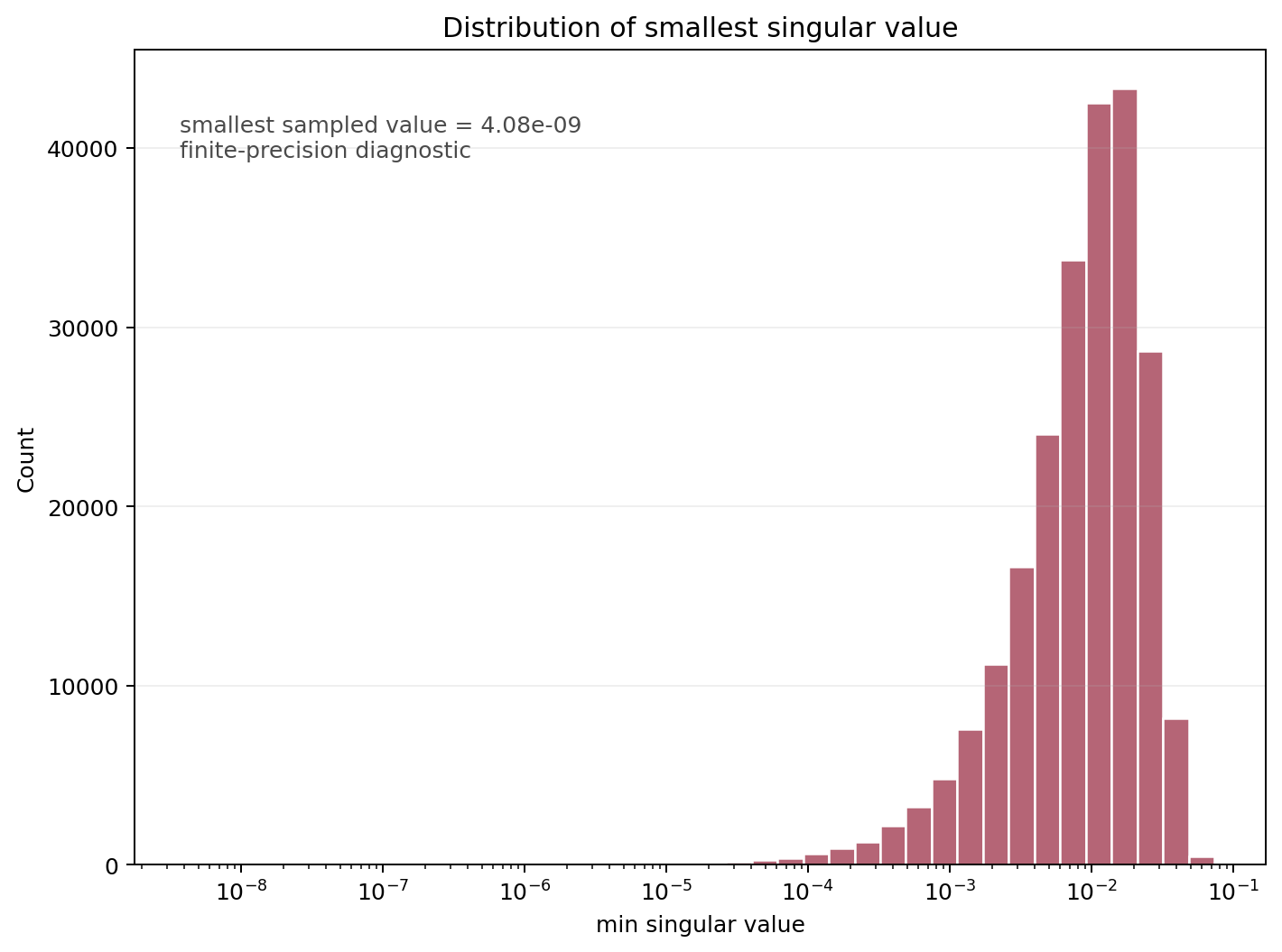}
        \vspace{-0.5em}
        \centerline{\small (a) Pointwise \(\sigma_{\min}(A_{\rm pt})\)}
    \end{minipage}
    \hfill
    \begin{minipage}{0.48\textwidth}
        \centering
        \includegraphics[width=\linewidth]{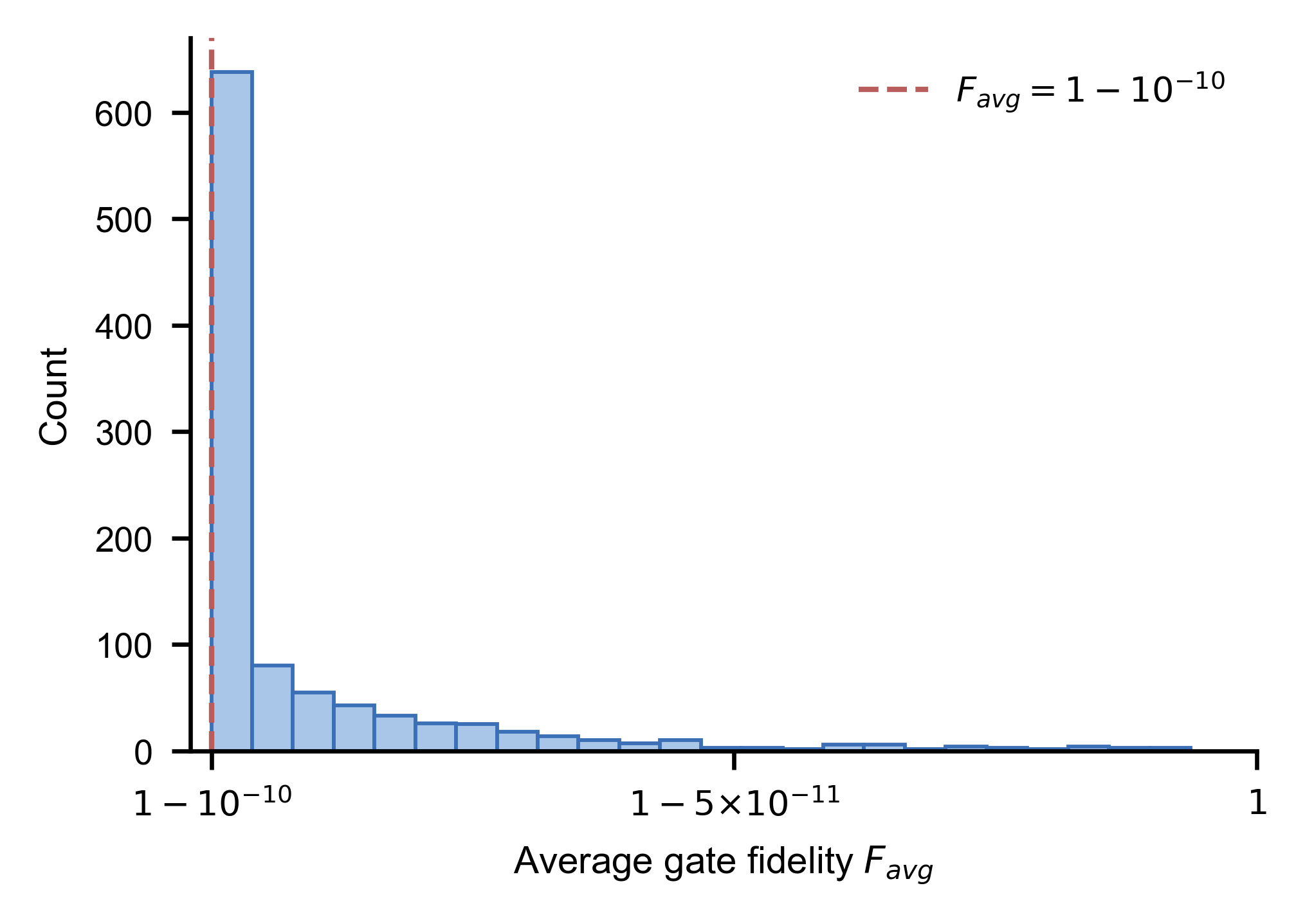}
        \vspace{-0.5em}
        \centerline{\small (b) Haar-random synthesis}
    \end{minipage}
    \caption{
    Compact Clifford-core diagnostics.
    Panel (a) shows the distribution of the sampled pointwise smallest singular value \(\sigma_{\min}(A_{\rm pt})\); the smallest recorded value, \(4.08\times10^{-9}\), lies well above the numerical rank floor, so no sampled rank collapse occurs.
    Panel (b) shows high-precision Haar-random synthesis for one randomly selected good Clifford-word core.
    The dashed line marks \(F_{\rm avg}=1-10^{-10}\).
    }
    \label{fig:clifford-haar-1em10}
\end{figure*}

Additional pointwise Jacobian diagnostics are reported in Appendix~\ref{app:clifford-diagnostics}.
The high-precision threshold \(1-10^{-10}\) is intentionally much stricter than the \(0.999\) threshold used later for the superconducting core.
It is used here because the Clifford-word construction is an algebraic reference case: the exact certificate at the identity local point and the finite-label structure make it a useful benchmark for asking how well the optimizer can exploit a particularly clean core.
The benchmark should not be interpreted as measuring the intrinsic maximum fidelity of \(K_{\rm Cl}\), since the stopping rule truncates optimization once the prescribed threshold has been reached.
Instead, it records whether the fixed four-core architecture reliably finds local layers that meet a very demanding target fidelity for gates sampled from Haar measure on \(SU(9)\).
The two numerical tests play different roles.
The pointwise Jacobian-rank diagnostics ask whether the transported local tangent spaces remain independent away from the identity local point.
The Haar-random synthesis benchmark asks whether this local differential regularity is useful for numerical compilation of the sampled Haar-random targets.
These diagnostics, together with the regular-pair computation below, have the finite-precision status summarized at the end of Sec.~\ref{sec:introduction}.
This combined evidence supports the use of \(K_{\rm Cl}\) as an algebraically ideal reference core for the minimal architecture.
The smallest singular values in the sampled pointwise test span several orders of magnitude, so the tail of the distribution is more informative than its visual bulk.
The finite-precision data do not imply uniform conditioning over all of \(L^5\), but they indicate that the exact regularity at the identity local point is not isolated in the sampled configurations.
The effective-volume diagnostic in Appendix~\ref{app:clifford-diagnostics} gives the complementary picture: even among full-rank sampled points, the local tangent-volume scale varies substantially.
This variation is expected in a nonlinear map at the dimension-saturating threshold, where there are no redundant local parameters to absorb poorly conditioned directions.

Appendix~\ref{app:regular-pairs} reports a \(100\)-candidate regular-pair verification for \(K_{\rm Cl}\).
For each sampled fold-like candidate, the numerical search finds a distinct regular partner whose output agrees with the candidate output within the stated Frobenius-residual tolerance.
The regular-pair data support the interpretation that the observed corank-one critical values are redundantly covered by regular sheets.
The same appendix records the operator-Schmidt and entangling-power diagnostics for \(K_{\rm Cl}\).
In particular, \(K_{\rm Cl}\) reaches the maximal two-qutrit entangling power \(e_p=1/2\) in the unnormalized linear-entropy convention.
This reinforces its role as a mathematically clean benchmark. 
It also motivates the hardware-oriented question in the next section: whether a core generated by a simpler superconducting Hamiltonian can retain enough of the same tangent-spreading behavior without being an exact Clifford splitting core.
The distinction between these two roles is important.
The Clifford core is not introduced as an immediately hardware-native pulse, but as a structural certificate that the dimension-saturating four-core depth is genuinely attainable at the local level.
It gives a concrete answer to the mathematical question left open by parameter counting.
The superconducting core then asks a different question: how much of this geometry survives when the core is constrained to arise from a compact Hamiltonian model with drift, drive, and exchange-type terms.
In particular, the absence of an exact splitting certificate for \(K_{\rm sc}\) does not preclude strong sampled synthesis performance.

\section{HARDWARE-MOTIVATED SUPERCONDUCTING CORE}
\label{sec:superconducting}

\subsection{Effective Hamiltonian core}
\label{subsec:sc-hamiltonian-model}

We next study a fixed core generated by a phenomenological, transition-resolved two-qutrit Hamiltonian motivated by superconducting-transmon control.
The model is restricted to the truncated computational subspace
\begin{equation}
    \mathcal H_{\rm qutrit}
    =
    \operatorname{span}\{|ij\rangle:i,j=0,1,2\}
    \cong \mathbb C^3\otimes\mathbb C^3 .
    \label{eq:sc-truncated-space}
\end{equation}
In this subspace, the coherent effective Hamiltonian is
\begin{equation}
H_{\rm sc}(t)=2\pi\sum_{\mu=1}^{4}
\left[
\Delta_\mu(t)H^{\rm drift}_\mu+
\Omega_\mu(t)H^{\rm drive}_\mu+
g_\mu(t)H^{\rm coup}_\mu
\right].
\label{eq:sc-hamiltonian}
\end{equation}
The explicit drift, drive, and coupling generators are listed in Appendix~\ref{app:sc-generators}.
It is not a full device-level simulation: leakage outside the three-level truncation, decoherence, waveform-bandwidth constraints, readout effects, and time-dependent control noise are not included; static Hamiltonian-coefficient offsets are considered separately in Appendix~\ref{app:robustness}.
While this boundary is important for practical purposes, the point we focus on here is not whether a complete experimental pulse has been optimized, but whether a physically structured fixed nonlocal core can pass the same differential and synthesis tests used for the Clifford benchmark.
The raw propagator and determinant-normalized \(SU(9)\) core are
\begin{align}
    K_{\rm sc}^{\rm raw}(T)
    &=
    \mathcal T\exp\left[-i\int_0^T H_{\rm sc}(t)\,dt\right],
    \label{eq:sc-raw-propagator}\\
    K_{\rm sc}
    &=
    K_{\rm sc}^{\rm raw}(T)
    \det\!\left(K_{\rm sc}^{\rm raw}(T)\right)^{-1/9}.
    \label{eq:sc-su9-core}
\end{align}
The representative time-dependent-control parameter set is given in Table~\ref{tab:sc-representative-parameters}.
All coefficients are constant except the first drive, for which we use the normalized skewed sine-squared envelope
\begin{equation*}
\Omega_1(t)=\frac{4~\mathrm{MHz}}{C_\eta}
\sin^2\!\left(\frac{\pi t}{T}\right)
\left[1+\eta\sin\!\left(\frac{2\pi t}{T}\right)\right].
\end{equation*}
Here \(0\le t\le T\), \(T=0.07~\mu\mathrm{s}\), \(\eta=0.35\), and \(C_\eta\) is the normalization constant chosen so that the peak of \(\Omega_1(t)\) is \(4~\mathrm{MHz}\).
The envelope vanishes at both endpoints and reaches its peak at approximately \(29.3~\mathrm{ns}\), before the pulse midpoint.
The time-ordered propagator in Eq.~\eqref{eq:sc-raw-propagator} is evaluated with QuTiP~\cite{qutip5}, using \texttt{sesolve} with the \texttt{vern9} integrator.
The branch of \(\det(\cdot)^{-1/9}\) fixes the \(SU(9)\) representative only up to a central ninth root of unity. 
The Jacobian, average-fidelity, operator-Schmidt, and entangling-power diagnostics are invariant under this phase, whereas matrix-level visualizations use the explicitly specified representative and, where appropriate, global-phase alignment.
The parameter set is used only to generate the fixed core; all subsequent synthesis experiments hold \(K_{\rm sc}\) fixed and optimize only the five local layers.
The full matrix representation and consistency checks are reported in Appendix~\ref{app:sc-generators}.
For this representative choice, the determinant-normalized core has operator-Schmidt rank \(9\), so it is not reducible to a low-rank bipartite operator class.
This nonlocality check is necessary but not sufficient: the decisive expressivity test remains the differential rank of the four-core map.
One structural choice in this parameter set deserves emphasis.
Although \(H_{\rm sc}(t)\) is real symmetric at every instant, its static and driven components satisfy \([H_0,H_1]\neq0\); consequently, \(H_{\rm sc}(t)\) and \(H_{\rm sc}(t')\) fail to commute whenever the time-dependent drive coefficients differ. 
The additional temporal asymmetry removes the palindromic time-reversal structure that could otherwise preserve a symmetric total propagator.
The pulse asymmetry is introduced specifically to break the complex-symmetry constraint of the constant real-Hamiltonian construction: the time-ordered propagator is not constrained to satisfy \(K_{\rm sc}=K_{\rm sc}^{T}\).
For the parameter set in Table~\ref{tab:sc-representative-parameters}, the direct numerical check gives \(\lVert K_{\rm sc}-K_{\rm sc}^{T}\rVert_{F}=0.4699801\), confirming that the generated core is not complex symmetric.
The parameter values in Table~\ref{tab:sc-representative-parameters} should therefore be viewed as a representative fixed-core construction rather than as a calibrated device pulse.
They define a single nonlocal gate \(K_{\rm sc}\), and this same gate is reused in every decomposition attempt.
The local layers before, between, and after the four copies of \(K_{\rm sc}\) absorb all target dependence.
This separation is what makes the test comparable to the Clifford-word benchmark: the nonlocal resource is fixed once and for all, and only local coordinates are adjusted during synthesis.

\begin{table}[!htbp]
    \centering
    \caption{Representative time-dependent superconducting-Hamiltonian parameter set used to generate \(K_{\rm sc}\). All entries except \(\Omega_1(t)\) are constant over the \(70~\mathrm{ns}\) core interval. The lower panel shows the four drive coefficients.}
    \label{tab:sc-representative-parameters}
    \begin{ruledtabular}
    \begin{tabular}{lc}
        Parameter group & Values \\
        \hline
        Drift \((\Delta_1,\Delta_2,\Delta_3,\Delta_4)\) & \((8,8,8,8)\) MHz \\
        Drive \((\Omega_1,\Omega_2,\Omega_3,\Omega_4)\) & \((\Omega_1(t),0,2,2)\) MHz \\
        Coupling \((g_1,g_2,g_3,g_4)\) & \((0,4,0,0)\) MHz \\
        First-drive peak & \(4\) MHz \\
        Evolution time \(T\) & \(0.07~\mu{\rm s}\) \\
    \end{tabular}
    \end{ruledtabular}
    \vspace{0.6em}

    \includegraphics[width=\columnwidth]{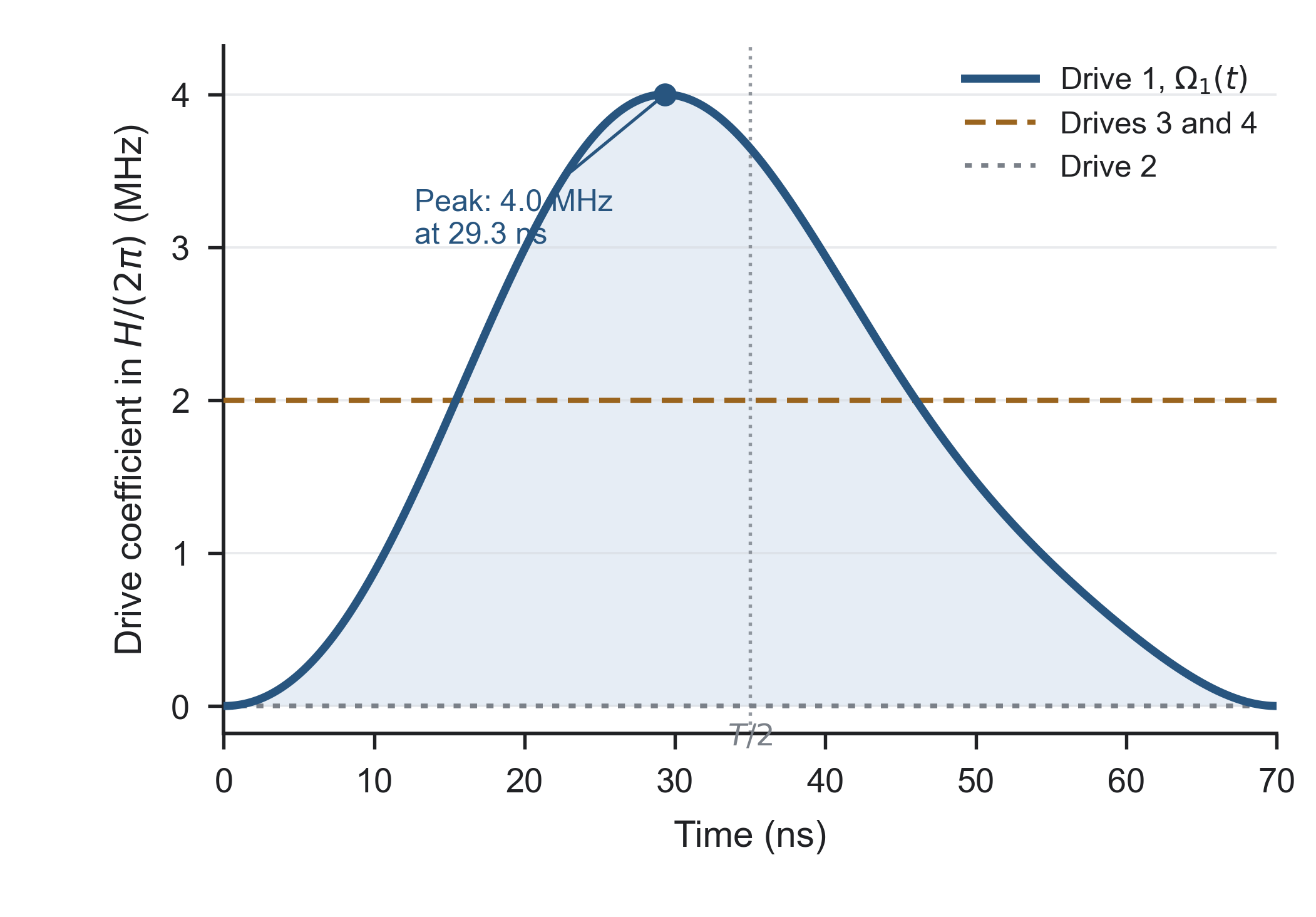}
\end{table}

\subsection{Jacobian-rank and Haar-random synthesis diagnostics}
\label{subsec:sc-diagnostics-haar}

At the identity local point, the directly computed Jacobian has numerical rank \(80\). 
Its smallest singular value is \(\sigma_{\min}(M_{K_{\rm sc}}(\boldsymbol I))=1.41\times10^{-3}\).
Across \(230{,}400\) sampled local configurations, no sampled rank collapse is observed in the tested data set; the smallest recorded value of \(\sigma_{\min}(A_{\rm pt})\) is \(1.36\times10^{-8}\).
These are finite-precision pointwise diagnostics and not proofs of a uniform analytic lower bound over the full local-layer manifold. In particular, the identity-local-point result is a direct numerical calculation for this nonsymmetric core, not a consequence of Proposition~\ref{prop:symmetric-core-identity-critical}.
For \(1000\) Haar-random targets, the same fixed core reaches \(F_{\rm avg}\ge0.999\) in all cases under the stated restart budget.
Detailed restart counts and convergence behavior are reported in Appendix~\ref{app:haar-restarts}, and the pointwise diagnostic distributions in Appendix~\ref{app:sc-diagnostics}.
Figure~\ref{fig:ksc-robustness-protocols} compares perturbed-core recompilation with frozen-compilation sensitivity.
The restart budget is part of the numerical protocol.
Each restart changes only the initialization of the local-layer optimization; it does not change the Hamiltonian parameters or the fixed core.
Thus the \(1000/1000\) result should be read as fixed-core synthesis success under the stated nonconvex optimization protocol, not as target-dependent pulse redesign.

\begin{figure*}[!t]
    \centering
    \includegraphics[width=\textwidth]{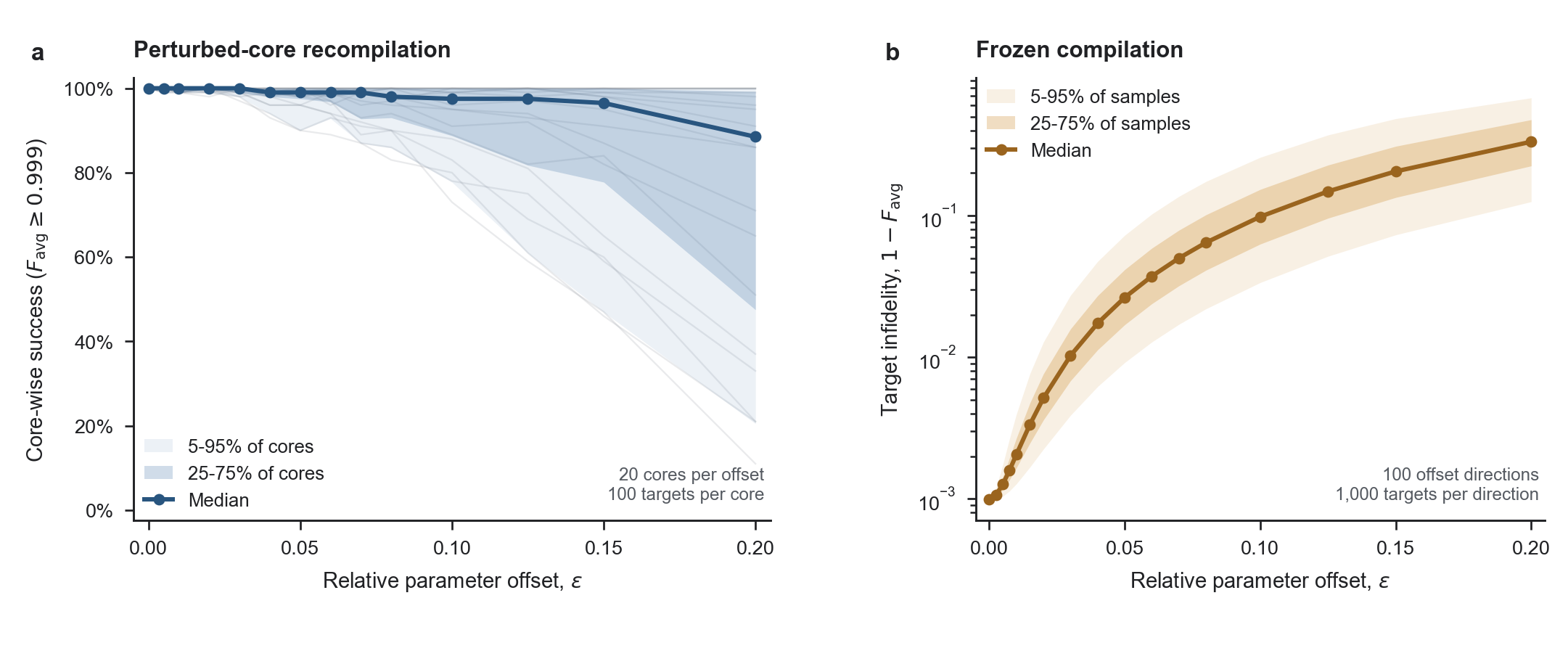}
    \caption{Hamiltonian-parameter robustness under shared static coefficient offsets. (a) Perturbed-core recompilation. The nonzero Hamiltonian coefficients are varied multiplicatively along \(20\) fixed offset directions, while nominally zero coefficients are held at zero. The complete dense scan reuses these directions at each of \(14\) sampled amplitudes through \(\epsilon=0.20\). Each perturbed core is evaluated on the same \(100\)-target subset, and success means \(F_{\rm avg}\ge0.999\). The line is the median core-wise success fraction, the dark band is the interquartile range, the light band is the 5th--95th percentile range, and faint lines show individual core directions. (b) Frozen-compilation target infidelity under the same perturbation rule. The nominal local layers remain fixed while the same perturbed core is used in all four core positions. The line is the pooled median across \(100\) offset directions and \(1000\) fixed targets per direction; the bands show the pooled 25th--75th and 5th--95th sample percentiles. No tolerance threshold or analytic robustness bound is inferred.}
    \label{fig:ksc-robustness-protocols}
\end{figure*}

As in the Clifford benchmark, the final-fidelity histogram is shaped by the stopping rule: once \(F_{\rm avg}\ge0.999\) is reached, the optimizer terminates.
The histogram therefore reports threshold-reaching behavior rather than the maximum fidelities attainable by the ansatz.
As in the Clifford case, the pointwise Jacobian-rank and Haar-random synthesis diagnostics ask different questions: the former samples local Jacobians, whereas the latter tests whether optimization finds target preimages.
For \(K_{\rm sc}\), both tests are favorable in the reported data.
The effective-volume distribution, reported in Appendix~\ref{app:sc-diagnostics}, shows that the map is anisotropic: some tangent directions are opened much more weakly than others.
This anisotropy is not surprising for a Hamiltonian-generated core with a sparse physical structure.
The key observation is that this anisotropy does not prevent the fixed core from succeeding across the reported Haar-random target ensemble under the allowed restart budget.

The \(1000/1000\) Haar-random result is an average-case synthesis benchmark, not a statement about every structured gate.
The target ensemble is sampled from Haar measure on \(SU(9)\), so it almost surely avoids any prescribed lower-dimensional family defined by special algebraic constraints.
A fixed core that consistently fails on Haar-random targets would be a poor candidate for compiling such targets, even if it solved selected structured gates.
Conversely, success on Haar-random targets does not guarantee success on gates with special exchange or permutation structure.
This is why the structured-target table is reported separately rather than folded into the Haar-random benchmark.

\subsection{Robustness and structured targets}
\label{subsec:sc-structured-targets}

The robustness study separates two operational questions, as shown in Fig.~\ref{fig:ksc-robustness-protocols} and detailed in Appendix~\ref{app:robustness}; like all numerical results in this paper, the outcomes are finite-precision sample statistics.

In both protocols, the twelve Hamiltonian coefficients are perturbed before core generation, while the normalized asymmetric envelope of the first drive is retained, and the resulting core \(K_\epsilon\) is used identically in all four core positions.
Each sampled perturbation direction is reused across all values of \(\epsilon\).
In the common cross-protocol comparison, every nonzero coefficient is varied multiplicatively, whereas coefficients that vanish in the nominal model are held fixed at zero.
Thus varying \(\epsilon\) changes only the magnitude of a fixed perturbation direction, representing a shared calibration offset held fixed throughout the pulse and across all four core uses rather than independently resampled noise during a circuit.

Protocol A regenerates the core and reoptimizes all five local layers for every target.
For \(20\) perturbed-core directions and a fixed subset of \(100\) Haar-random targets per core, the median core-wise success fractions are \(0.990\), \(0.975\), \(0.965\), and \(0.885\) at \(\epsilon=0.05\), \(0.10\), \(0.15\), and \(0.20\), respectively, with success defined by \(F_{\rm avg}\ge0.999\).
At \(\epsilon=0.20\), the mean core-wise success fraction is \(0.7305\), the interquartile range is \(0.475\)--\(0.9925\), and the sampled range is \(0.11\)--\(1.00\).
The broad direction-to-direction spread shows that recompilability is strongly anisotropic in the sampled parameter neighborhood: some perturbed cores retain high synthesis performance, whereas others do not under the same optimization budget.

Protocol B instead freezes the nominal local layers for all \(1000\) fixed targets and invokes no optimizer after the nominal compilation.
Across \(100\) offset directions, the pooled median target fidelities are \(0.973560\), \(0.901700\), \(0.794632\), and \(0.667396\) at the same four perturbation strengths.
At \(\epsilon=0.20\), the 5th--95th percentile interval is \(0.321210\)--\(0.874557\).
Frozen compilation is therefore substantially more sensitive to shared core-parameter mismatch than the ability to recompile around the perturbed core.
Under the stated protocols, these results neither define a worst-case tolerance nor model time-dependent noise, leakage, or decoherence.

Separately, we tested selected structured targets under a fixed-budget protocol.
All four targets use the same representative \(K_{\rm sc}\), the same optimizer, \(100\) random restarts, and \(5000\) epochs per restart, with success-threshold early stopping disabled for reporting.
Under this fixed-budget protocol, the best values found are listed in Table~\ref{tab:ksc-structured-targets}.
The qutrit CZ and qutrit CNOT targets reach the \(0.999\) threshold in this run, whereas qutrit iSWAP and \(\sqrt{\mathrm{SWAP}}\) remain below it.
These are fixed-budget numerical outcomes rather than no-go statements for the below-threshold targets.
Matrix-level visualizations are given in Appendix~\ref{app:structured-targets}.
For \(a,b\in\{0,1,2\}\) and \(\omega=e^{2\pi i/3}\), the qutrit CZ target is the generalized controlled-phase gate
\[
    \mathrm{qutrit\ CZ}\,|a,b\rangle=\omega^{ab}|a,b\rangle.
\]
The qutrit CNOT target is the controlled modular-addition gate
\[
    \mathrm{qutrit\ CNOT}\,|a,b\rangle=|a,a+b\!\!\!\pmod 3\rangle,
\]
while the qutrit iSWAP target leaves \(|a,a\rangle\) fixed and maps
\[
    |a,b\rangle\longmapsto i|b,a\rangle,\qquad a\ne b .
\]
Together, these targets probe controlled-phase, controlled-addition, and exchange-like structures.
The contrast between the Haar-random and structured-target tests is useful: Haar-random synthesis probes average-case expressivity over targets sampled from Haar measure on \(SU(9)\), whereas the four structured gates probe particular algebraic patterns.
The fact that qutrit CZ and qutrit CNOT cross the threshold while qutrit iSWAP and \(\sqrt{\mathrm{SWAP}}\) do not under the same fixed budget indicates nonuniform target difficulty for this fixed core.
Accordingly, the Haar-random success rate and the Jacobian-rank diagnostics should be complemented by targeted probes of gates that are important for algorithms or hardware calibration.
Future compilers could treat such structured gates with target-adaptive initialization, continuation from easier exchange gates, or a Hamiltonian-core search objective that includes those targets explicitly.
The present work does not pursue those refinements, because the aim is to test whether a single hardware-motivated fixed core can already support broad synthesis at the minimal dimension-saturating depth.

\begin{table}[t]
\caption{Fixed-budget structured-target synthesis with the
representative superconducting core \(K_{\rm sc}\). All targets
are optimized under the same restart and epoch budget. The
reported value is the best \(F_{\rm avg}\) found within that
budget.}
\label{tab:ksc-structured-targets}
\begin{ruledtabular}
\begin{tabular}{lccc}
Target & Best \(F_{\rm avg}\) & Reaches \(0.999\)? & Restarts used \\
\hline
qutrit CZ & \(0.999990445\) & yes & \(100\) \\
qutrit CNOT & \(0.999990511\) & yes & \(100\) \\
qutrit iSWAP & \(0.998129188\) & no & \(100\) \\
\(\sqrt{\mathrm{SWAP}}\) & \(0.998106396\) & no & \(100\)
\end{tabular}
\end{ruledtabular}
\end{table}

\section{Discussion and Outlook}
\label{sec:discussion}

This work shows that the local expressivity and conditioning of a fixed two-qutrit core are governed by the geometry of the transported local tangent spaces.
In the four-core ansatz \(L_5KL_4KL_3KL_2KL_1\), parameter counting is exactly saturated: \(5\dim(SU(3)\otimes SU(3))=\dim SU(9)=80\).
Expressivity is therefore controlled by whether the five transported local tangent spaces span \(\mathfrak{su}(9)\).
The singular-value certificate developed here gives a direct test of this condition.
The study therefore combines an exact structural certificate for a minimal tight architecture with a hardware-motivated numerical test case; their distinct logical statuses are summarized at the end of Sec.~\ref{sec:introduction}.

The Clifford-word core \(K_{\rm Cl}\) gives an exact example.
Its finite symplectic action splits the \(80\) nonzero two-qutrit Pauli labels into five transported local blocks, which proves full rank of \(D\Phi_{K_{\rm Cl}}\) at the identity local point; the differential there is in fact an exact isometry.
This establishes local universality at the minimal dimension-saturating depth.
The finite-precision regular-pair tests reported in the appendices suggest that some observed fold-like critical values are not genuine image walls.
The transpose-involution bound of Appendix~\ref{app:differential-proofs} complements this exact example with an exact obstruction: complex-symmetric cores, which include all cores generated by real time-independent Hamiltonians, can never be certified at the identity local point, since their differential rank there is at most \(78\).

The superconducting Hamiltonian core \(K_{\rm sc}\) addresses the hardware-oriented version of the same question.
Its noncommuting, temporally asymmetric evolution is designed to break the transposition symmetry, and the direct check \(K_{\rm sc}\neq K_{\rm sc}^{\mathsf T}\) confirms that the complex-symmetric rank bound does not apply.
It does not possess an exact Pauli-label splitting certificate, but no sampled rank collapse is observed in the pointwise diagnostics, and it synthesizes all \(1000\) Haar-random targets at \(F_{\rm avg}\ge0.999\) under the nominal-core protocol.
Under shared static Hamiltonian-coefficient offsets, perturbed-core recompilation and frozen execution exhibit distinct behavior: at \(\epsilon=0.20\), the median core-wise recompilation success is \(0.885\), whereas the frozen-compilation median target fidelity is \(0.667396\).
Structured targets show more nonuniform fixed-budget behavior: qutrit CZ and qutrit CNOT reach the threshold in the fixed-budget run, whereas qutrit iSWAP and \(\sqrt{\mathrm{SWAP}}\) remain below it for the present core and optimizer.
These superconducting results should therefore be read as a controlled numerical study of a fixed physically structured core, not as an optimal-control claim for a particular device.
They show that exact Clifford structure is not necessary for strong sampled performance in the four-core architecture, while also demonstrating that target geometry and optimization protocol still matter.
This is the main practical lesson of the comparison.
The Clifford construction gives a clean existence proof for local universality at the minimal depth.
The superconducting construction shows that a much less algebraically rigid core can still pass the same sampled rank and synthesis tests.
The remaining gap between these two cases is precisely where future hardware-aware core design can act: one can optimize Hamiltonian parameters not only for entangling power or gate fidelity to a single target, but for the singular-value geometry of the whole fixed-core synthesis map.
The symmetric-core bound adds a concrete design constraint to this program: a core that retains the transposition symmetry \(K=K^{T}\) cannot be regular at the identity local point. Breaking that symmetry, as the temporally asymmetric drive does here, removes this obstruction but does not by itself guarantee full rank; the resulting core must still be checked independently.

Several issues remain open.
The most important mathematical problem is to prove or disprove global surjectivity for the Clifford-word cores.
On the hardware side, one should optimize Hamiltonian parameters directly against the differential diagnostics, include leakage and decoherence, and develop target-adaptive initialization for difficult structured gates.
The same framework can also be extended to higher-dimensional qudits and multipartite systems.
For larger systems, the same parameter-count logic will identify candidate minimal depths, but the tangent-space certificate and finite-label constructions will become more demanding.
The two-qutrit case studied here provides a compact setting in which exact finite geometry, Lie-algebraic certificates, and hardware-motivated numerical synthesis can be compared within the same fixed-core framework.
The limitations of the present superconducting model should also be kept explicit.
The Hamiltonian is truncated to the qutrit subspace and does not include leakage, decoherence, pulse-shape smoothness, bandwidth constraints, crosstalk, or measurement effects.
The synthesis optimizer is likewise a numerical procedure over local coordinates with a finite restart budget.
Consequently, these numerical conclusions do not provide device-level calibration guarantees.

\begin{acknowledgments}
This work is supported by the National Natural Science Foundation of China (Grant No. 92565111), Quantum Science and Technology-National Science and Technology Major Project (2021ZD0301701), the National Key Research and Development Program of China (2023YFC2205802) and the NSF of Shandong Province (Grant No. ZR2023LZH002).
\end{acknowledgments}

\vskip 2cm

\centerline{\bf Appendices}

\bigskip

\appendix

\section{Explicit five-block Pauli-label splitting}
\label{app:explicit-label-splitting}
\label{app:pauli-splitting}

This appendix makes the finite Pauli-label splitting certificate completely
explicit, using the finite-field Pauli-label conventions recalled in the main
text~\cite{Gottesman1999,HostensDehaeneDeMoor2005,Hirschfeld1998}.  In the main text, the local full-rank result for \(K_{\rm Cl}\)
was reduced to the statement that the five transported local label blocks
\[
    \Lambda_t=S_{K_{\rm Cl}}^t(\Lambda_{\rm loc}),
    \qquad
    t=0,1,2,3,4,
\]
form a disjoint partition of \(V\setminus\{0\}\).  Here we list all five
blocks explicitly and verify the partition at the finite-field level.  Thus
this appendix supplies the concrete combinatorial data behind
the Pauli-label splitting criterion in Sec.~\ref{subsec:pauli-label-splitting} and
Theorem~\ref{thm:short-clifford-word-local-certificate}.

We work over
\[
    V=\mathbb F_3^4,
    \qquad
    v=(a_1,a_2\mid b_1,b_2),
\]
where \(v\) labels the two-qutrit Pauli direction
\(X^{a_1}Z^{b_1}\otimes X^{a_2}Z^{b_2}\), up to the chosen Weyl phase
convention. Define the two one-qutrit label planes
\begin{equation*}
    \begin{aligned}
    V_A &=\{(a,0\mid b,0):a,b\in\mathbb F_3\}, \\
    V_B &=\{(0,a\mid 0,b):a,b\in\mathbb F_3\},
    \end{aligned}
\end{equation*}
The local Pauli-label block is
\[
    \Lambda_{\rm loc}
    =
    (V_A\setminus\{0\})\sqcup (V_B\setminus\{0\}).
\]
For the short Clifford word \(K_{\rm Cl}\), the induced symplectic matrix is
\[
S_{K_{\rm Cl}}
=
\begin{pmatrix}
0 & 0 & 2 & 1\\
1 & 1 & 0 & 0\\
2 & 1 & 2 & 1\\
1 & 1 & 2 & 2
\end{pmatrix}
\in\mathrm{Sp}(4,\mathbb F_3).
\]
We set
\[
    \Lambda_t:=S_{K_{\rm Cl}}^t(\Lambda_{\rm loc}),
    \qquad
    t=0,1,2,3,4.
\]
The five blocks are listed explicitly below.  All entries are understood as
elements of \(\mathbb F_3\).

The local block \(\Lambda_0=\Lambda_{\rm loc}\) is
\begin{equation}
\begingroup
\setlength{\arraycolsep}{4pt}
\renewcommand{\arraystretch}{1.12}
\Lambda_0=
\left\{
\begin{array}{@{}llll@{}}
(0,0\mid 0,1) & (0,0\mid 0,2) & (0,0\mid 1,0) & (0,0\mid 2,0)\\
(0,1\mid 0,0) & (0,1\mid 0,1) & (0,1\mid 0,2) & (0,2\mid 0,0)\\
(0,2\mid 0,1) & (0,2\mid 0,2) & (1,0\mid 0,0) & (1,0\mid 1,0)\\
(1,0\mid 2,0) & (2,0\mid 0,0) & (2,0\mid 1,0) & (2,0\mid 2,0)
\end{array}
\right\}.
\endgroup
\label{eq:explicit-Lambda0}
\end{equation}

The transported blocks are as follows:
\begin{equation}
\begingroup
\setlength{\arraycolsep}{4pt}
\renewcommand{\arraystretch}{1.12}
\Lambda_1=
\left\{
\begin{array}{@{}llll@{}}
(0,1\mid 1,1) & (0,1\mid 2,1) & (0,2\mid 1,2) & (0,2\mid 2,2)\\
(1,0\mid 1,1) & (1,0\mid 1,2) & (1,1\mid 0,2) & (1,1\mid 2,0)\\
(1,2\mid 0,1) & (1,2\mid 2,0) & (2,0\mid 2,1) & (2,0\mid 2,2)\\
(2,1\mid 0,2) & (2,1\mid 1,0) & (2,2\mid 0,1) & (2,2\mid 1,0)
\end{array}
\right\}.
\endgroup
\label{eq:explicit-Lambda1}
\end{equation}

\begin{equation}
\begingroup
\setlength{\arraycolsep}{4pt}
\renewcommand{\arraystretch}{1.12}
\Lambda_2=
\left\{
\begin{array}{@{}llll@{}}
(0,1\mid 1,2) & (0,1\mid 2,2) & (0,2\mid 1,1) & (0,2\mid 2,1)\\
(1,0\mid 2,1) & (1,0\mid 2,2) & (1,1\mid 0,1) & (1,1\mid 1,0)\\
(1,2\mid 0,2) & (1,2\mid 1,0) & (2,0\mid 1,1) & (2,0\mid 1,2)\\
(2,1\mid 0,1) & (2,1\mid 2,0) & (2,2\mid 0,2) & (2,2\mid 2,0)
\end{array}
\right\}.
\endgroup
\label{eq:explicit-Lambda2}
\end{equation}

\begin{equation}
\begingroup
\setlength{\arraycolsep}{4pt}
\renewcommand{\arraystretch}{1.12}
\Lambda_3=
\left\{
\begin{array}{@{}llll@{}}
(0,1\mid 1,0) & (0,1\mid 2,0) & (0,2\mid 1,0) & (0,2\mid 2,0)\\
(1,0\mid 0,1) & (1,0\mid 0,2) & (1,1\mid 1,2) & (1,1\mid 2,1)\\
(1,2\mid 1,1) & (1,2\mid 2,2) & (2,0\mid 0,1) & (2,0\mid 0,2)\\
(2,1\mid 1,1) & (2,1\mid 2,2) & (2,2\mid 1,2) & (2,2\mid 2,1)
\end{array}
\right\}.
\endgroup
\label{eq:explicit-Lambda3}
\end{equation}

\begin{equation}
\begingroup
\setlength{\arraycolsep}{4pt}
\renewcommand{\arraystretch}{1.12}
\Lambda_4=
\left\{
\begin{array}{@{}llll@{}}
(0,0\mid 1,1) & (0,0\mid 1,2) & (0,0\mid 2,1) & (0,0\mid 2,2)\\
(1,1\mid 0,0) & (1,1\mid 1,1) & (1,1\mid 2,2) & (1,2\mid 0,0)\\
(1,2\mid 1,2) & (1,2\mid 2,1) & (2,1\mid 0,0) & (2,1\mid 1,2)\\
(2,1\mid 2,1) & (2,2\mid 0,0) & (2,2\mid 1,1) & (2,2\mid 2,2)
\end{array}
\right\}.
\endgroup
\label{eq:explicit-Lambda4}
\end{equation}

\begin{lemma}[Explicit five-block splitting]
\label{lem:explicit-five-block-splitting}
The five transported local label blocks
\[
    \Lambda_t=S_{K_{\rm Cl}}^t(\Lambda_{\rm loc}),
    \qquad
    t=0,1,2,3,4,
\]
satisfy
\[
    V\setminus\{0\}
    =
    \Lambda_0\sqcup\Lambda_1\sqcup\Lambda_2\sqcup\Lambda_3\sqcup\Lambda_4.
\]
\end{lemma}

\begin{proof}
Equations~\eqref{eq:explicit-Lambda0}--\eqref{eq:explicit-Lambda4}
list the five blocks explicitly.  Each block contains \(16\) labels.  This
also follows from
\[
    |\Lambda_t|
    =
    |S_{K_{\rm Cl}}^t(\Lambda_{\rm loc})|
    =
    |\Lambda_{\rm loc}|
    =
    16,
\]
because \(S_{K_{\rm Cl}}^t\) is an invertible linear transformation of
\(V=\mathbb F_3^4\).

A direct finite-field check of the displayed lists shows that no nonzero label
appears in two different blocks.  Hence
\(\Lambda_0,\Lambda_1,\Lambda_2,\Lambda_3,\Lambda_4\) are pairwise disjoint.
Finally,
\[
    |V\setminus\{0\}|=3^4-1=80,
    \qquad
    \sum_{t=0}^{4}|\Lambda_t|=5\cdot16=80.
\]
Therefore the pairwise-disjoint union of the five blocks has the same
cardinality as \(V\setminus\{0\}\), and hence
\[
    V\setminus\{0\}
    =
    \Lambda_0\sqcup\Lambda_1\sqcup\Lambda_2\sqcup\Lambda_3\sqcup\Lambda_4.
\]
\end{proof}

The splitting immediately upgrades the full-rank statement at the identity
local point to an isometry statement.

\begin{corollary}[Identity-point isometry]
\label{cor:identity-point-isometry}
Let \(\mathfrak s_t\subset\mathfrak{su}(9)\) denote the real span of the
sign-pair Pauli directions with labels in \(\Lambda_t\).  The five subspaces
\(\mathfrak s_0,\ldots,\mathfrak s_4\) are mutually orthogonal with respect to
the Hilbert--Schmidt inner product, and
\(\operatorname{Ad}_{K_{\rm Cl}^t}\) maps \(\mathfrak l=\mathfrak s_0\)
isometrically onto \(\mathfrak s_t\).  Consequently
\(A_{K_{\rm Cl}}(\boldsymbol I):\mathfrak l^5\to\mathfrak{su}(9)\) is an
isometry: every singular value of \(M_{K_{\rm Cl}}(\boldsymbol I)\) equals
\(1\), and in particular
\(\sigma_{\min}(M_{K_{\rm Cl}}(\boldsymbol I))=1\) exactly.
\end{corollary}

\begin{proof}
Weyl--Pauli directions with distinct label sign pairs are orthogonal in the
Hilbert--Schmidt inner product, and by
Lemma~\ref{lem:explicit-five-block-splitting} the five blocks contain disjoint
label sets, so the subspaces \(\mathfrak s_t\) are mutually orthogonal.
Conjugation by a unitary preserves the Hilbert--Schmidt inner product, and by
the label splitting
\(\operatorname{Ad}_{K_{\rm Cl}^t}(\mathfrak s_0)=\mathfrak s_t\).
For \((X_1,\ldots,X_5)\in\mathfrak l^5\), the five terms of
\(A_{K_{\rm Cl}}(\boldsymbol I)(X_1,\ldots,X_5)
=\sum_{i}\operatorname{Ad}_{K_{\rm Cl}^{5-i}}(X_i)\)
therefore lie in mutually orthogonal subspaces, whence
\(\|A_{K_{\rm Cl}}(\boldsymbol I)(X_1,\ldots,X_5)\|^2=\sum_i\|X_i\|^2\).
Thus \(A_{K_{\rm Cl}}(\boldsymbol I)\) is an isometry between
\(80\)-dimensional inner-product spaces, and all its singular values equal
\(1\).
\end{proof}

We have confirmed numerically that all \(80\) computed singular values of
\(M_{K_{\rm Cl}}(\boldsymbol I)\) equal \(1\) to machine precision.

\section{Finite Clifford-word families and symmetry structure}
\label{app:clifford-family-symmetries}
\label{app:clifford-family}

This appendix provides the finite Clifford-word classification behind the
short-word representative used in the main text.  The finite-geometric input is
the symplectic geometry of \(\mathbb F_3^4\) and the associated
nonsingular-pair spreads studied in Ref.~\cite{HoffmanWeintraubSymplecticF3};
we also use standard finite-projective-geometry terminology~\cite{Hirschfeld1998}.  Appendix~\ref{app:explicit-label-splitting}
verified one explicit five-block splitting for \(K_{\rm Cl}\).  Here we ask
the complementary finite-geometric questions: how many Clifford symplectic
actions realize such a splitting, how these actions are organized, and
whether they lead to genuinely different fixed-core ansatz images.

We keep the Pauli-label notation from
Appendix~\ref{app:explicit-label-splitting}: \(V=\mathbb F_3^4\) with the
standard symplectic form, the local symplectic planes \(V_A,V_B\), and
\[
    \Delta_0:=\{V_A,V_B\},
    \qquad
    B_0:=\Lambda_{\rm loc}
    =(V_A\setminus\{0\})\sqcup(V_B\setminus\{0\}).
\]
For \(S\in\Sp(4,\mathbb F_3)\), write
\[
    \Lambda_t(S):=S^tB_0,
    \qquad
    t=0,1,2,3,4.
\]
We call \(S\) good if
\[
    V\setminus\{0\}
    =
    \bigsqcup_{t=0}^{4}\Lambda_t(S),
\]
and denote the good set by
\[
    \mathcal S_{\rm good}
    :=
    \left\{
    S\in\Sp(4,\mathbb F_3):
    V\setminus\{0\}
    =
    \bigsqcup_{t=0}^{4}S^tB_0
    \right\}.
\]
The main results are that \(|\mathcal S_{\rm good}|=2304\), that these
representatives split equally into order-\(5\) and order-\(10\) symplectic
actions, and that all good Clifford cores are two-sided locally equivalent.
Consequently, the many finite Clifford-word representatives do not correspond
to distinct reachable sets; instead, they justify using \(K_{\rm Cl}\) as a
canonical representative.

\paragraph*{Clifford-group conventions.}
We write \(\operatorname{Cliff}_{n,3}\) for the \(n\)-qutrit Clifford group,
namely the normalizer of the \(n\)-qutrit Pauli group in \(U(3^n)\)~\cite{Gottesman1999,HostensDehaeneDeMoor2005}.  After
quotienting out Pauli displacements and scalar phases, the Clifford action on
Pauli labels gives the finite symplectic group \(\Sp(2n,\mathbb F_3)\).  For
one qutrit, this symplectic quotient is generated by the Fourier gate \(F\)
and the phase gate \(P\).  For two qutrits, it is generated by the local gates
\(F_1,P_1,F_2,P_2\) together with the entangling SUM gate.  Including the
Pauli displacements \(X_j,Z_j\) and the central phases gives the full
two-qutrit Clifford group.  Since this paper only uses the induced action on
Pauli labels, the finite calculations below are performed in
\(\Sp(4,\mathbb F_3)\).

\subsection{Blocks disjoint from the local block}

We first classify the \(16\)-element sign-closed blocks disjoint from
\(B_0=\Lambda_{\rm loc}\).  For \(A\in\SL(2,\mathbb F_3)\), define the graph
plane
\[
    \Gamma(A):=\{(x,Ax):x\in V_A\}\subset V_A\oplus V_B.
\]
Since in dimension two \(\SL(2,\mathbb F_3)=\Sp(2,\mathbb F_3)\), one has
\[
    \Gamma(A)^\perp=\Gamma(-A).
\]
Hence
\[
    B_{[A]}:=
    \bigl(\Gamma(A)\setminus\{0\}\bigr)
    \sqcup
    \bigl(\Gamma(-A)\setminus\{0\}\bigr)
\]
is a \(16\)-element block disjoint from \(B_0\).  Here
\([A]\) denotes the class of \(A\) in
\[
    \PSL(2,\mathbb F_3)=\SL(2,\mathbb F_3)/\{\pm I\}.
\]
The block \(B_{[A]}\) is independent of the representative \(A\), since
\(\Gamma(A)\) and \(\Gamma(-A)\) are already paired.

\begin{lemma}[Classification of blocks disjoint from the local block]
\label{lem:blocks-disjoint-from-local}
The \(16\)-element blocks disjoint from \(B_0\) are precisely
\[
    \{B_{[A]}:[A]\in\PSL(2,\mathbb F_3)\}.
\]
In particular, there are
\[
    |\PSL(2,\mathbb F_3)|=12
\]
such blocks.
\end{lemma}

\begin{proof}
Let \(B\) be a sign-closed \(16\)-element block disjoint from \(B_0\).  Then
\[
    B=(W\setminus\{0\})\sqcup(W^\perp\setminus\{0\})
\]
for some two-dimensional nonsingular symplectic plane \(W\subset V\).
Since \(B\cap(V_A\setminus\{0\})=\varnothing\) and
\(B\cap(V_B\setminus\{0\})=\varnothing\), the projections
\[
    W\to V_A,
    \qquad
    W\to V_B
\]
are injective, hence isomorphisms.  Therefore \(W=\Gamma(A)\) for an
invertible map \(A:V_A\to V_B\).

For vectors \((x,Ax),(y,Ay)\in\Gamma(A)\), the restricted symplectic form is
\[
    \omega((x,Ax),(y,Ay))
    =
    \omega_A(x,y)+\omega_B(Ax,Ay).
\]
In dimension two,
\[
    \omega_B(Ax,Ay)=\det(A)\,\omega_A(x,y).
\]
Thus the restriction to \(\Gamma(A)\) is
\[
    (1+\det A)\omega_A(x,y).
\]
Since \(W\) is nonsingular, \(1+\det A\neq0\) in \(\mathbb F_3\).  As
\(\det A\in\mathbb F_3^\times=\{1,2\}\), this forces \(\det A=1\).  Hence
\(A\in\SL(2,\mathbb F_3)\).  Finally, for such \(A\),
\(\Gamma(A)^\perp=\Gamma(-A)\), so \(B=B_{[A]}\).  Conversely, every
\(A\in\SL(2,\mathbb F_3)\) gives such a block.  Passing from \(A\) to
\(-A\) does not change the block, so the parameter set is
\(\PSL(2,\mathbb F_3)\).
\end{proof}

\subsection{Compatibility graph and the three five-block splittings}

We next determine when two nonlocal blocks \(B_{[A]}\) and \(B_{[B]}\) are
disjoint.

\begin{lemma}[Block-disjointness criterion]
\label{lem:block-disjointness-criterion}
For \(A,B\in\SL(2,\mathbb F_3)\), the following are equivalent:
\[
    B_{[A]}\cap B_{[B]}=\varnothing;
\]
\[
    A-B \ \text{and}\ A+B \ \text{are both invertible};
\]
\[
    [A]^{-1}[B]\in\PSL(2,\mathbb F_3)
    \ \text{is a nontrivial element of order two}.
\]
\end{lemma}

\begin{proof}
By definition,
\[
    B_{[A]}
    =
    (\Gamma(A)\setminus\{0\})\sqcup(\Gamma(-A)\setminus\{0\}),
\]
and similarly for \(B_{[B]}\).  Thus
\(B_{[A]}\cap B_{[B]}=\varnothing\) is equivalent to
\[
    \Gamma(\varepsilon A)\cap\Gamma(\delta B)=\{0\},
    \qquad
    \varepsilon,\delta\in\{\pm1\}.
\]
This holds if and only if both \(A-B\) and \(A+B\) are invertible.

Set \(g=A^{-1}B\in\SL(2,\mathbb F_3)\).  Then
\[
    A-B=A(I-g),
    \qquad
    A+B=A(I+g).
\]
Hence \(A-B\) and \(A+B\) are invertible if and only if \(g\) has neither
\(1\) nor \(-1\) as an eigenvalue.  For \(g\in\SL(2,\mathbb F_3)\), this is
equivalent to \(\tr g=0\).  By Cayley--Hamilton,
\[
    g^2-(\tr g)g+I=0,
\]
so \(\tr g=0\) is equivalent to
\[
    g^2=-I.
\]
Therefore \([g]\in\PSL(2,\mathbb F_3)\) is a nontrivial element of order two.
This proves the criterion.
\end{proof}

The group \(\PSL(2,\mathbb F_3)\) is isomorphic to \(A_4\).  Under this
isomorphism, the three nontrivial order-two elements form the normal Klein
four subgroup
\[
    V_4=\{e,\alpha,\beta,\gamma\}\triangleleft A_4.
\]

\begin{proposition}[Compatibility graph]
\label{prop:compatibility-graph}
Let \(\mathscr G\) be the graph whose vertices are the \(12\) blocks
\(B_{[A]}\), with an edge between two vertices when the corresponding blocks
are disjoint.  Then
\[
    \mathscr G
    \cong
    \operatorname{Cay}(A_4,V_4\setminus\{e\})
    =
    K_4\sqcup K_4\sqcup K_4 .
\]
Consequently, there are exactly three five-block splittings containing
\(B_0=\Lambda_{\rm loc}\).
\end{proposition}

\begin{proof}
By Lemma~\ref{lem:block-disjointness-criterion},
\[
    B_{[A]}\cap B_{[B]}=\varnothing
    \quad\Longleftrightarrow\quad
    [A]^{-1}[B]\in V_4\setminus\{e\}.
\]
Thus \(\mathscr G\) is the Cayley graph of \(A_4\) with generating set
\(V_4\setminus\{e\}\).  Since \(V_4\triangleleft A_4\), the left cosets of
\(V_4\) form three four-element subsets.  Within each coset any two distinct
vertices are adjacent, and vertices in different cosets are not adjacent.
Hence
\[
    \mathscr G=K_4\sqcup K_4\sqcup K_4.
\]

A five-block splitting containing \(B_0\) consists of \(B_0\) and four
pairwise-disjoint nonlocal blocks.  Such a quadruple is exactly one of the
three \(K_4\) components.  Conversely, each \(K_4\) component gives four
pairwise-disjoint nonlocal blocks; together with \(B_0\), their total size is
\[
    16+4\cdot16=80=|V\setminus\{0\}|,
\]
so they form a five-block splitting.
\end{proof}

Thus the \(2304\) good Clifford-label actions do not produce \(2304\)
different unordered five-block systems.  They produce only three unordered
systems; the remaining multiplicity comes from the possible five-cycle
actions and their stabilizer fibers.

\subsection{Counting good symplectic matrices}

Fix one of the three five-block systems and write it as
\[
    \Sigma=\{\Lambda_0,\Lambda_1,\Lambda_2,\Lambda_3,\Lambda_4\},
    \qquad
    \Lambda_0=B_0.
\]
Let
\[
    \operatorname{Stab}(\Sigma):=\{M\in\Sp(4,\mathbb F_3):M(\Sigma)=\Sigma\}
\]
be its setwise stabilizer.  There is a natural permutation representation
\[
    \rho:\operatorname{Stab}(\Sigma)\longrightarrow S_5.
\]
The finite-geometry input used here is the stabilizer theorem for
nonsingular-pair spreads in \(\operatorname{PSp}(4,3)\): projectively,
the stabilizer acts as \(A_5\) on the five nonsingular pairs and has an
elementary \(2\)-group kernel of order \(16\)~\cite[Lemma~2.12]{HoffmanWeintraubSymplecticF3}.  Pulling back to
\(\Sp(4,\mathbb F_3)\) gives the following exact sequence.

\begin{proposition}[Five-block stabilizer]
\label{prop:five-block-stabilizer}
For every five-block system \(\Sigma\) above, there is a short exact sequence
\[
    1
    \longrightarrow
    E
    \longrightarrow
    \operatorname{Stab}(\Sigma)
    \xrightarrow{\ \rho\ }
    A_5
    \longrightarrow
    1,
    \qquad
    |E|=32.
\]
Equivalently,
\[
    \operatorname{Stab}(\Sigma)/E\cong A_5 .
\]
\end{proposition}

\begin{proof}
Let
\[
    \pi:\Sp(4,\mathbb F_3)\longrightarrow\operatorname{PSp}(4,\mathbb F_3)
\]
be the quotient by the central subgroup \(\{\pm I\}\).  Projectively, by the nsp-spread stabilizer sequence of
Ref.~\cite[Lemma~2.12]{HoffmanWeintraubSymplecticF3}, the five-block stabilizer has exact sequence
\[
    1
    \longrightarrow
    E_{\rm proj}
    \longrightarrow
    K_{\rm proj}(\Sigma)
    \longrightarrow
    A_5
    \longrightarrow
    1,
    \qquad
    |E_{\rm proj}|=16.
\]
The full stabilizer in \(\Sp(4,\mathbb F_3)\) is the preimage
\[
    \operatorname{Stab}(\Sigma)=\pi^{-1}(K_{\rm proj}(\Sigma)).
\]
The kernel \(E\) of the induced map \(\operatorname{Stab}(\Sigma)\to A_5\) fits into
\[
    1
    \longrightarrow
    \{\pm I\}
    \longrightarrow
    E
    \longrightarrow
    E_{\rm proj}
    \longrightarrow
    1.
\]
Therefore \(|E|=2|E_{\rm proj}|=32\), and exactness follows by pullback.
\end{proof}

\begin{lemma}[Good elements are lifts of five-cycles]
\label{lem:good-iff-five-cycle}
For a fixed five-block system \(\Sigma\), the good elements \(S\in \operatorname{Stab}(\Sigma)\)
with \(\Sigma(S)=\Sigma\) are exactly
\[
    \rho^{-1}(\mathcal C_5),
\]
where \(\mathcal C_5\subset A_5\) is the set of five-cycles.
\end{lemma}

\begin{proof}
If \(S\) is good and \(\Sigma(S)=\Sigma\), then the sequence
\[
    B_0,\ SB_0,\ S^2B_0,\ S^3B_0,\ S^4B_0
\]
runs through all five blocks in \(\Sigma\).  Hence \(\rho(S)\) is a
five-cycle.  Conversely, if \(\rho(S)\) is a five-cycle, then the five blocks
\(S^tB_0\), \(t=0,\ldots,4\), are all distinct and exhaust \(\Sigma\).  Since
\(\Sigma\) is a five-block splitting of \(V\setminus\{0\}\), \(S\) is good.
\end{proof}

Since \(A_5\) contains \(24\) five-cycles, each fixed five-block system
contributes
\[
    24\cdot |E|=24\cdot32=768
\]
good symplectic matrices.

\begin{theorem}[Good Clifford-word symplectic count]
\label{thm:good-clifford-word-count-app}
The good set satisfies
\[
    |\mathcal S_{\rm good}|=2304.
\]
Moreover, exactly \(1152\) elements have order \(5\), and exactly \(1152\)
elements have order \(10\).
\end{theorem}

\begin{proof}
By Proposition~\ref{prop:compatibility-graph}, there are exactly three
five-block systems containing \(B_0\).  Each contributes \(768\) good elements.
A good element determines its five-block system uniquely:
\[
    \Sigma(S)
    =
    \{B_0,SB_0,S^2B_0,S^3B_0,S^4B_0\}.
\]
Hence the three contributions are disjoint, and
\[
    |\mathcal S_{\rm good}|
    =
    3\cdot768
    =
    3\cdot24\cdot32
    =
    2304.
\]

It remains to explain the order split.  Let \(S\in\mathcal S_{\rm good}\),
and let \(\bar S\) be its projective class in
\(\operatorname{PSp}(4,\mathbb F_3)\).  Since \(\rho(S)\) is a five-cycle,
\(S^5\in E\).  Projectively, \(\bar S^5\in E_{\rm proj}\).  The conjugation
action of a five-cycle on \(E_{\rm proj}\cong(\mathbb Z/2)^4\) is the
standard four-dimensional \(A_5\)-module over \(\mathbb F_2\), and a
five-cycle has no nonzero fixed vector in this module.  Since \(\bar S^5\)
commutes with \(\bar S\), it is fixed by this action.  Hence
\[
    \bar S^5=1.
\]
Therefore
\[
    S^5\in\{\pm I\},
\]
so every good element has order \(5\) or \(10\).

Finally, multiplication by \(-I\) preserves \(\mathcal S_{\rm good}\) and
pairs elements without fixed points:
\[
    S\longmapsto -S.
\]
If \(S^5=I\), then \((-S)^5=-I\), so \(-S\) has order \(10\).  Conversely,
if \(S^5=-I\), then \((-S)^5=I\), so \(-S\) has order \(5\).  Hence the
\(2304\) good elements split equally between order \(5\) and order \(10\).
\end{proof}

\subsection{Two-sided local equivalence of all good cores}
\label{app:Sgood-local-double-coset}

We now prove that the \(2304\) good symplectic matrices lie in a single
two-sided local double coset.  Let
\[
    H_{\rm loc}^0
    :=
    \SL(V_A)\times \SL(V_B)
    \subset \Sp(V)
\]
be the strict local symplectic subgroup preserving \(V_A\) and \(V_B\)
separately.  In gate language, this is the symplectic image of the local
single-qutrit Clifford group
\[
    \operatorname{Cliff}_{1,3}\otimes\operatorname{Cliff}_{1,3} .
\]
If one also allows the swap of the two tensor factors, one obtains the larger
normalizer
\[
    H_{\rm loc}:=H_{\rm loc}^0\rtimes C_2,
\]
but the argument below already works with \(H_{\rm loc}^0\).

\begin{lemma}[Local double-coset lemma]
\label{lem:local-double-coset}
Let
\[
    \mathcal D
    :=
    \{S\in\Sp(V):S\Delta_0\cap\Delta_0=\varnothing\}.
\]
Then \(\mathcal D\) is a single double coset under \(H_{\rm loc}^0\).  That is,
for any \(S,T\in\mathcal D\), there exist
\[
    h_L,h_R\in H_{\rm loc}^0
\]
such that
\[
    T=h_LSh_R .
\]
\end{lemma}

\begin{proof}
We identify \(V=V_A\oplus V_B\) and write a vector as \((x,y)\), with
\(x\in V_A\) and \(y\in V_B\).  Let \(\Delta=\{\delta,\delta^\perp\}\) be a nonsingular pair disjoint from \(\Delta_0\).  Since
\(\delta\cap V_B=0\), the projection \(\delta\to V_A\) is an isomorphism.
Hence \(\delta\) is the graph of an invertible map
\[
    A:V_A\to V_B,
    \qquad
    \delta=\Gamma(A).
\]
By the same determinant computation as in
Lemma~\ref{lem:blocks-disjoint-from-local}, nonsingularity forces
\(A\in\SL(2,\mathbb F_3)\), and \(\delta^\perp=\Gamma(-A)\).  Thus every nonsingular pair disjoint from \(\Delta_0\) can be written as
\[
    \Delta_A:=\{\Gamma(A),\Gamma(-A)\},
    \qquad
    A\in\SL(2,\mathbb F_3),
\]
where \(A\) and \(-A\) define the same unordered pair.

The group \(H_{\rm loc}^0=\SL(V_A)\times\SL(V_B)\) acts on graphs by
\[
    (g_A,g_B)\cdot\Gamma(A)=\Gamma(g_BAg_A^{-1}).
\]
Since left and right multiplication act transitively on
\(\SL(2,\mathbb F_3)\), the local group \(H_{\rm loc}^0\) acts transitively
on the nonsingular pairs disjoint from \(\Delta_0\).  Moreover, by choosing
the sign of \(A\), we may match the ordered components \(\Gamma(A),\Gamma(-A)\) with any prescribed ordered pair of complementary components.

Now take \(S,T\in\mathcal D\).  Then \(S\Delta_0\) and \(T\Delta_0\) are both
nonsingular pairs disjoint from \(\Delta_0\).  By the transitivity just proved,
there exists \(h_L\in H_{\rm loc}^0\) such that
\[
    h_LS(V_A)=T(V_A),
    \qquad
    h_LS(V_B)=T(V_B).
\]
Therefore
\[
    h_R:=S^{-1}h_L^{-1}T
\]
preserves \(V_A\) and \(V_B\) separately, so \(h_R\in H_{\rm loc}^0\).  Hence \(T=h_LSh_R\).
\end{proof}

\begin{theorem}[All good symplectic cores are two-sided locally equivalent]
\label{thm:Sgood-one-double-coset}
For any two good cores
\[
    S,T\in\mathcal S_{\rm good},
\]
there exist
\[
    h_L,h_R\in H_{\rm loc}^0
\]
such that
\[
    T=h_LSh_R.
\]
Equivalently,
\[
    \mathcal S_{\rm good}\subset H_{\rm loc}^0 S H_{\rm loc}^0
\]
for any fixed \(S\in\mathcal S_{\rm good}\).
\end{theorem}

\begin{proof}
If \(S\in\mathcal S_{\rm good}\), then
\[
    B_0,\ SB_0,\ S^2B_0,\ S^3B_0,\ S^4B_0
\]
are pairwise disjoint.  In particular,
\[
    SB_0\cap B_0=\varnothing.
\]
Passing from label blocks to nonsingular pairs, this says
\[
    S\Delta_0\cap\Delta_0=\varnothing.
\]
Thus \(S\in\mathcal D\).  The same holds for \(T\).  Lemma~\ref{lem:local-double-coset}
therefore gives
\(
    T=h_LSh_R
\)
with \(h_L,h_R\in H_{\rm loc}^0\).
\end{proof}

\subsection{Clifford-level and ansatz-level consequences}

By the Clifford-group convention above, there is a natural projection
\(\pi:\operatorname{Cliff}_{2,3}\longrightarrow \Sp(4,\mathbb F_3)\) 
obtained from the action of Clifford unitaries on Pauli labels, after
quotienting out Pauli displacements and central phases.  For a two-qutrit
Clifford core \(K\), write \(S_K:=\pi(K)\), and define
\[
    \mathcal K_{\rm good}
    :=
    \{K\in\operatorname{Cliff}_{2,3}:S_K\in\mathcal S_{\rm good}\}.
\]

\begin{corollary}[All good Clifford cores are two-sided locally equivalent]
\label{cor:Clifford-local-equivalence}
For any two good Clifford cores
\[
    K,K'\in\mathcal K_{\rm good},
\]
there exist local Clifford unitaries
\[
    A,B\in\operatorname{Cliff}_{1,3}\otimes\operatorname{Cliff}_{1,3}
\]
such that, after choosing compatible determinant-one representatives,
\[
    K'=AKB.
\]
\end{corollary}

\begin{proof}
By Theorem~\ref{thm:Sgood-one-double-coset}, there exist
\(h_L,h_R\in H_{\rm loc}^0\) such that \(S_{K'}=h_LS_Kh_R\).
Choose local Clifford lifts \(A_0,B_0\) satisfying \(\pi(A_0)=h_L\) and
\(\pi(B_0)=h_R\).  Then \(\pi(A_0KB_0)=S_{K'}\), so \(K'\) and
\(A_0KB_0\) have the same symplectic action.  Two Clifford
unitaries with the same symplectic action differ by a Pauli operator and a
central phase:
\[
    K'=e^{i\theta}Q\,A_0KB_0 .
\]
Every two-qutrit Pauli operator factors as a tensor product of one-qutrit
Paulis, hence
\[
    Q=Q_A\otimes Q_B
    \in
    \operatorname{Cliff}_{1,3}\otimes\operatorname{Cliff}_{1,3}.
\]
Absorbing \(Q\) into the left local Clifford factor gives the desired
two-sided local equivalence, up to the harmless central phase convention of
the Clifford representative.  With compatible determinant-one representatives,
this is written as \(K'=AKB\).
\end{proof}

\begin{proposition}[Two-sided local factors do not change the reachable set]
\label{prop:two-sided-local-equivalence-same-image}
Let \(A,B\in\Loc\), and set
\[
    K'=AKB.
\]
Then
\[
    \operatorname{Im}\Phi_{K'}=\operatorname{Im}\Phi_K.
\]
\end{proposition}

\begin{proof}
For arbitrary \(L_1,\ldots,L_5\in\Loc\), substitute \(K'=AKB\) into
\(\Phi_{K'}(L_1,\ldots,L_5)\).  Define
\[
    \widetilde L_5:=L_5A,\qquad
    \widetilde L_j:=BL_jA\ (j=2,3,4),\qquad
    \widetilde L_1:=BL_1 .
\]
Then all \(\widetilde L_j\in\Loc\), and the substituted expression becomes
\[
    \Phi_{K'}(L_1,\ldots,L_5)
    =
    \Phi_K(\widetilde L_1,\ldots,\widetilde L_5).
\]
Therefore
\[
    \operatorname{Im}\Phi_{K'}\subseteq\operatorname{Im}\Phi_K.
\]
Applying the same argument to \(K=A^{-1}K'B^{-1}\) gives the reverse
inclusion.  Hence \(\operatorname{Im}\Phi_{K'}=\operatorname{Im}\Phi_K.\)
\end{proof}

Combining Corollary~\ref{cor:Clifford-local-equivalence} with
Proposition~\ref{prop:two-sided-local-equivalence-same-image}, all good
Clifford-word cores have the same synthesis capability in the four-core
ansatz, up to a reparameterization of the five local layers.  Thus the short
word \(K_{\rm Cl}\) used in the main text is a canonical representative of
this two-sided local equivalence class.  The count \(2304\) should therefore
not be interpreted as \(2304\) distinct decomposition models.

\begin{proposition}[Transported full-rank points from local equivalence]
\label{prop:transported-full-rank-points}
Let \(K,K'\in SU(9)\) be related by two-sided local factors
\[
    K'=AKB,
    \qquad
    A,B\in\Loc.
\]
Define the reparameterization
\[
    \mathcal R_{A,B}:\Loc^5\to\Loc^5
\]
by
\[
    \mathcal R_{A,B}(L_1,\ldots,L_5)
    =
    (BL_1,\,BL_2A,\,BL_3A,\,BL_4A,\,L_5A).
\]
Then
\[
    \Phi_{K'}=\Phi_K\circ\mathcal R_{A,B}.
\]
In particular, if the differential at the identity local point has full rank for \(K'\), then
\[
    \mathbf L_{A,B}
    :=
    \mathcal R_{A,B}(I,I,I,I,I)
    =
    (B,\,BA,\,BA,\,BA,\,A)
\]
is a full-rank point for the fixed-core map \(\Phi_K\):
\[
    \operatorname{rank}D\Phi_K(\mathbf L_{A,B})=80.
\]
\end{proposition}

\begin{proof}
The identity
\[
    \Phi_{K'}=\Phi_K\circ\mathcal R_{A,B}
\]
is exactly the reparameterization used in the proof of
Proposition~\ref{prop:two-sided-local-equivalence-same-image}.  The map
\(\mathcal R_{A,B}\) is a diffeomorphism of \(\Loc^5\), since it is built from
left and right multiplication by fixed elements of \(\Loc\).  Therefore
\[
    D\Phi_{K'}(\mathbf L)
    =
    D\Phi_K\bigl(\mathcal R_{A,B}(\mathbf L)\bigr)
    \circ
    D\mathcal R_{A,B}(\mathbf L),
\]
and \(D\mathcal R_{A,B}(\mathbf L)\) is an isomorphism.  Hence
\[
    \operatorname{rank}D\Phi_{K'}(\mathbf L)
    =
    \operatorname{rank}D\Phi_K\bigl(\mathcal R_{A,B}(\mathbf L)\bigr).
\]
Taking \(\mathbf L=(I,I,I,I,I)\) proves the claim.
\end{proof}

Consequently, after choosing compatible determinant-one representatives and
fixing \(K=K_{\rm Cl}\), every good Clifford core
\(K_j=A_jK_{\rm Cl}B_j\) supplies an algebraically certified full-rank point
\[
    \mathbf L_j
    =
    (B_j,\,B_jA_j,\,B_jA_j,\,B_jA_j,\,A_j)
    \in\Loc^5
\]
for the single fixed-core map \(\Phi_{K_{\rm Cl}}\).  Thus the local
equivalence classification does more than identify equal reachable sets: it
also produces a finite family of explicit full-rank points for the same
Clifford-word core.  These points are inherited from the Pauli-label
splitting certificates established at the identity local point for the locally equivalent good cores.

\subsection{Inversion symmetry and the local subgroup}
\label{subsec:inversion-symmetry-reachable-set}

We record one useful consequence of the finite Clifford symmetry: the inverse
of the short Clifford core is locally equivalent to the core itself.  At the
symplectic-label level, a direct finite-field computation gives a local
reversing symmetry
\[
    R_\ast\in \Sp(V_A)\times\Sp(V_B)
\]
such that
\[
    R_\ast S_{K_{\rm Cl}}R_\ast^{-1}=S_{K_{\rm Cl}}^{-1}.
\]
For example, in the label order \((a_1,a_2\mid b_1,b_2)\), one may take
\[
    R_\ast=
    \begin{pmatrix}
        0&0&1&0\\
        0&2&0&2\\
        2&0&0&0\\
        0&2&0&1
    \end{pmatrix}.
\]
This matrix lies in the strict local symplectic subgroup and satisfies the
displayed reversing relation.

Choose a local Clifford lift
\[
    U_R=U_A\otimes U_B\in\Loc
\]
of \(R_\ast\), with the one-qutrit factors normalized to determinant one.
Then \(U_RK_{\rm Cl}U_R^{-1}\) and \(K_{\rm Cl}^{-1}\) have
the same induced symplectic action.  Hence they differ only by a Pauli
operator and a central phase,
\[
    U_RK_{\rm Cl}U_R^{-1}
    =
    \zeta\, Q\, K_{\rm Cl}^{-1},
\]
where \(Q\) is a two-qutrit Pauli operator and, by taking determinants,
\(\zeta^9=1\).
Every two-qutrit Pauli operator is a tensor product of one-qutrit Paulis with
unit determinant, so \(Q\in\Loc\).
The scalar requires one further check: the scalars contained in
\(\Loc=SU(3)\otimes SU(3)\) are exactly the cube roots of unity, so the
central phase can be absorbed into a local factor only if \(\zeta^3=1\),
which is not automatic from the determinant argument alone.
For the explicit lift used here, direct computation gives
\(\zeta=e^{2\pi i/3}\), a cube root of unity.
Absorbing \(\zeta\) and \(Q\) into a single local factor
\(P_\ast:=\zeta\,Q\in\Loc\), we obtain
\[
    U_RK_{\rm Cl}U_R^{-1}
    =
    P_\ast K_{\rm Cl}^{-1}
\]
exactly.
Equivalently,
\[
    K_{\rm Cl}^{-1}=AK_{\rm Cl}B,
    \qquad
    A:=P_\ast^{-1}U_R\in\Loc,
    \qquad
    B:=U_R^{-1}\in\Loc.
\]
Therefore \(K_{\rm Cl}^{-1}\) is two-sided locally equivalent to
\(K_{\rm Cl}\).  By
Proposition~\ref{prop:two-sided-local-equivalence-same-image}, this implies
\(
    \operatorname{Im}\Phi_{K_{\rm Cl}^{-1}}
    =
    \operatorname{Im}\Phi_{K_{\rm Cl}}.
\)
This same inversion relation also implies that the whole local subgroup is
already contained in the reachable set of the four-core ansatz.  Indeed, from
\(K_{\rm Cl}^{-1}=AK_{\rm Cl}B\) we obtain
\[
    K_{\rm Cl}AK_{\rm Cl}B=I.
\]
For any \(L_0\in\Loc\), choose
\[
    (L_1,L_2,L_3,L_4,L_5)=(B,A,B,A,L_0).
\]
Then
\[
\begin{aligned}
    \Phi_{K_{\rm Cl}}(L_1,\ldots,L_5)
    &=
    L_0K_{\rm Cl}AK_{\rm Cl}B
       K_{\rm Cl}AK_{\rm Cl}B  \\
    &=
    L_0.
\end{aligned}
\]
Hence
\(
    \Loc\subseteq \operatorname{Im}\Phi_{K_{\rm Cl}}.
\)
Thus the inversion symmetry shows not only that \(K_{\rm Cl}\) and
\(K_{\rm Cl}^{-1}\) define the same reachable set, but also that all local
two-qutrit gates are exactly reachable by the same fixed-core architecture.

\section{Proofs of differential and local-universality statements}
\label{app:differential-proofs}

This appendix records the proofs of the basic differential statements used in
the main text.  Let \(U=\Phi_K(\mathbf L)\) and write
\(X_i=L_i^{-1}\dot L_i\in\mathfrak l\).  Differentiating the product
\(L_5KL_4KL_3KL_2KL_1\) with respect to one local factor at a time gives a
sum of five terms.  Right multiplication by \(U^{-1}\) cancels all factors to
the right of \(L_i\), leaving
\[
    P_i(\mathbf L)(L_i^{-1}\dot L_i)P_i(\mathbf L)^{-1}.
\]
Summing the five independent variations gives
\[
    \tau_U^R\!\left(
    D\Phi_K(\mathbf L)[\dot L_1,\ldots,\dot L_5]
    \right)
    =
    \sum_{i=1}^{5}\operatorname{Ad}_{P_i(\mathbf L)}(X_i),
\]
which is Eq.~\eqref{eq:AK-definition}.  Since right trivialization is a
linear isomorphism on tangent spaces, it preserves rank.

After choosing orthonormal bases of \(\mathfrak l^5\) and \(\mathfrak g\), the
matrix \(M_K(\mathbf L)\) represents this right-trivialized differential.
Thus \(\sigma_{\min}(M_K(\mathbf L_\ast))>0\) is equivalent to
\(\operatorname{rank}D\Phi_K(\mathbf L_\ast)=80\).  The domain \(L^5\) and the
target \(SU(9)\) both have real dimension \(80\), so the inverse function
theorem gives neighborhoods \(\mathcal U\subset L^5\) and
\(\mathcal V\subset SU(9)\) with
\(\Phi_K(\mathcal U)=\mathcal V\).  This proves the local-surjectivity
certificate used in the main text.

For the global discussion, compactness gives the closed-image statement:
\(L\) is compact, hence \(L^5\) is compact, and
\(\operatorname{Im}\Phi_K=\Phi_K(L^5)\) is compact and therefore closed in
the Hausdorff space \(SU(9)\).  If every reachable value has at least one
regular preimage, the inverse function theorem makes every point of
\(\operatorname{Im}\Phi_K\) an interior point of the image.  The image is then
both open and closed.  Since \(SU(9)\) is connected and the image is nonempty,
this proves global surjectivity.  This is the regular-preimage criterion
referenced in the main text.

\paragraph*{Remark on degree theory.}
A nonzero-degree argument does not directly resolve the global question for
the maps considered here.  Since \(SU(9)\) is connected, any \(K\in SU(9)\)
can be joined to the identity by a smooth path \(K_s\).  The maps
\[
\Phi_{K_s}(L_1,\ldots,L_5)
=
L_5K_sL_4K_sL_3K_sL_2K_sL_1
\]
therefore give a homotopy from \(\Phi_K\) to \(\Phi_I\).  But
\[
\Phi_I(L_1,\ldots,L_5)=L_5L_4L_3L_2L_1
\]
has image contained in the local subgroup \(L\), and hence has degree zero as
a map into \(SU(9)\).  Thus \(\deg \Phi_K=0\) for all \(K\).  A proof of
global surjectivity, if true, must use finer information about fibers,
critical values, or regular preimages.

\paragraph*{Complex-symmetric cores are critical at the identity local point.}
The genericity statement of Sec.~\ref{subsec:right-trivialized-differential}
asserts that the identity local point is regular for Haar-almost every core.
The following proposition shows that the exceptional set contains an entire
structurally important core class.

\begin{proposition}[Symmetric-core identity-point bound]
\label{prop:symmetric-core-identity-critical}
Let \(K\in SU(9)\) be complex symmetric, \(K=K^{T}\).  Then
\[
    \operatorname{rank}A_K(\boldsymbol I)\le 78 .
\]
In particular, the identity local point is a critical point of \(\Phi_K\) for
every complex-symmetric core, and the same bound holds for every core of the
form \(UKU^{\dagger}\) with \(U\in\Loc\).
\end{proposition}

\begin{proof}
Consider the real-linear involution
\(\tau(Y):=Y^{T}=-\overline{Y}\) on \(\mathfrak{su}(9)\).
It decomposes \(\mathfrak{su}(9)=E_+\oplus E_-\), where
\(E_+=\{iR: R\ \text{real symmetric traceless}\}\) has real dimension \(44\)
and \(E_-=\{Q: Q\ \text{real antisymmetric}\}\) has real dimension \(36\).
The local subalgebra is \(\tau\)-invariant, \(\tau(\mathfrak l)=\mathfrak l\),
with \(\dim(\mathfrak l\cap E_+)=10\), coming from the
\(i\,(R_3\otimes I)\) and \(i\,(I\otimes R_3')\) directions with \(R_3,R_3'\)
real symmetric traceless.
If \(K=K^{T}\), then \(K^{\dagger}=\overline K\), and for
\(X\in\mathfrak{su}(9)\)
\[
    \tau\!\left(\operatorname{Ad}_K X\right)
    =(KXK^{-1})^{T}
    =(K^{T})^{-1}X^{T}K^{T}
    =\operatorname{Ad}_{K^{-1}}\!\left(\tau X\right).
\]
Hence \(\tau\bigl(\operatorname{Ad}_{K^{t}}\mathfrak l\bigr)
=\operatorname{Ad}_{K^{-t}}\mathfrak l\) for all \(t\).
Let \(W:=\operatorname{im}A_K(\boldsymbol I)
=\sum_{t=0}^{4}\operatorname{Ad}_{K^{t}}(\mathfrak l)\).
Since \(\operatorname{Ad}_{K^{-2}}\) is an isometry,
\(\dim W=\dim W'\) with
\(W':=\operatorname{Ad}_{K^{-2}}W
=\sum_{t=-2}^{2}\operatorname{Ad}_{K^{t}}(\mathfrak l)\),
and by the displayed relation \(W'\) is \(\tau\)-invariant.
Write \(P_\pm=(\mathrm{id}\pm\tau)/2\).
For any subspace \(U\) one has \(P_+(U+\tau U)=P_+(U)\), because
\(P_+(\tau u)=P_+(u)\).  Applying this to
\(U_t:=\operatorname{Ad}_{K^{t}}(\mathfrak l)\) with \(\tau U_t=U_{-t}\)
gives
\[
    P_+(W')
    =P_+(\mathfrak l)+P_+(U_1)+P_+(U_2),
\]
so
\(\dim P_+(W')\le 10+16+16=42<44=\dim E_+\).
Since \(W'\) is \(\tau\)-invariant,
\(\dim W'=\dim P_+(W')+\dim P_-(W')\le 42+36=78\).
Finally, for \(U\in\Loc\) one has
\(\operatorname{Ad}_{(UKU^{\dagger})^{t}}(\mathfrak l)
=\operatorname{Ad}_U\operatorname{Ad}_{K^{t}}
\operatorname{Ad}_{U^{\dagger}}(\mathfrak l)
=\operatorname{Ad}_U\operatorname{Ad}_{K^{t}}(\mathfrak l)\),
so the image for \(UKU^{\dagger}\) is \(\operatorname{Ad}_U(W)\), of the same
dimension.
\end{proof}

The bound uses only \(K=K^{T}\) and the odd number of transported local blocks: the symmetric five-block window \(t\in\{-2,-1,0,1,2\}\) contains a distinguished central block whose \(E_+\) contribution is \(10\) rather than \(16\).
If a time-independent Hamiltonian is real in the chosen basis, then \(K=\exp(-iHT)\) is complex symmetric and the proposition applies.
The same conclusion holds for a time-dependent real Hamiltonian when its values commute at different times, so that the time ordering reduces to a single exponential of a real symmetric matrix. 
It need not hold for noncommuting time-dependent real Hamiltonians, because an ordered product of complex-symmetric propagators is not generally complex symmetric.
The local-conjugation extension applies whenever a core is locally conjugate to a complex-symmetric propagator, for example in a time-independent model whose constant drive phases can be removed by local diagonal frame changes.
The representative superconducting core of Table~\ref{tab:sc-representative-parameters} uses precisely such a noncommuting time-dependent construction. 
Its measured asymmetry \(\lVert K_{\rm sc}-K_{\rm sc}^{T}\rVert_{F}=0.4699801\) confirms that it is outside the proposition's hypothesis. 
A separate finite-precision calculation gives numerical rank \(80\) at the identity local point, with the three smallest singular values approximately \(4.20\times10^{-3}\), \(1.82\times10^{-3}\), and \(1.41\times10^{-3}\). 
These values are numerical diagnostics for this particular core, whereas the rank-\(78\) result above remains an exact theorem for every complex-symmetric core.

\begin{figure*}[!t]
    \centering
    \begin{minipage}{0.48\textwidth}
        \centering
        \includegraphics[width=\linewidth]{figures/fig_rank_diagnostics_clifford_min_singular_value.png}
        \vspace{-0.5em}
        \centerline{\small (a) Smallest singular value}
    \end{minipage}
    \hfill
    \begin{minipage}{0.48\textwidth}
        \centering
        \includegraphics[width=\linewidth]{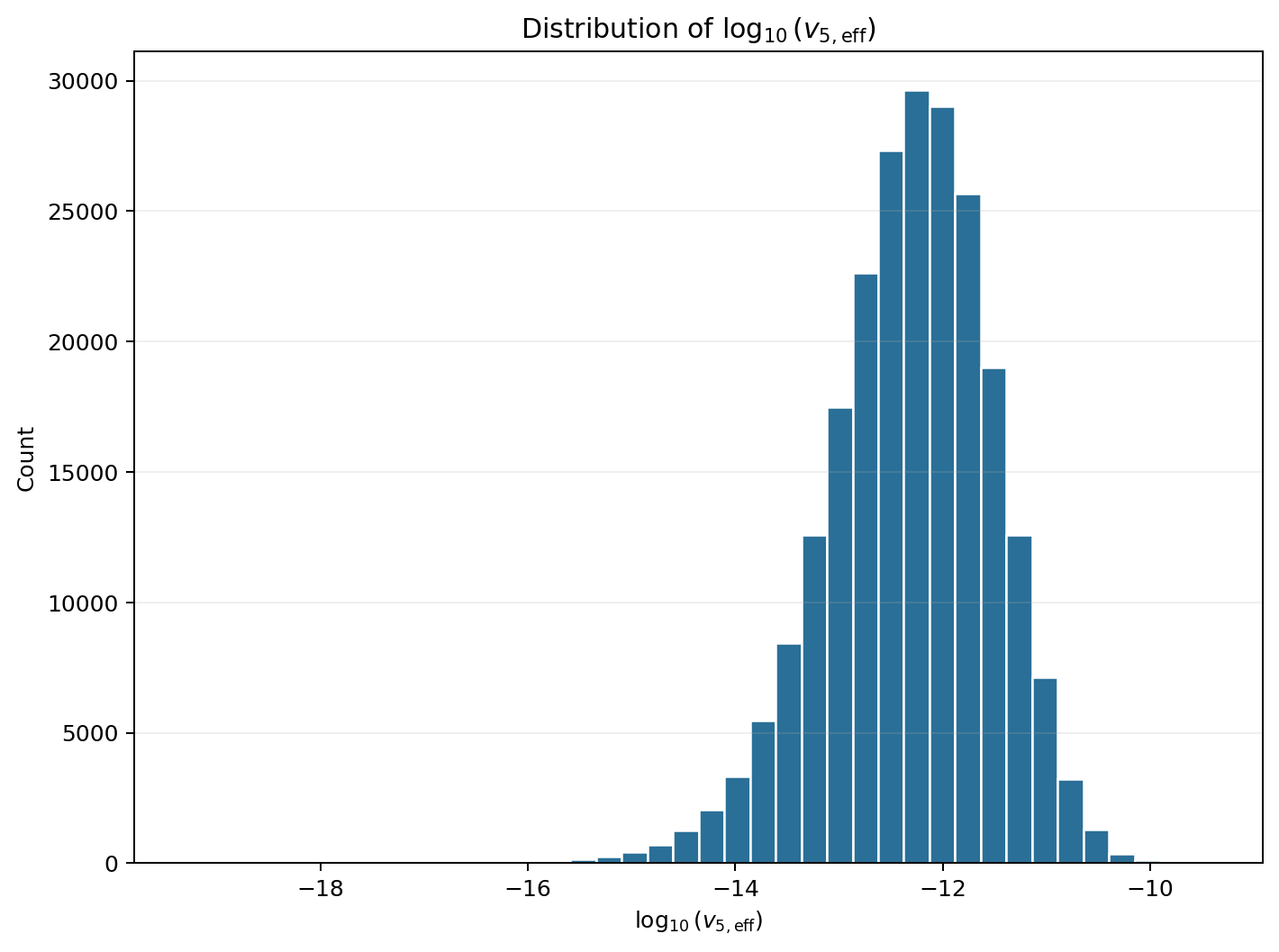}
        \vspace{-0.5em}
        \centerline{\small (b) Effective tangent volume}
    \end{minipage}
    \caption{
    Pointwise differential diagnostics for the Clifford-word cores.
    Panel (a) shows the distribution of \(\sigma_{\min}(A_{\rm pt})\); the smallest recorded value is \(4.08\times10^{-9}\), so no sampled configuration reaches the numerical rank floor.
    Panel (b) shows the distribution of \(\log_{10}(v_{5,\mathrm{eff}})\).
    Both panels are finite-precision sampled diagnostics rather than uniform bounds over \(L^5\).
    }
    \label{fig:app-clifford-rank-diagnostics}
\end{figure*}

\section{Clifford-core pointwise diagnostics}
\label{app:clifford-diagnostics}

This appendix records the pointwise differential diagnostics for the
Clifford-word core family, complementing the summary in
Sec.~\ref{subsec:clifford-pointwise-diagnostics}.
Across the tested data set the smallest recorded value of
\(\sigma_{\min}(A_{\rm pt})\) is \(4.08\times10^{-9}\), several orders of
magnitude above the numerical rank floor, and the effective-volume
distribution gives a complementary measure of the local tangent-volume scale.
The plotted data comprise the \(230{,}400\) core--configuration pairs obtained from the shared pool of Haar-sampled local configurations described in Sec.~\ref{subsec:clifford-pointwise-diagnostics}.
As emphasized in the main text, these are finite-precision sampled
diagnostics, not uniform analytic lower bounds over \(L^5\).

\section{A refined commutator observation for the five Clifford blocks}
\label{app:refined-commutator-geometry}

This appendix records an additional Lie-algebraic pattern observed in the
five-block Pauli-label splitting of \(K_{\rm Cl}\).  This structure is not
needed for the local-universality certificate in the main text; it is included
only as a useful byproduct of the finite Clifford-label analysis.

Recall that
\[
    V\setminus\{0\}
    =
    \Lambda_0\sqcup\Lambda_1\sqcup\Lambda_2\sqcup\Lambda_3\sqcup\Lambda_4,
    \qquad
    \Lambda_t=S_{K_{\rm Cl}}^t(\Lambda_{\rm loc}),
\]
and write
\[
    \mathfrak s_t:=s(\Lambda_t)\subset\mathfrak{su}(9).
\]
For a sign pair \([v]=\{v,-v\}\), let
\[
    \mathfrak p_{[v]}:=s(\{v,-v\})
\]
be the associated two-dimensional real Pauli plane.  Thus
\[
    \mathfrak s_t
    =
    \bigoplus_{[v]\in\Lambda_t/\{\pm1\}}\mathfrak p_{[v]} .
\]

The basic microscopic rule comes from the Weyl--Pauli multiplication law
\[
    W_uW_v=\omega^{[u,v]}W_vW_u,
    \qquad
    \omega=e^{2\pi i/3}.
\]
Consequently,
\[
    [\mathfrak p_{[u]},\mathfrak p_{[v]}]=0
    \quad\Longleftrightarrow\quad
    [u,v]=0,
\]
while for \([u,v]\ne0\),
\[
    [\mathfrak p_{[u]},\mathfrak p_{[v]}]
    =
    \mathfrak p_{[u+v]}\oplus\mathfrak p_{[u-v]} .
\]
This rule is the only input needed for the following finite-block observation.

\begin{proposition}[Commutator fusion of Clifford blocks]
\label{prop:block-level-commutator-fusion}
For every \(i\ne j\) in \(\mathbb Z_5\),
\[
    [\mathfrak s_i,\mathfrak s_j]
    =
    \bigoplus_{r\in\mathbb Z_5\setminus\{i,j\}}\mathfrak s_r .
\]
\end{proposition}

\begin{proof}
The cyclic group generated by \(S_{K_{\rm Cl}}\) acts on unordered pairs
\(\{i,j\}\subset\mathbb Z_5\) with exactly two orbits: the adjacent pairs
\(\{i,i+1\}\) and the next-nearest pairs \(\{i,i+2\}\) (indices mod \(5\)).
Transporting a single adjacent pair therefore does not reach the
next-nearest pairs, and it suffices to check one representative of each
orbit; we take \((0,1)\) and \((0,2)\).
Using the explicit label lists in
Appendix~\ref{app:explicit-label-splitting}, one verifies directly over
\(\mathbb F_3\), for both representative pairs \((i,j)\), that whenever
\[
    [u]\in\Lambda_i/\{\pm1\},
    \qquad
    [v]\in\Lambda_j/\{\pm1\},
\]
the nonzero output sign pairs \([u+v]\) and \([u-v]\) lie in the three
complementary blocks
\(\Lambda_r/\{\pm1\}\), \(r\in\mathbb Z_5\setminus\{i,j\}\),
and moreover that every sign pair in these complementary blocks occurs as one
of the outputs.
The Pauli-plane commutator rule above then gives
\[
    [\mathfrak s_0,\mathfrak s_1]
    =
    \mathfrak s_2\oplus\mathfrak s_3\oplus\mathfrak s_4,
    \qquad
    [\mathfrak s_0,\mathfrak s_2]
    =
    \mathfrak s_1\oplus\mathfrak s_3\oplus\mathfrak s_4.
\]
Transporting by the cyclic Clifford action within each orbit gives the same
statement for all ordered pairs \(i\ne j\).
(We have also verified the displayed identity directly for all ten unordered
pairs.)
\end{proof}

Thus the Clifford splitting has a stronger internal structure than simple
direct-sum spanning.  Any two distinct tangent blocks generate, by Lie
brackets, exactly the three remaining blocks.  In the explicit finite data
one can further refine this statement at the sign-pair level: the eight sign
pairs in each block can be grouped into four two-element classes, and the
microscopic commutators are organized by a Klein-four colored bridge table.
The local reversing symmetry of \(K_{\rm Cl}\) adds a dihedral reflection of
this bridge table.  These refinements are useful for recognizing hidden
symmetries of the Clifford-word construction, but they are not used in the
main local or numerical universality arguments.

\section{Regular-pair construction and fold diagnostics}
\label{app:regular-pair-mechanism}
\label{app:regular-pairs}

This appendix gives the numerical and differential-geometric details behind
the regular-pair evidence summarized in
Sec.~\ref{subsec:clifford-pointwise-diagnostics}.  The main text summarizes the 
all-\(100\) verification result.  Here we record the underlying
fold-wall logic, the map conventions used in the numerical validation, the
step-gap diagnostics that identify the observed rank loss, the construction
of fold-like critical candidates, the finite-difference fold check, and the strict
recomputation protocol used to accept regular partners.

All computations in this appendix use the determinant-normalized short
Clifford-word core
\[
    K_{\rm Cl}
    =
    e^{5i\pi/6}
    \operatorname{SUM}^{\dagger}P_2\operatorname{SUM}F_1\operatorname{SUM}
    \in SU(9).
\]
The all-\(100\) data set concerns this fixed short-word representative.  The
other good Clifford-word representatives discussed in
Appendix~\ref{app:clifford-family-symmetries} are locally equivalent to
\(K_{\rm Cl}\), but were not separately used in the all-\(100\) regular-pair
search.

\subsection{Fold values and regular partners}
\label{app:regular-pair-fold-wall-logic}

The critical set \(\Sigma_K\) and critical value set \(C_K\) were defined in
Eqs.~\eqref{eq:critical-set} and \eqref{eq:critical-values}.  A useful
sharpening of the boundary inclusion
\[
    \partial(\operatorname{Im}\Phi_K)\subseteq C_K
\]
is that boundary values must be critical along their entire fiber.

\begin{proposition}[Boundary values are fully critical]
\label{prop:app-boundary-values-fully-critical}
Let \(U\in\partial(\operatorname{Im}\Phi_K)\).  Then every point in
\(\Phi_K^{-1}(U)\) is critical.  Equivalently,
\[
    \partial(\operatorname{Im}\Phi_K)
    \subseteq
    \{U\in C_K:\Phi_K^{-1}(U)\subset\Sigma_K\}.
\]
\end{proposition}

\begin{proof}
If \(U=\Phi_K(\mathbf L)\) for a regular point
\(\mathbf L\in L^5\), then \(D\Phi_K(\mathbf L)\) is an isomorphism because
\(\dim L^5=\dim SU(9)=80\).  By the inverse function theorem~\cite{GuilleminPollack1974},
\(\Phi_K\) maps a neighborhood of \(\mathbf L\) diffeomorphically onto an
open neighborhood of \(U\).  Hence \(U\) is an interior point of
\(\operatorname{Im}\Phi_K\), and cannot lie on its boundary.
\end{proof}

For smooth maps between manifolds of the same dimension, the stable
codimension-one singularity relevant to a local hypersurface boundary is the
simple fold~\cite{Whitney1955,GolubitskyGuillemin1973}.  In local coordinates, the
normal form is
\[
    (x_1,\ldots,x_{79},t)
    \longmapsto
    (x_1,\ldots,x_{79},t^2).
\]
Thus the fold branch is locally one-sided, and its critical value hypersurface
is a natural candidate for a local image wall.  This is why our numerical
search focuses on corank-one points with a clear singular-value gap
\[
    \sigma_1\ll\sigma_2
\]
and with nonzero quadratic variation in the null direction.

Even if a point is a simple fold for one branch, it need not be a global
boundary point of the image.  The same target value may be covered by another
regular branch.

\begin{definition}[Regular partner]
\label{def:app-regular-partner}
Let \(\mathbf L_{\rm bad}\in\Sigma_K\), and set
\[
    U_{\rm bad}:=\Phi_K(\mathbf L_{\rm bad}).
\]
A point \(\mathbf L_{\rm reg}\in L^5\) is called a regular partner of
\(\mathbf L_{\rm bad}\) if
\[
    \Phi_K(\mathbf L_{\rm reg})=U_{\rm bad},
    \qquad
    \rank D\Phi_K(\mathbf L_{\rm reg})=80.
\]
\end{definition}

\begin{proposition}[Regular partners remove boundary walls]
\label{prop:app-regular-partners-remove-walls}
If \(\mathbf L_{\rm bad}\in\Sigma_K\) has a regular partner, then
\(\Phi_K(\mathbf L_{\rm bad})\) is an interior point of
\(\operatorname{Im}\Phi_K\).  In particular, it is not a boundary value of
the reachable set.
\end{proposition}

\begin{proof}
Let \(\mathbf L_{\rm reg}\) be a regular partner and set
\[
    U=\Phi_K(\mathbf L_{\rm reg})=\Phi_K(\mathbf L_{\rm bad}).
\]
Since \(\mathbf L_{\rm reg}\) is regular, the inverse function theorem~\cite{GuilleminPollack1974} gives
an open neighborhood of \(U\) contained in \(\operatorname{Im}\Phi_K\).
Therefore \(U\) is an interior point of the image.
\end{proof}

The regular-pair computation is therefore a test of whether observed
fold-like critical values are genuine image-wall candidates or only
branch-level folds.

\subsection{Map conventions and validation order}
\label{app:regular-pair-map-conventions}

The manuscript uses the fixed-core map
\[
    \Phi_{K_{\rm Cl}}(L_1,\ldots,L_5)
    =
    L_5K_{\rm Cl}L_4K_{\rm Cl}L_3K_{\rm Cl}L_2K_{\rm Cl}L_1 .
\]
The numerical validation was implemented in the opposite local-layer ordering
\[
    h=(\widehat L_0,\widehat L_1,\widehat L_2,\widehat L_3,\widehat L_4)
    \in L^5,
\]
with the no-appended validation map
\[
    \Phi_{K_{\rm Cl}}^{\rm val}(h)
    =
    \widehat L_0K_{\rm Cl}\widehat L_1K_{\rm Cl}
    \widehat L_2K_{\rm Cl}\widehat L_3K_{\rm Cl}\widehat L_4 .
\]
This is the same four-core ansatz as the manuscript map, under the relabeling
\[
    (L_1,L_2,L_3,L_4,L_5)
    =
    (\widehat L_4,\widehat L_3,\widehat L_2,\widehat L_1,\widehat L_0).
\]

The search also used an appended-\(K_{\rm Cl}\) pointwise convention
\[
    F_{\rm pt}(h):=\Phi_{K_{\rm Cl}}^{\rm val}(h)K_{\rm Cl}.
\]
Since right multiplication by \(K_{\rm Cl}\) is a diffeomorphism of \(SU(9)\),
equality in the appended convention is equivalent to equality in the
no-appended manuscript convention:
\[
    F_{\rm pt}(h_{\rm reg})=F_{\rm pt}(h_{\rm bad})
    \quad\Longleftrightarrow\quad
    \Phi_{K_{\rm Cl}}^{\rm val}(h_{\rm reg})
    =
    \Phi_{K_{\rm Cl}}^{\rm val}(h_{\rm bad}).
\]
The Frobenius residuals are also equal up to floating-point roundoff, because
the Frobenius norm is unitarily invariant.  Thus the appended convention is
only a validation bookkeeping convention; it cannot create a false
regular-pair equality in the manuscript map.

The same applies to differential ranks.  Since
\[
    F_{\rm pt}=R_{K_{\rm Cl}}\circ \Phi_{K_{\rm Cl}}^{\rm val},
    \qquad
    R_{K_{\rm Cl}}(U)=UK_{\rm Cl},
\]
and \(R_{K_{\rm Cl}}\) is a diffeomorphism, we have
\[
    \rank DF_{\rm pt}(h)
    =
    \rank D\Phi_{K_{\rm Cl}}^{\rm val}(h).
\]
Therefore regularity is unchanged by the appended-\(K_{\rm Cl}\) convention.

\subsection{Layer-block decomposition and step gaps}
\label{app:regular-pair-step-gaps}

Let
\[
    M_{K_{\rm Cl}}(h)\in\mathbb R^{80\times80}
\]
denote the right-trivialized differential matrix of
\(\Phi_{K_{\rm Cl}}^{\rm val}\) in orthonormal bases of
\(\mathfrak l^5\) and \(\mathfrak{su}(9)\).
One convention detail deserves an explicit statement.
The main text parameterized layer velocities by left trivialization,
\(X_i=L_i^{-1}\dot L_i\); this appendix, matching the numerical
implementation and the scan convention below, uses right-trivialized layer
velocities \(r_j=\dot{\widehat L}_j\widehat L_j^{-1}\).
The two column conventions differ by the block-diagonal map
\(\bigoplus_j\operatorname{Ad}_{\widehat L_j^{-1}}\), which is orthogonal with
respect to the Hilbert--Schmidt metric, so ranks, singular values, and step
gaps are identical in the two conventions.  We decompose it into five
\(80\times16\) local-layer blocks:
\[
    M_{K_{\rm Cl}}(h)
    =
    \bigl[
        C_0(h)\;C_1(h)\;C_2(h)\;C_3(h)\;C_4(h)
    \bigr],
\]
where
\[
    C_j(h):\mathfrak l\to\mathfrak{su}(9)
\]
is the contribution of the \(j\)-th local layer in validation order.  To
avoid confusing the linear maps with their images, set
\[
    \mathcal C_j(h):=\operatorname{im}C_j(h),
    \qquad
    W_j(h):=\mathcal C_0(h)+\cdots+\mathcal C_j(h).
\]
Thus \(W_j(h)\) is the tangent subspace opened by the first \(j+1\) layer
blocks.  For \(j=1,\ldots,4\), define the residual block
\[
    D_j(h)
    :=
    P_{W_{j-1}(h)^\perp}\circ C_j(h)
    :
    \mathfrak l\to W_{j-1}(h)^\perp,
\]
and the step gap
\[
    d_j(h):=\sigma_{\min}(D_j(h)).
\]
The quantity \(d_j\) measures how transversely the \(j\)-th layer block augments 
the span of the previous blocks.

If \(d_j(h)>0\), then the \(j\)-th block opens \(16\) new tangent directions.
In particular, if the first four layer blocks have opened a \(64\)-dimensional
subspace and \(D_4(h)\) has rank \(15\), then
\[
    \rank M_{K_{\rm Cl}}(h)
    =
    64+15
    =
    79.
\]
The observed critical candidates exhibit precisely this final-step closing pattern:
\(d_4\) collapses to the floating-point scale, with
\(\operatorname{median}_m d_4=1.28\times10^{-15}\), while the earlier step
gaps remain macroscopic:
\[
    \min_m d_2=1.06,
    \qquad
    \min_m d_3=0.509.
\]
We refer to this as the \(d_4\)-degeneration model.

\subsection{Search for fold-like critical candidates}
\label{app:regular-pair-bad-search}

The all-\(100\) fold-like critical candidates were generated by a structured middle-layer
search rather than by a blind \(80\)-dimensional all-layer random restart.
Each local layer is parameterized by a real \(16\)-vector
\[
    p_i=(a_i,b_i)\in\mathbb R^8\oplus\mathbb R^8
\]
using the Gell--Mann basis:
\[
    L_i(p_i)
    =
    \exp\!\left(i\sum_{\alpha=1}^{8}a_{i,\alpha}\lambda_\alpha\right)
    \otimes
    \exp\!\left(i\sum_{\alpha=1}^{8}b_{i,\alpha}\lambda_\alpha\right).
\]
A five-layer point is stored as a \(5\times16\) real array
\[
    p=(p_0,p_1,p_2,p_3,p_4).
\]

The critical-candidate search uses random normalized directions
\[
    A,B\in\mathbb R^{3\times16}
\]
on the middle three layers and sets the endpoint layers to zero:
\[
    p_0=p_4=0.
\]
For loop parameters \(r\ge0\) and \(t\in[0,1]\), the middle-layer parameters
are
\[
    (p_1,p_2,p_3)
    =
    r\left[
        \bigl(\cos(2\pi t)-1\bigr)A+\sin(2\pi t)B
    \right].
\]
This gives a two-parameter family
\[
    h(r,t;A,B)\in L^5
\]
which returns to the identity at \(t=0\) and \(t=1\), while exploring
nontrivial middle-layer deformations at intermediate \(t\).

For each random seed, the scripts first perform a coarse scan over
\((r,t)\).  At each point they compute \(M_{K_{\rm Cl}}(h)\) and its two
smallest singular values
\[
    \sigma_1(h)\le\sigma_2(h).
\]
Promising starts are then refined only in the two loop coordinates \((r,t)\),
using a Nelder--Mead minimization of
\[
    \log\!\left(\sigma_1(M_{K_{\rm Cl}}(h))^2+10^{-300}\right),
\]
with a penalty when \(\sigma_2\) falls below the prescribed gap threshold.

A refined point is accepted as a fold-like critical candidate only if it passes
the filters
\[
    \sigma_1<10^{-12},
    \qquad
    \sigma_2>10^{-3},
\]
and
\[
    d_4<10^{-9},
    \qquad
    d_2>5\times10^{-2},
    \qquad
    d_3>5\times10^{-2}.
\]
A separation filter is applied before accepting a new candidate: during
candidate generation, two candidates whose stored \(5\times16\)
real local-coordinate arrays were closer than \(0.85\) in Euclidean norm were
treated as the same rediscovered loop minimum.  This filter is only a
deduplication device; all geometric claims use the subsequent strict
recomputation of singular spectra, step gaps, and regular-pair residuals.

The final archived set contains \(100\) candidates generated by this
middle-three-layer loop pipeline, including a resumed search stage used to
complete the data set.  All accepted candidates are verified a posteriori
from the saved local parameters by strict recomputation of the rank-\(79\)
singular spectrum, the gap to \(\sigma_2\), and the final-step
\(d_4\)-closing diagnostic; the archive contains source-compatible generation
code together with an audit record of the original pipeline, so the
verification is independent of the original generation run.

\subsection{Finite-difference fold diagnostic}
\label{app:regular-pair-fold-scan}

The singular-value and step-gap diagnostics identify corank-one candidates,
but a simple fold also requires a second-order nondegeneracy condition in the
null direction~\cite{Whitney1955,GolubitskyGuillemin1973}.  We test this condition with a one-dimensional finite-difference
scan in validation-order local coordinates.

Let
\[
    h_{\rm bad}
    =
    (\widehat L_0,\widehat L_1,\widehat L_2,\widehat L_3,\widehat L_4)
\]
be a retained critical candidate.  Let
\[
    r=(r_0,r_1,r_2,r_3,r_4)\in\mathfrak l^5,
    \qquad
    \ell\in\mathfrak{su}(9)
\]
be approximate right and left null vectors:
\[
    M_{K_{\rm Cl}}(h_{\rm bad})r\approx0,
    \qquad
    \ell^T M_{K_{\rm Cl}}(h_{\rm bad})\approx0.
\]
The scan curve is
\[
\begin{aligned}
    h_{\rm bad}(\epsilon)
    :=
    (&e^{\epsilon r_0}\widehat L_0,\,
      e^{\epsilon r_1}\widehat L_1,\,
      e^{\epsilon r_2}\widehat L_2,\\
     &e^{\epsilon r_3}\widehat L_3,\,
      e^{\epsilon r_4}\widehat L_4).
\end{aligned}
\]
This convention satisfies
\[
    \dot{\widehat L}_j\widehat L_j^{-1}=r_j
    \qquad
    \text{at }\epsilon=0.
\]

Set
\[
    U_{\rm bad}:=\Phi_{K_{\rm Cl}}^{\rm val}(h_{\rm bad}),
\]
and use the right-trivialized target residual
\[
    \Theta_{\rm bad}(h)
    :=
    \log\!\left(
        \Phi_{K_{\rm Cl}}^{\rm val}(h)U_{\rm bad}^{-1}
    \right)
    \in\mathfrak{su}(9),
\]
where the principal logarithm is used on the small scan interval.  The
projected fold coordinate is
\[
    y(\epsilon)
    :=
    \bigl\langle
        \ell,\,
        \Theta_{\rm bad}(h_{\rm bad}(\epsilon))
    \bigr\rangle .
\]
Since \(r\) is a numerical right-null direction, the linear term in
\(y(\epsilon)\) should vanish to numerical precision.  A simple fold has a
nonzero quadratic coefficient:
\[
    y(\epsilon)
    =
    \frac12
    \bigl\langle
        \ell,\,
        D^2\Theta_{\rm bad}(h_{\rm bad})[r,r]
    \bigr\rangle
    \epsilon^2
    +
    O(\epsilon^3).
\]
The diagnostic therefore fits the computed values to
\[
    y(\epsilon)\approx a\epsilon^2+b\epsilon+c.
\]
A high-quality quadratic fit with \(a\ne0\) is numerical evidence for behavior 
consistent with a simple fold in the null direction.

\begin{figure*}[!tbp]
\centering
\includegraphics[width=0.92\textwidth]{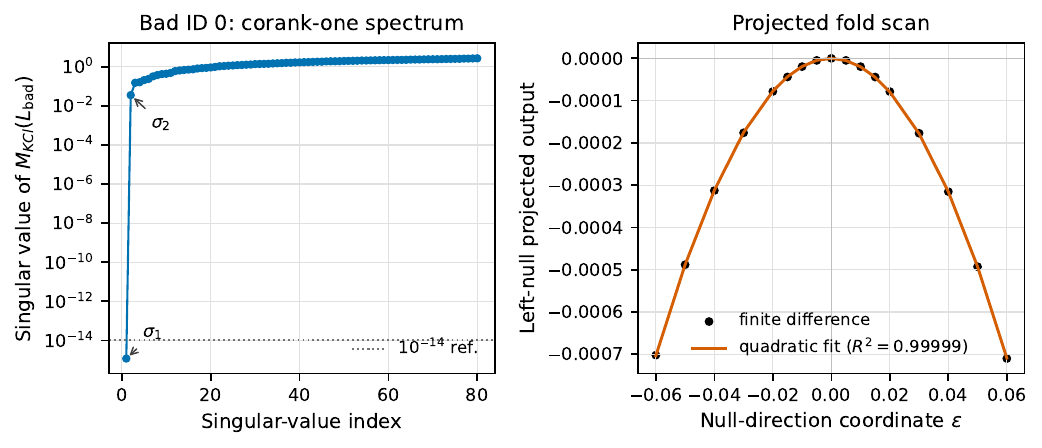}
\caption{
Fold-like diagnostics for the Clifford-word core.  Left, a representative
singular-value spectrum of \(M_{K_{\rm Cl}}(h_{\rm bad})\) shows one
near-zero singular value at the numerical floor and a visible gap to
\(\sigma_2\).  Right, a finite-difference scan along the approximate
right-null direction, projected onto the left-null direction, is well fit by
a quadratic curve.  These finite-precision diagnostics are numerically
consistent with simple-fold behavior; they are not an interval-certified fold
proof.
}
\label{fig:app-clifford-foldlike-diagnostics}
\end{figure*}

For representative candidates \(m=0,\ldots,4\), the quadratic fits have
\(R^2\) values between \(0.996\) and \(0.99999\).  These scans support
simple-fold normal behavior for representative members of the data set, but
they are finite-precision diagnostics rather than exact singularity
certificates.

\subsection{Regular-partner search and strict recomputation}
\label{app:regular-pair-strict-recompute}

For each fold-like critical candidate \(h_{\rm bad}^{(m)}\), the regular-partner
search fixes the no-appended target
\[
    U_{\rm bad}^{(m)}
    :=
    \Phi_{K_{\rm Cl}}^{\rm val}(h_{\rm bad}^{(m)}).
\]
It then searches for a second point \(h_{\rm reg}^{(m)}\in L^5\) satisfying
\[
    \Phi_{K_{\rm Cl}}^{\rm val}(h_{\rm reg}^{(m)})
    \approx
    U_{\rm bad}^{(m)}
\]
and having full-rank differential.

This mechanism was first identified using a fold-normal continuation strategy:
one perturbs the bad target along the numerical cokernel direction, solves on
the one-sided regular branch, and then refines the solution back toward the
original bad value.  The all-\(100\) search uses this idea as one proposal
mode, together with middle-layer perturbations, all-layer perturbations,
end-layer perturbations, soft-null-direction perturbations, and additional
fold-normal-branch seeds.

Each proposal is refined by Newton-type correction in local exponential
coordinates.  The optimizer only proposes a candidate.  Acceptance is decided
after reloading the saved local-layer parameters and recomputing the map
values and differential singular spectra from scratch.

For each pair ID \(m\), the strict verification script loads
\[
    h_{\rm bad}^{(m)},
    \qquad
    h_{\rm reg}^{(m)}
\]
and computes both residuals
\[
    r_m^{\rm pt}
    :=
    \left\|
    F_{\rm pt}(h_{\rm reg}^{(m)})
    -
    F_{\rm pt}(h_{\rm bad}^{(m)})
    \right\|_{\rm F},
\]
and
\[
    r_m^{\rm main}
    :=
    \left\|
    \Phi_{K_{\rm Cl}}^{\rm val}(h_{\rm reg}^{(m)})
    -
    \Phi_{K_{\rm Cl}}^{\rm val}(h_{\rm bad}^{(m)})
    \right\|_{\rm F}.
\]
It also recomputes the full singular-value spectrum of
\[
    M_{K_{\rm Cl}}(h_{\rm reg}^{(m)}).
\]

The product local-layer distance used in the verification is
\[
    d_{L^5}(h,h')
    =
    \left(
    \sum_{j=0}^{4}
    \left\|
    \log\!\left(\widehat L_j^\dagger \widehat L'_j\right)
    \right\|_{\rm F}^{2}
    \right)^{1/2},
\]
where \(h=(\widehat L_0,\ldots,\widehat L_4)\),
\(h'=(\widehat L'_0,\ldots,\widehat L'_4)\), and the principal matrix
logarithm is used.

A proposed pair is accepted only if
\[
    r_m^{\rm main}<10^{-10},
    \qquad
    r_m^{\rm pt}<10^{-10},
\]
\[
    \rank M_{K_{\rm Cl}}(h_{\rm reg}^{(m)})=80,
    \qquad
    \sigma_{\min}\bigl(M_{K_{\rm Cl}}(h_{\rm reg}^{(m)})\bigr)>10^{-4},
\]
and
\[
    d_{L^5}\bigl(h_{\rm bad}^{(m)},h_{\rm reg}^{(m)}\bigr)>10^{-1}.
\]
The rank tolerance used in the all-\(100\) verification is \(10^{-8}\).  Thus
accepted partners are not merely alternative nearly singular points in the
same critical fiber; they are numerically regular, finite-distance preimages
of the same target values.

\subsection{All-100 verification and figures}
\label{app:regular-pair-all100-results}

All \(100\) observed fold-like critical candidates admit a numerically regular
partner under strict recomputation:
\[
    \Phi_{K_{\rm Cl}}^{\rm val}(h_{\rm reg}^{(m)})
    \approx
    \Phi_{K_{\rm Cl}}^{\rm val}(h_{\rm bad}^{(m)}),
    \qquad
    \rank M_{K_{\rm Cl}}(h_{\rm reg}^{(m)})=80.
\]
The aggregate statistics are shown in
Table~\ref{tab:app-regular-pair-all100-summary}.  The residuals are
Frobenius norms in the no-appended validation convention, which is equivalent
to the manuscript map after relabeling of local layers.

\begin{table*}[!tbp]
\caption{
Full numerical summary for the Clifford-word regular-pair experiment.
Here \(\sigma_1\) and \(\sigma_2\) are evaluated at the fold-like critical
candidates, \(d_4\) is the final step gap, \(r_m^{\rm main}\) is the
no-appended target residual, and \(d_{L^5}\) is the product local-layer
distance used in strict verification.
}
\label{tab:app-regular-pair-all100-summary}
\begin{ruledtabular}
\begin{tabular}{lccc}
Quantity & Minimum & Median & Maximum \\
\hline
\multicolumn{4}{c}{Fold-like critical candidates} \\
\(\sigma_1(M_{K_{\rm Cl}}(h_{\rm bad}))\) &
\(4.33{\times}10^{-17}\) &
\(3.98{\times}10^{-16}\) &
\(1.20{\times}10^{-15}\) \\
\(\sigma_2(M_{K_{\rm Cl}}(h_{\rm bad}))\) &
\(9.87{\times}10^{-3}\) &
\(6.26{\times}10^{-2}\) &
\(1.54{\times}10^{-1}\) \\
\(d_4(h_{\rm bad})\) &
\(9.33{\times}10^{-16}\) &
\(1.28{\times}10^{-15}\) &
\(2.16{\times}10^{-15}\) \\
\hline
\multicolumn{4}{c}{Regular partners} \\
\(r_m^{\rm main}\) &
\(8.76{\times}10^{-16}\) &
\(3.16{\times}10^{-15}\) &
\(8.31{\times}10^{-13}\) \\
\(\sigma_{\min}(M_{K_{\rm Cl}}(h_{\rm reg}))\) &
\(3.71{\times}10^{-4}\) &
\(2.46{\times}10^{-2}\) &
\(9.75{\times}10^{-2}\) \\
\(d_{L^5}(h_{\rm reg},h_{\rm bad})\) &
\(1.26{\times}10^{-1}\) &
\(1.59\) &
\(11.74\)
\end{tabular}
\end{ruledtabular}
\end{table*}

\begin{figure*}[!tbp]
\centering
\includegraphics[width=0.92\textwidth]{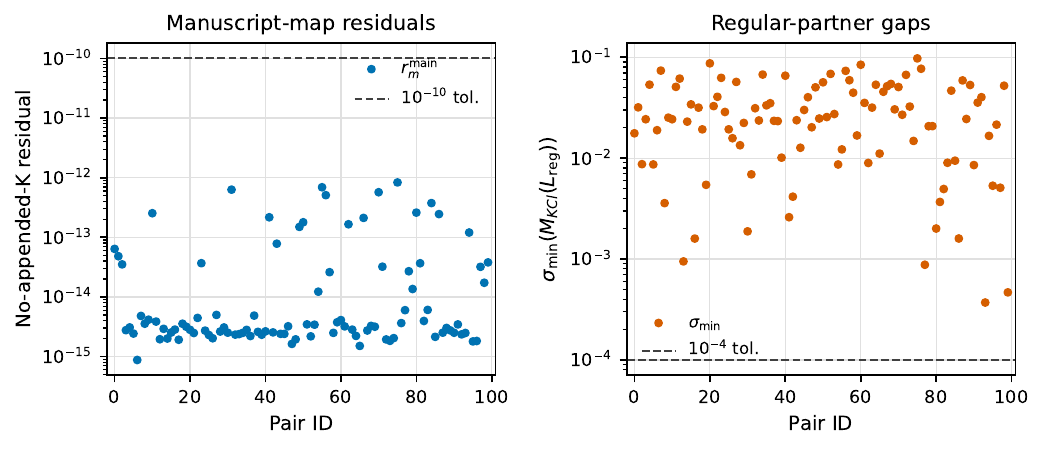}
\caption{
All-\(100\) regular-pair verification.  For all \(100/100\) observed fold-like
critical candidates, strict recomputation from the saved local parameters finds
a distinct numerically regular partner with no-appended manuscript-map residual below
\(10^{-10}\), Jacobian rank \(80\) at tolerance \(10^{-8}\), and
\(\sigma_{\min}(M_{K_{\rm Cl}}(h_{\rm reg}))>10^{-4}\).  The finite
critical-to-regular separations are summarized in
Table~\ref{tab:app-regular-pair-all100-summary}.
}
\label{fig:app-regular-pair-all100}
\end{figure*}

The appended-\(K_{\rm Cl}\) residuals are comparable, with median
\(3.19\times10^{-15}\) and maximum \(8.31\times10^{-13}\).  The last row of
Table~\ref{tab:app-regular-pair-all100-summary} shows that the accepted
partners are finite-distance preimages, not infinitesimal perturbations of
the same fold point.

At the stated numerical tolerances, every accepted pair is consistent with the candidate output being covered by a distinct regular branch. 
Consequently, none of the 100 observed fold-like candidates remains a numerical image-wall candidate under this verification protocol.

\subsection{Archived outputs and reproducibility}
\label{app:regular-pair-archive}

The regular-pair computation is archived as a machine-readable numerical
package.  The archive separates raw local-layer parameters, critical-candidate
generation records, strict verification outputs, aggregate summary files, and
figure-generation inputs.  Exact file names, directory paths, and command-line
instructions are recorded in the package README and manifest rather than
reproduced in the manuscript.

The strict verification outputs record, for each pair ID, the
appended-\(K_{\rm Cl}\) residual, the no-appended manuscript-map residual, the
numerical rank at the proposed regular partner, the smallest singular value
at the partner, the critical-to-regular local-layer distance, and the final
acceptance status.  The aggregate quantities reported in
Table~\ref{tab:app-regular-pair-all100-summary} are computed from these
strict verification outputs.

The figures in this appendix are generated from post-processed CSV/JSON files,
not by rerunning the search.  The fold-like diagnostic figure uses the saved
representative singular spectrum and finite-difference fold-scan data.  The
all-\(100\) verification figure uses the saved pair-by-pair manuscript-map
residuals and partner singular-value gaps.  The archive also contains the
plotting script, strict batch-verification script, and README instructions
for regenerating the figures and verification tables from the saved data.

As with all numerical certificates in this paper, the regular-pair archive is
a finite-precision verification package.  It supports the claim that all
\(100\) observed fold-like candidates admit distinct regular partners under
the stated tolerances.  It is not an interval-arithmetic proof that every
critical value of \(\Phi_{K_{\rm Cl}}\) has a regular preimage.

\subsection{Entangling-power diagnostics for \texorpdfstring{\(K_{\rm Cl}\)}{KCl}}
\label{app:kcl-entangling-power}

For a two-qutrit gate \(U\), we use the standard linear-entropy entangling
power
\[
e_p(U)=\int d\psi\,d\phi\,
\left[1-\operatorname{Tr}(\rho_A^2)\right],
\]
where
\(\rho_A=\operatorname{Tr}_B[U(\ket{\psi}\ket{\phi})(\bra{\psi}\bra{\phi})U^\dagger]\)
and the average is over product inputs.  Equivalently, for \(d\times d\)
systems one may use the Zanardi formula~\cite{ZanardiZalkaFaoro2000,WangSandersBerry2003}
\[
    e_p(U)
    =
    \left(\frac{d}{d+1}\right)^2
    \left[
        E_{\rm op}(U)+E_{\rm op}(U\mathsf S)-E_{\rm op}(\mathsf S)
    \right],
\]
where \(\mathsf S\) is the bipartite SWAP gate and \(E_{\rm op}\) is the
operator entanglement.  If
\(U=\sum_\alpha s_\alpha A_\alpha\otimes B_\alpha\) is an
operator-Schmidt decomposition with Hilbert--Schmidt orthonormal local
operators, then \(\sum_\alpha s_\alpha^2=d^2\) and
\[
    E_{\rm op}(U)=1-\frac{1}{d^4}\sum_\alpha s_\alpha^4 .
\]
For the short Clifford-word core,
\[
    s_\alpha(K_{\rm Cl})
    =
    s_\alpha(K_{\rm Cl}\mathsf S)
    =
    s_\alpha(\mathsf S)
    =
    1,
    \qquad \alpha=1,\ldots,9.
\]
Thus
\[
E_{\rm op}(K_{\rm Cl})=E_{\rm op}(K_{\rm Cl}\mathsf S)
=E_{\rm op}(\mathsf S)=\frac{8}{9},
\]
and for \(d=3\),
\[
    e_p(K_{\rm Cl})
    =
    \left(\frac{3}{4}\right)^2
    \left(\frac{8}{9}+\frac{8}{9}-\frac{8}{9}\right)
    =
    \frac{1}{2}.
\]
This is the maximal possible two-qutrit entangling power in the unnormalized
linear-entropy convention, where the general upper bound is
\((d-1)/(d+1)\)~\cite{ZanardiZalkaFaoro2000,WangSandersBerry2003}.

\section{Superconducting Hamiltonian generators and \texorpdfstring{\(K_{\rm sc}\)}{Ksc} matrix}
\label{app:sc-core-matrix}
\label{app:sc-generators}
\label{app:sc-diagnostics}

This appendix records the Hamiltonian generators, matrix representation, and
basic consistency checks for the superconducting transmon core used in the main text.
In transmon qutrits, the dominant single-photon drive matrix elements connect adjacent levels; accordingly, our effective model retains adjacent-level drives and exchange-type flip-flop couplings between the two truncated qutrits.
\begin{figure*}[!t]
    \centering
    \includegraphics[width=0.92\textwidth]{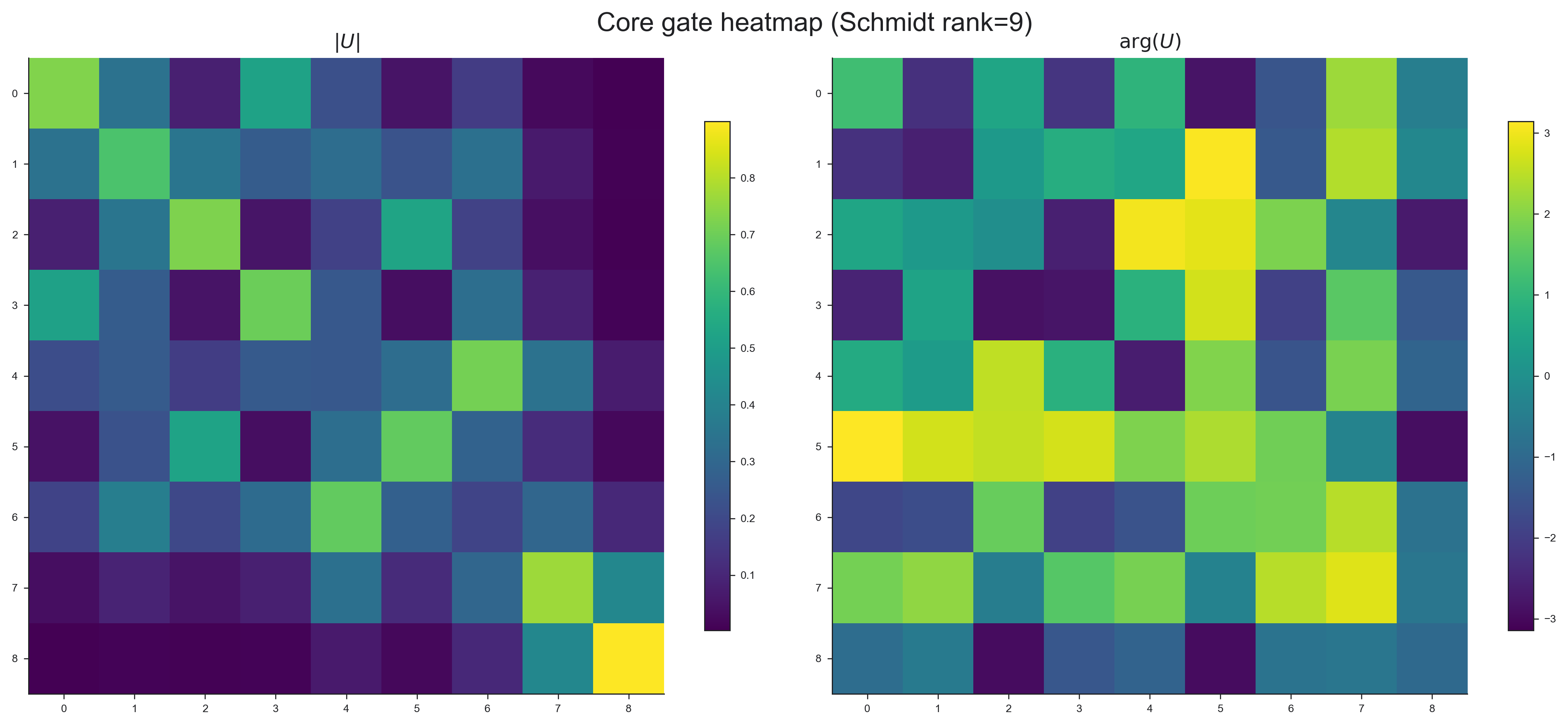}
    \caption{
    Matrix representation of the representative asymmetric-drive superconducting core \(\Ksc\) in the standard two-qutrit computational basis.
    The left panel shows the entrywise magnitude \(|\Ksc|\), and the right panel shows the entrywise phase \(\arg(\Ksc)\) in radians.
    The title reports the operator-Schmidt rank of the core, which is \(9\) for the two-qutrit bipartition.
    }
    \label{fig:app-ksc-heatmap}
\end{figure*}
Let
\[
    P_{ij,kl}:=|ij\rangle\langle kl|
\]
denote the matrix unit in the truncated two-qutrit basis.  The diagonal drift
generators are
\begin{align}
    H^{\rm drift}_1
    &=
    \sum_{i=1}^{2}\sum_{j=0}^{2}P_{ij,ij},
    &
    H^{\rm drift}_2
    &=
    \sum_{i=0}^{2}\sum_{j=1}^{2}P_{ij,ij},
    \nonumber\\
    H^{\rm drift}_3
    &=
    \sum_{j=0}^{2}P_{2j,2j},
    &
    H^{\rm drift}_4
    &=
    \sum_{i=0}^{2}P_{i2,i2}.
    \label{eq:app-sc-drift-terms}
\end{align}
The single-qutrit drive generators are
\begin{align}
    H^{\rm drive}_1
    &=
    \sum_{j=0}^{2}
    \left(P_{0j,1j}+P_{1j,0j}\right),
    \nonumber\\
    H^{\rm drive}_2
    &=
    \sum_{j=0}^{2}
    \left(P_{1j,2j}+P_{2j,1j}\right),
    \nonumber\\
    H^{\rm drive}_3
    &=
    \sum_{i=0}^{2}
    \left(P_{i0,i1}+P_{i1,i0}\right),
    \nonumber\\
    H^{\rm drive}_4
    &=
    \sum_{i=0}^{2}
    \left(P_{i1,i2}+P_{i2,i1}\right).
    \label{eq:app-sc-drive-terms}
\end{align}
The exchange-type coupling generators are
\begin{align}
    H^{\rm coup}_1
    &=
    P_{10,01}+P_{01,10},
    &
    H^{\rm coup}_2
    &=
    P_{20,11}+P_{11,20},
    \nonumber\\
    H^{\rm coup}_3
    &=
    P_{11,02}+P_{02,11},
    &
    H^{\rm coup}_4
    &=
    P_{21,12}+P_{12,21}.
    \label{eq:app-sc-coupling-terms}
\end{align}
The coefficient ordering is
\begin{equation}
    \left(
    \Delta_1,\ldots,\Delta_4,
    \Omega_1,\ldots,\Omega_4,
    g_1,\ldots,g_4
    \right).
\label{eq:sc-coefficient-ordering}
\end{equation}

The matrix is represented in the standard two-qutrit computational basis
\begin{equation}
    \{|00\rangle,|01\rangle,|02\rangle,|10\rangle,\ldots,|22\rangle\}.
    \label{eq:app-ksc-computational-basis}
\end{equation}
In this basis, \(\Ksc\) is a \(9\times9\) complex unitary matrix.
The full list of matrix entries is not displayed in the main text because the matrix is a numerical object rather than a structural formula.
For inspection, Fig.~\ref{fig:app-ksc-heatmap} visualizes the magnitude and phase of the generated core.\par
The magnitude panel shows that the generated core is not close to a sparse permutation-type gate.
The phase panel shows that the entries also carry nontrivial phases.
Together, the two panels provide a compact visualization of the numerical core without listing all \(81\) complex entries.\par
\begin{figure*}[!t]
    \centering
    \begin{minipage}{0.48\textwidth}
        \centering
        \includegraphics[width=\linewidth]{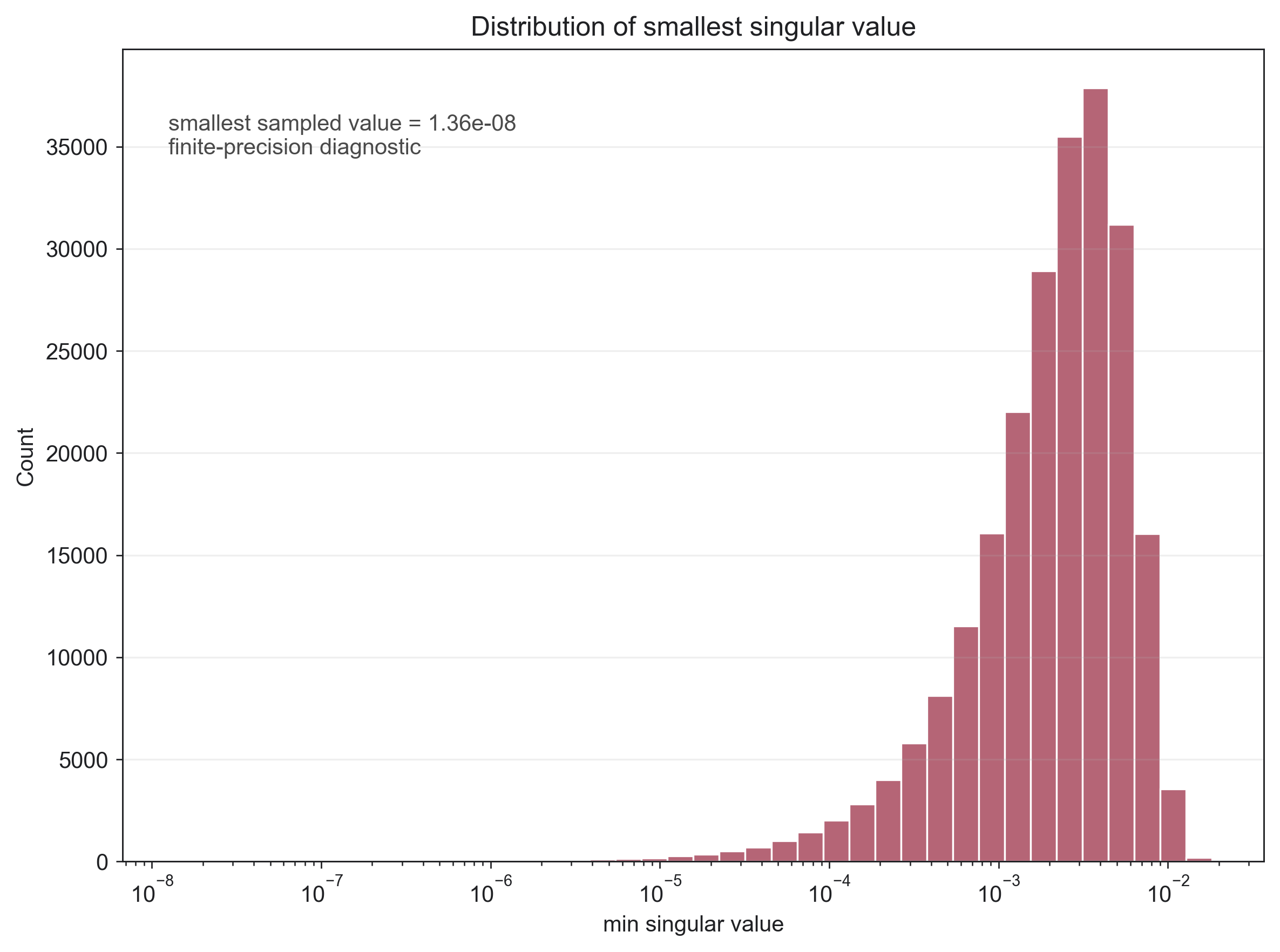}
        \vspace{-0.5em}
        \centerline{\small (a) Smallest singular value}
    \end{minipage}
    \hfill
    \begin{minipage}{0.48\textwidth}
        \centering
        \includegraphics[width=\linewidth]{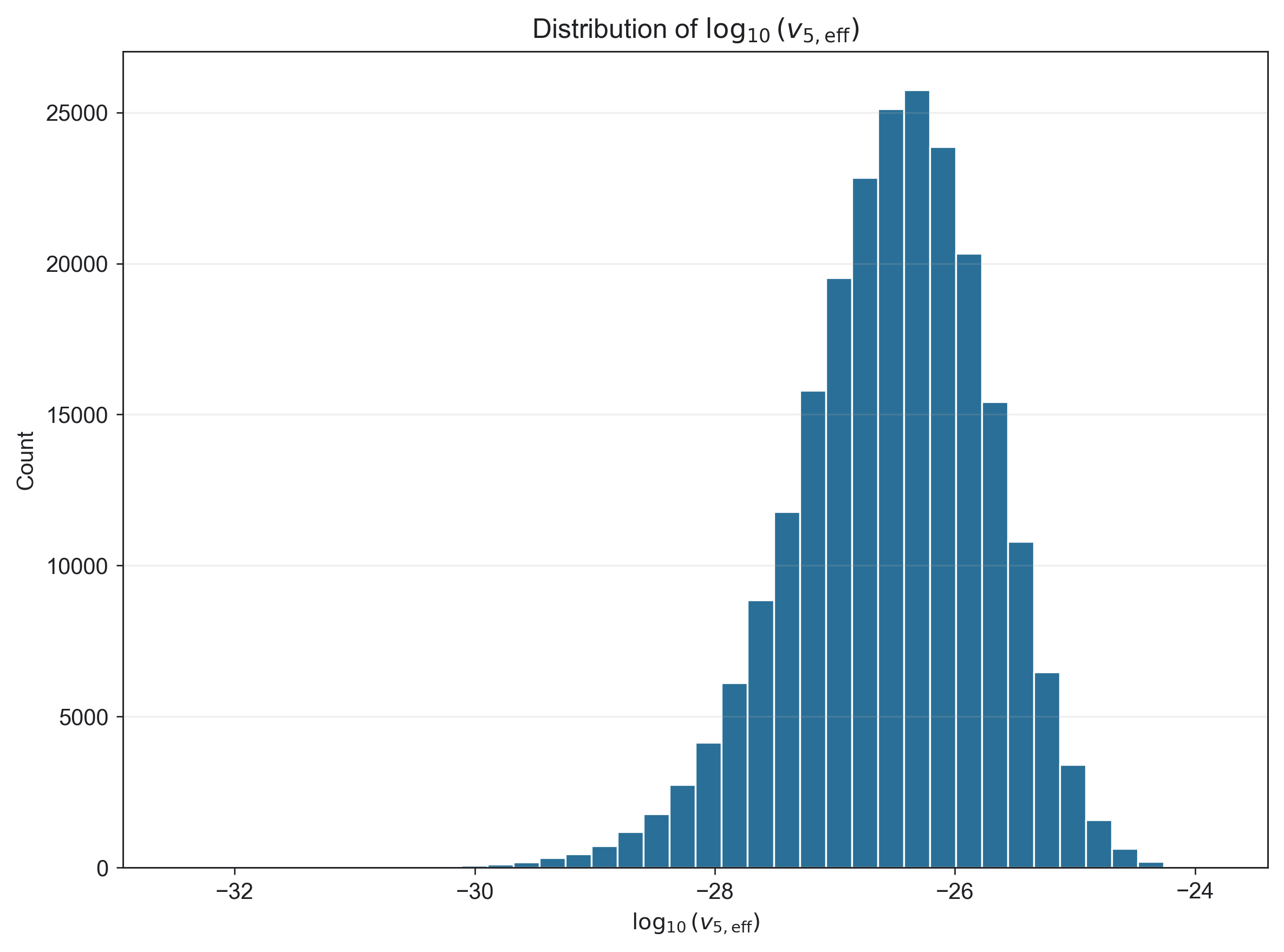}
        \vspace{-0.5em}
        \centerline{\small (b) Effective tangent volume}
    \end{minipage}
    \caption{
    Pointwise Jacobian diagnostics for the representative asymmetric-drive superconducting core \(K_{\rm sc}\), evaluated at \(230{,}400\) sampled local configurations.
    Panel (a) shows the distribution of \(\sigma_{\min}(A_{\rm pt})\); no sampled rank collapse is observed, and the smallest recorded value is \(1.36\times10^{-8}\).
    Panel (b) shows the distribution of \(\log_{10}(v_{5,\mathrm{eff}})\), where \(v_{5,\mathrm{eff}}\) is the product of all \(80\) singular values of \(A_{\rm pt}\).
    These finite-precision pointwise quantities are sampled conditioning and tangent-volume diagnostics, not uniform analytic bounds over the local-layer manifold.
    }
    \label{fig:app-ksc-rank-diagnostics}
\end{figure*}
We next record the numerical checks confirming that the generated matrix is a valid element of \(SU(9)\) up to floating-point precision.
The unitarity error is measured by
\begin{equation}
    \epsilon_{\rm unitary}
    :=
    \left\|
    \Ksc^\dagger\Ksc-I_9
    \right\|_{\rm F}.
    \label{eq:app-ksc-unitarity-error}
\end{equation}
Here \(\|\cdot\|_{\rm F}\) denotes the Frobenius norm.
The determinant-normalization error is measured by
\begin{equation}
    \epsilon_{\rm det}
    :=
    \left|
    \det(\Ksc)-1
    \right|.
    \label{eq:app-ksc-det-error}
\end{equation}
These errors are used only to confirm the numerical consistency of the generated matrix.
They are not, by themselves, evidence of expressivity or synthesis performance.\par
Table~\ref{tab:app-ksc-consistency} summarizes the consistency checks.
\begin{table}[!t]
    \centering
    \caption{
    Numerical consistency checks for the representative asymmetric-drive core \(\Ksc\).
    }
    \label{tab:app-ksc-consistency}
    \begin{tabular}{c c}
        \hline\hline
        Quantity & Value \\
        \hline
        Matrix dimension & \(9\times9\) \\
        Ambient group after normalization & \(SU(9)\) \\
        \(\epsilon_{\rm unitary}=\|\Ksc^\dagger\Ksc-I_9\|_{\rm F}\)
        & \(1.85\times10^{-14}\) \\
        \(\epsilon_{\rm det}=|\det(\Ksc)-1|\)
        & \(1.36\times10^{-15}\) \\
        Operator-Schmidt rank & \(9\) \\
        \hline\hline
    \end{tabular}
\end{table}
Finally, we record the operator-Schmidt rank of \(\Ksc\).
For a two-qutrit operator \(U\), we reshape its entries as
\begin{equation}
    U_{(i_A i_B),(j_A j_B)}
    \longmapsto
    R_U{}_{(i_A j_A),(i_B j_B)} ,
    \label{eq:app-operator-schmidt-reshaping}
\end{equation}
where \(i_A,i_B,j_A,j_B\in\{0,1,2\}\).
The singular values of \(R_U\) are the operator-Schmidt coefficients.
The operator-Schmidt rank is the number of coefficients above the chosen numerical tolerance.
For the superconducting core used in this work, this rank is
\begin{equation}
    \mathrm{SchmidtRank}_{\rm op}(\Ksc)=9.
    \label{eq:app-ksc-schmidt-rank}
\end{equation}
Thus the generated core has full operator-Schmidt rank for the two-qutrit bipartition.
This full operator-Schmidt rank should not be confused with the full-rank differential certificate used in the main text.
The operator-Schmidt rank detects nonlocal operator structure, whereas the differential certificate tests whether the four-core map opens all tangent directions of \(SU(9)\).\par
These checks establish that the generated \(\Ksc\) is a numerically valid \(SU(9)\) fixed core.
Its matrix visualization, unitarity check, determinant normalization check, and operator-Schmidt rank provide the reproducibility data needed to verify the core used in the main text.
The pointwise Jacobian-rank diagnostics and synthesis benchmarks in the main text then evaluate its expressive power.

\section{Restart protocol and target-wise convergence for the Haar-random \texorpdfstring{\(K_{\rm sc}\)}{Ksc} benchmark}
\label{app:sc-haar-restarts}
\label{app:haar-restarts}

This appendix records the restart protocol used in the Haar-random synthesis benchmark for the representative superconducting core \(\Ksc\).
Throughout this benchmark, the core \(\Ksc\) is held fixed.
Only the five local layers in the four-core architecture are optimized.
No Hamiltonian parameters are re-optimized for individual Haar-random targets.
Thus, a restart is a new initialization of the local-layer optimization problem, not a change of the physical core.\par

For each Haar-random target \(V_i\in SU(9)\), we run a multi-restart optimization over the local-layer parameters.
Each restart uses an independent initialization of the local coordinates.
Each single-qutrit factor is parameterized as \(\exp[i\sum_{a=1}^{8}\theta_a\lambda_a]\), where \(\lambda_a\) are the Gell-Mann generators, and every restart initializes the \(80\) real coefficients independently from \(\mathcal N(0,0.02^2)\).
We use Adam with learning rate \(0.05\), gradient clipping at \(2.0\), and a ReduceLROnPlateau scheduler with reduction factor \(0.5\), patience \(150\), and minimum learning rate \(10^{-6}\).
Each restart is limited to \(5000\) epochs, with at most \(100\) restarts per target; stagnation stopping uses patience \(1000\) and minimum fidelity improvement \(10^{-7}\).
The Haar-random targets are generated by complex-Gaussian QR decomposition with phase correction, followed by determinant normalization to \(SU(9)\), using seed \(7000+i\) for target index \(i\).
The optimization for target \(V_i\) stops as soon as one restart reaches the success threshold
\begin{equation}
    F_{\rm avg}\ge 0.999 .
    \label{eq:app-ksc-restart-success-threshold}
\end{equation}
If no restart reaches the threshold within the allowed restart budget, the target is recorded as unsuccessful at that budget.
In the reported \(1000\)-target benchmark, every target reaches the threshold within the allowed restart budget.\par

Let \(R_i\) denote the first successful restart index for target \(V_i\).
Equivalently,
\begin{equation}
    R_i
    :=
    \min\left\{
    r:
    F_{\rm avg}^{(i,r)}\ge 0.999
    \right\},
    \label{eq:app-ksc-first-success-restart}
\end{equation}
where \(F_{\rm avg}^{(i,r)}\) is the best average gate fidelity reached by restart \(r\) for target \(V_i\).
The value \(R_i\) measures the optimization difficulty of target \(V_i\) under the chosen initialization and optimizer.
It should not be interpreted as an additional circuit depth or as a modification of the fixed core.\par

To summarize the restart dependence over the full target ensemble, we define the empirical cumulative success curve
\begin{equation}
    C(R)
    :=
    \frac{1}{N_{\rm tar}}
    \#\left\{
    i:
    R_i\le R
    \right\},
    \qquad
    N_{\rm tar}=1000 .
    \label{eq:app-ksc-restart-cdf}
\end{equation}
Thus \(C(R)\) is the fraction of Haar-random targets that are successfully synthesized within a restart budget of \(R\).
This curve separates empirical synthesis success within a prescribed restart budget from the number of random initializations required by the chosen optimizer to reach the fidelity threshold.\par

Figure~\ref{fig:app-ksc-restart-budget} shows the empirical success rate as a function of restart budget.
The curve rises rapidly at small restart budgets.
More than half of the targets succeed on the first restart, and the success rate increases substantially by \(R=2\) and \(R=5\).
At \(R=20\), \(999\) of the \(1000\) targets reach the threshold, and the remaining target first succeeds by \(R=50\), where the curve reaches \(100\%\).
Therefore, the reported \(1000/1000\) Haar-random success result does not require changing the superconducting core or redesigning the Hamiltonian pulse. 
It is obtained by allowing multiple initializations of the same five local layers.\par

\begin{figure}[!t]
    \centering
    \includegraphics[width=0.98\columnwidth]{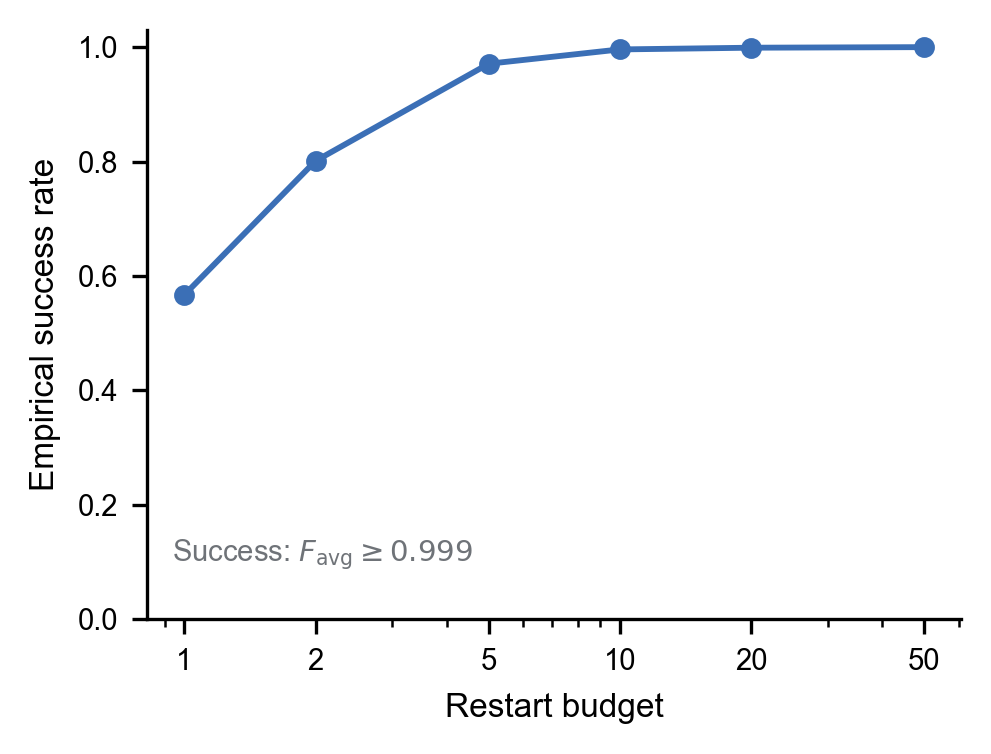}
    \caption{
    Empirical success rate as a function of restart budget for Haar-random synthesis with the representative asymmetric-drive superconducting core \(\Ksc\).
    At each displayed budget \(R\), the curve gives the fraction of the same \(1000\) Haar-random \(SU(9)\) targets for which at least one of the first \(R\) independently initialized restarts reaches \(F_{\rm avg}\ge0.999\).
    The rate is \(99.9\%\) at \(R=20\) and reaches \(100\%\) at \(R=50\), the largest displayed budget.
    }
    \label{fig:app-ksc-restart-budget}
\end{figure}

The restart-budget curve provides information complementary to the final average-gate-fidelity histogram in the main text.
Because the optimizer stops once \(F_{\rm avg}\ge0.999\) is reached, the final fidelity distribution is a threshold-reaching distribution rather than a distribution of maximum attainable fidelities.
By contrast, the curve \(C(R)\) measures the optimization difficulty of the target ensemble under the fixed-core protocol.
Targets that require more restarts are harder for the chosen optimizer and initialization distribution, even though they remain successfully synthesized within the allowed budget.\par

The need for restarts is expected because the optimization over five local layers is nonconvex.
Different initial local layers can converge to different basins of attraction.
A failed restart does not imply that the target lies outside the image of the fixed-core map.
It only means that the optimizer did not find a successful local-layer configuration from that particular initialization.
For this reason, the restart budget is part of the numerical synthesis protocol.\par

The target-wise restart counts and convergence histories are stored in the machine-readable benchmark outputs.
These files allow the benchmark to be reproduced target by target.
They also make it possible to identify outlier targets that require unusually many restarts.
Such outliers are useful for diagnosing the geometry of the optimization landscape associated with the fixed superconducting core.
In the present benchmark, all \(1000\) Haar-random targets are solved within the displayed restart range, with the final target succeeding by \(R=50\).
This supports the conclusion in the main text that the representative \(\Ksc\) has strong average-case synthesis performance under the four-core fixed-core architecture.

\section{Hamiltonian-parameter robustness details for \texorpdfstring{\(K_{\rm sc}\)}{Ksc}}
\label{app:sc-parameter-robustness}
\label{app:robustness}
This appendix gives the detailed protocol behind the two Hamiltonian-parameter robustness tests reported in Sec.~\ref{subsec:sc-structured-targets}.
Both tests use the same nominal core \(\Ksc\), the same ordered set of fixed Haar-random \(SU(9)\) targets, and the same shared-offset convention.
The perturbations are applied to the Hamiltonian coefficients before core generation, not directly to entries of the final \(9\times9\) core matrix.
Each parameter offset is fixed over the full core interval, while the nominal time dependence of the first-drive envelope is retained, and the same perturbed core is used in all four core positions.\par
\begin{table*}[!t]
    \centering
    \caption{
    Shared static Hamiltonian-coefficient-offset results for the representative superconducting core.
    Nonzero coefficients are varied multiplicatively along fixed perturbation directions, while nominally zero coefficients are held at zero.
    Protocol A reports core-wise recompilation success over \(20\) cores and \(100\) fixed targets per core.
    Protocol B reports pooled frozen-compilation target fidelities over \(100\) offset directions and \(1000\) fixed targets per direction.
    The displayed values are direct sample statistics, not analytic bounds.
    }
    \label{tab:app-robustness-summary}
    \begin{ruledtabular}
    \begin{tabular}{c c c c c c c}
        & \multicolumn{3}{c}{Protocol A: core-wise success} &
        \multicolumn{3}{c}{Protocol B: target \(F_{\rm avg}\)} \\
        \(\epsilon\) & Median & Mean & IQR & Median & 5th percentile & 95th percentile \\
        \hline
        \(0\)    & \(1.000\) & \(1.000\) & \(1.000\)--\(1.000\) & \(0.999003\) & \(0.999000\) & \(0.999031\) \\
        \(0.05\) & \(0.990\) & \(0.980\) & \(0.975\)--\(1.000\) & \(0.973560\) & \(0.927277\) & \(0.990821\) \\
        \(0.10\) & \(0.975\) & \(0.933\) & \(0.888\)--\(1.000\) & \(0.901700\) & \(0.741956\) & \(0.966330\) \\
        \(0.15\) & \(0.965\) & \(0.855\) & \(0.778\)--\(1.000\) & \(0.794632\) & \(0.515963\) & \(0.926841\) \\
        \(0.20\) & \(0.885\) & \(0.731\) & \(0.475\)--\(0.993\) & \(0.667396\) & \(0.321210\) & \(0.874557\) \\
    \end{tabular}
    \end{ruledtabular}
\end{table*}

The canonical parameter vector contains the twelve transition-resolved coefficients
\begin{equation}
    \mathbf p_0=\left(
    \Delta_1,\ldots,\Delta_4,
    \Omega_1,\ldots,\Omega_4,
    g_1,\ldots,g_4
    \right),
    \label{eq:app-robustness-parameter-vector}
\end{equation}
following the ordering in Eq.~\eqref{eq:sc-coefficient-ordering}.
The four \(\Delta_\mu\) parameters correspond to drift or detuning terms.
The four \(\Omega_\mu\) parameters correspond to local drive coefficients; the first sets the peak coefficient of the normalized asymmetric envelope, whereas the other three remain constant during the pulse.
The four \(g_\mu\) parameters correspond to exchange-type coupling terms.
The representative unperturbed values are listed in Table~\ref{tab:sc-representative-parameters}.
For each perturbation direction \(c\), a vector \(\mathbf z_c\in[-1,1]^{12}\) is sampled once with independent uniform components and reused at every value of \(\epsilon\).
For a nonzero nominal coefficient, the paired-radial rule is
\begin{equation}
    p_\mu(\epsilon,c)
    =p_{0,\mu}\left(1+\epsilon z_{c,\mu}\right),
    \qquad p_{0,\mu}\ne0,
    \label{eq:app-robustness-nonzero-rule}
\end{equation}
where \(\epsilon\) is the relative offset amplitude.
In the common cross-protocol comparison, a nominally zero coefficient is held fixed:
\begin{equation}
    p_\mu(\epsilon,c)=0,
    \qquad p_{0,\mu}=0.
    \label{eq:app-robustness-zero-rule}
\end{equation}
Thus each \(\mathbf z_c\) specifies one fixed trajectory through parameter space rather than a newly resampled perturbation at every \(\epsilon\).
A separate frozen-only scan allows nominally zero coefficients to acquire additive perturbations scaled by the nonzero coefficients in the same parameter group.
Those results and the separate group-ablation scans are not pooled with the common cross-protocol comparison in Fig.~\ref{fig:ksc-robustness-protocols} or Table~\ref{tab:app-robustness-summary}.\par

For each parameter vector, the time-dependent Hamiltonian is constructed with the unchanged transition-resolved projector model and the same normalized asymmetric first-drive envelope used for the nominal core.
Writing this Hamiltonian as \(H(t;\mathbf p)\), the core used in the robustness study is
\begin{equation}
    \begin{aligned}
    H_{\rm su}(t;\mathbf p)
    &=H(t;\mathbf p)-\frac{\Tr H(t;\mathbf p)}{9}I_9,\\
    K_{\rm raw}(\mathbf p)
    &=\mathcal T\exp\!\left[-i\int_0^T H_{\rm su}(t;\mathbf p)\,dt\right],\\
    K(\mathbf p)
    &=K_{\rm raw}(\mathbf p)\det\!\left(K_{\rm raw}(\mathbf p)\right)^{-1/9}.
    \end{aligned}
    \label{eq:app-robustness-time-dependent-core}
\end{equation}
The time-ordered propagator is evaluated with the same QuTiP \texttt{sesolve} and \texttt{vern9} settings as the nominal publication path; no polar or QR correction is applied.
At \(\epsilon=0\), both the parameter vector and the generated core reproduce the nominal publication path within the configured numerical tolerances.\par

The target data set contains \(1000\) Haar-random \(SU(9)\) matrices generated once, stored with hashes, and reused without regeneration.
All \(1000\) nominal compilations reach
\begin{equation}
    F_{\rm avg}\ge 0.999 .
    \label{eq:app-robustness-success-threshold}
\end{equation}
The nominal minimum, mean, and maximum fidelities are \(0.999000002\), \(0.999007837\), and \(0.999173682\), respectively.
The best-performing numerical local-parameter tensors, local matrices, tensor-product local layers, and nominal predicted circuits found by the optimizer are retained for both protocols.\par

In Protocol A, each perturbed core is recompiled against a fixed \(100\)-target subset of the full data set.
The dense scan uses \(20\) perturbation directions at
\begin{multline*}
\epsilon\in\{0,0.005,0.01,0.02,0.03,0.04,0.05,\\
0.06,0.07,0.08,0.10,0.125,0.15,0.20\}.
\end{multline*}
For every target--core pair, optimization begins from the corresponding nominal local-parameter tensor.
If the warm start does not reach Eq.~\eqref{eq:app-robustness-success-threshold}, up to five independently seeded random rescue restarts are used under the same optimizer and epoch budget.
For core direction \(c\), the primary statistic is the core-wise success fraction
\begin{equation}
    s_c(\epsilon)=\frac{1}{100}
    \sum_{i=1}^{100}
    \mathbf 1\!\left[F_{\rm avg}^{(i,c,\epsilon)}\ge0.999\right].
    \label{eq:app-robustness-core-success}
\end{equation}
This keeps the perturbed core, rather than a pooled target--core pair, as the statistical unit.\par

In Protocol B, the five local layers obtained from the nominal compilation are frozen exactly.
For each of \(100\) perturbation directions and every one of the \(1000\) targets, the evaluated circuit is
\begin{equation}
    U_{\rm frozen}^{(i,c,\epsilon)}
    =L_{5,i}^{0}K_{c,\epsilon}L_{4,i}^{0}K_{c,\epsilon}
     L_{3,i}^{0}K_{c,\epsilon}L_{2,i}^{0}K_{c,\epsilon}L_{1,i}^{0}.
    \label{eq:app-robustness-frozen-circuit}
\end{equation}
No optimization, gradient update, scheduler step, or random restart is invoked in this protocol.
We report the absolute target fidelity
\(F_{\rm target}=F_{\rm avg}(U_{\rm frozen},U_{\rm target})\)
and the degradation relative to the actual nominal synthesized circuit,
\(\Delta\mathcal I=F_0-F_{\rm target}\), rather than assuming \(F_0=1\).
At \(\epsilon=0\), the stored nominal circuits and fidelities are reconstructed to within \(10^{-12}\).\par

\begin{figure*}[!t]
    \centering
    \begin{minipage}{0.45\textwidth}
        \centering
        \includegraphics[width=\linewidth]{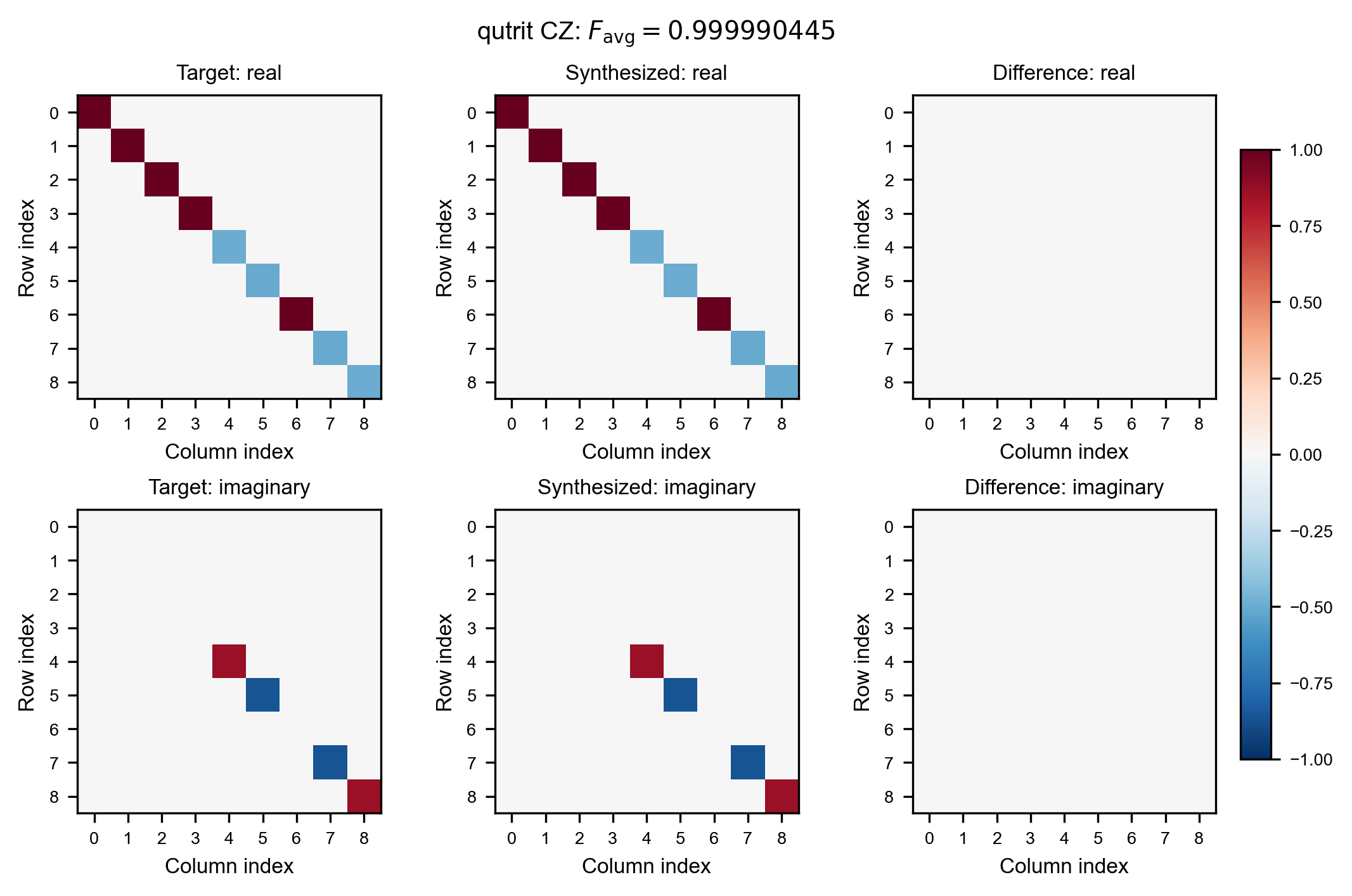}
        \vspace{-0.5em}
        \centerline{\small (a) qutrit CZ}
    \end{minipage}
    \hfill
    \begin{minipage}{0.45\textwidth}
        \centering
        \includegraphics[width=\linewidth]{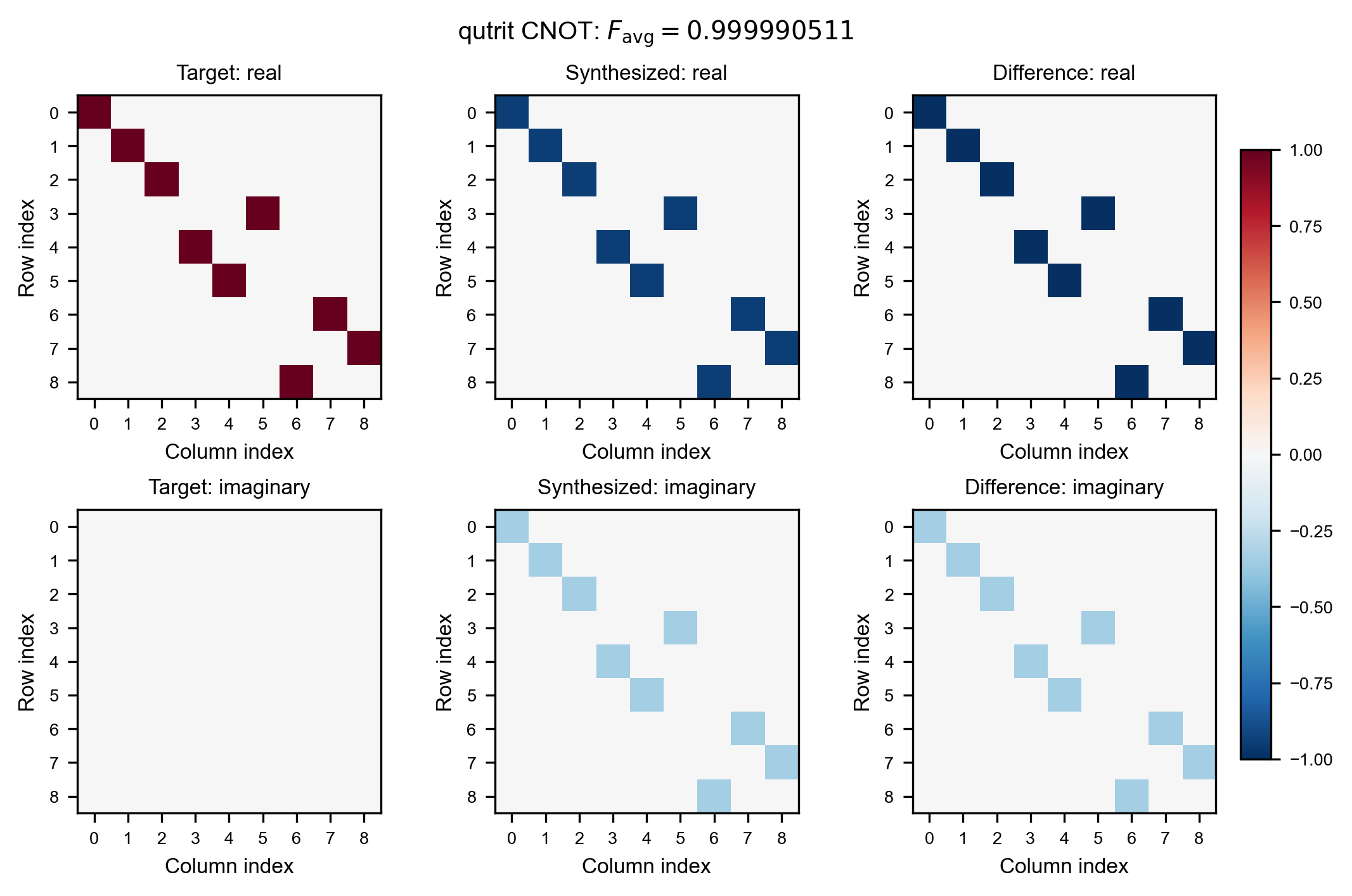}
        \vspace{-0.5em}
        \centerline{\small (b) qutrit CNOT}
    \end{minipage}
    \vspace{0.4em}

    \begin{minipage}{0.45\textwidth}
        \centering
        \includegraphics[width=\linewidth]{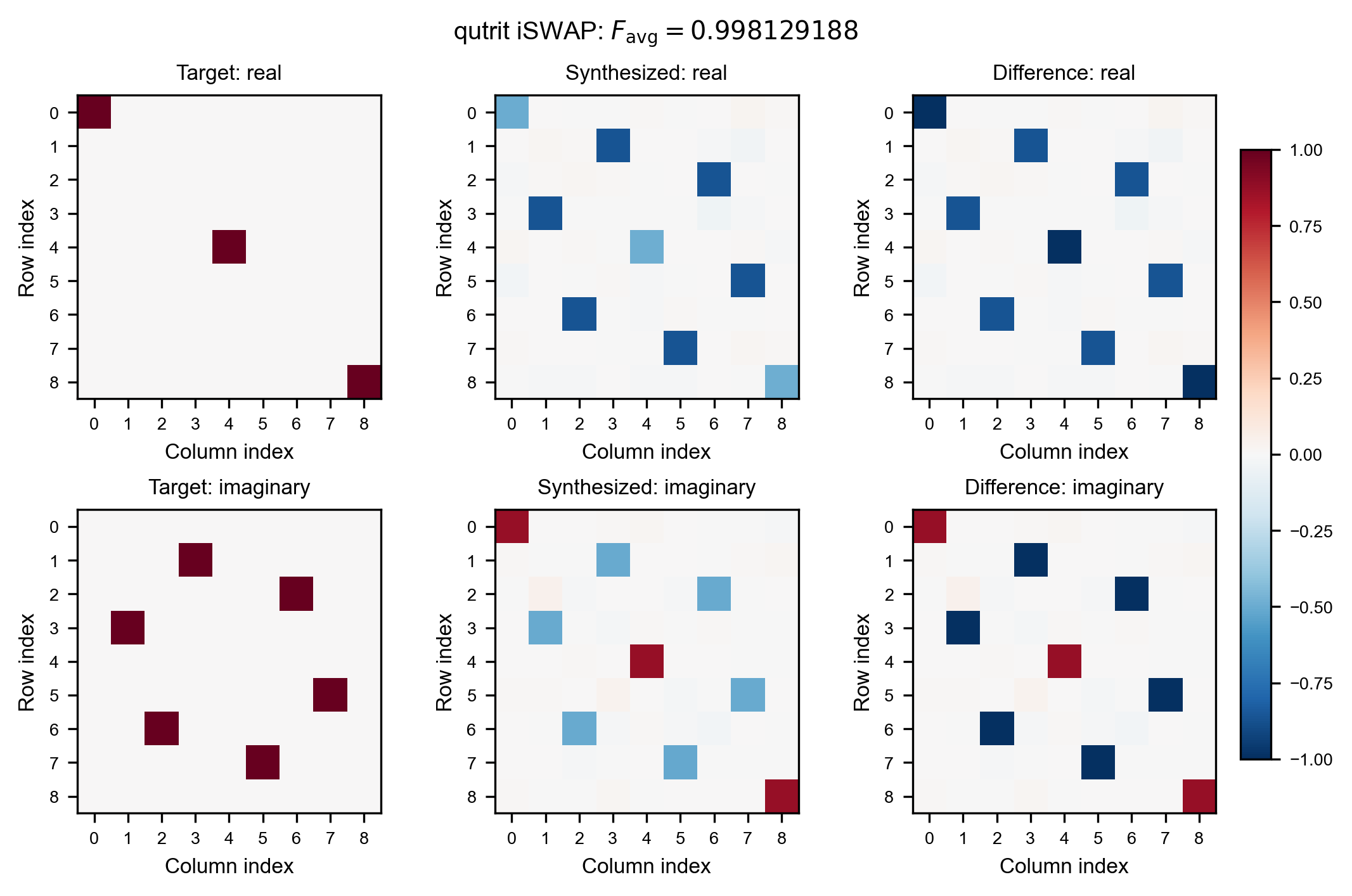}
        \vspace{-0.5em}
        \centerline{\small (c) qutrit iSWAP}
    \end{minipage}
    \hfill
    \begin{minipage}{0.45\textwidth}
        \centering
        \includegraphics[width=\linewidth]{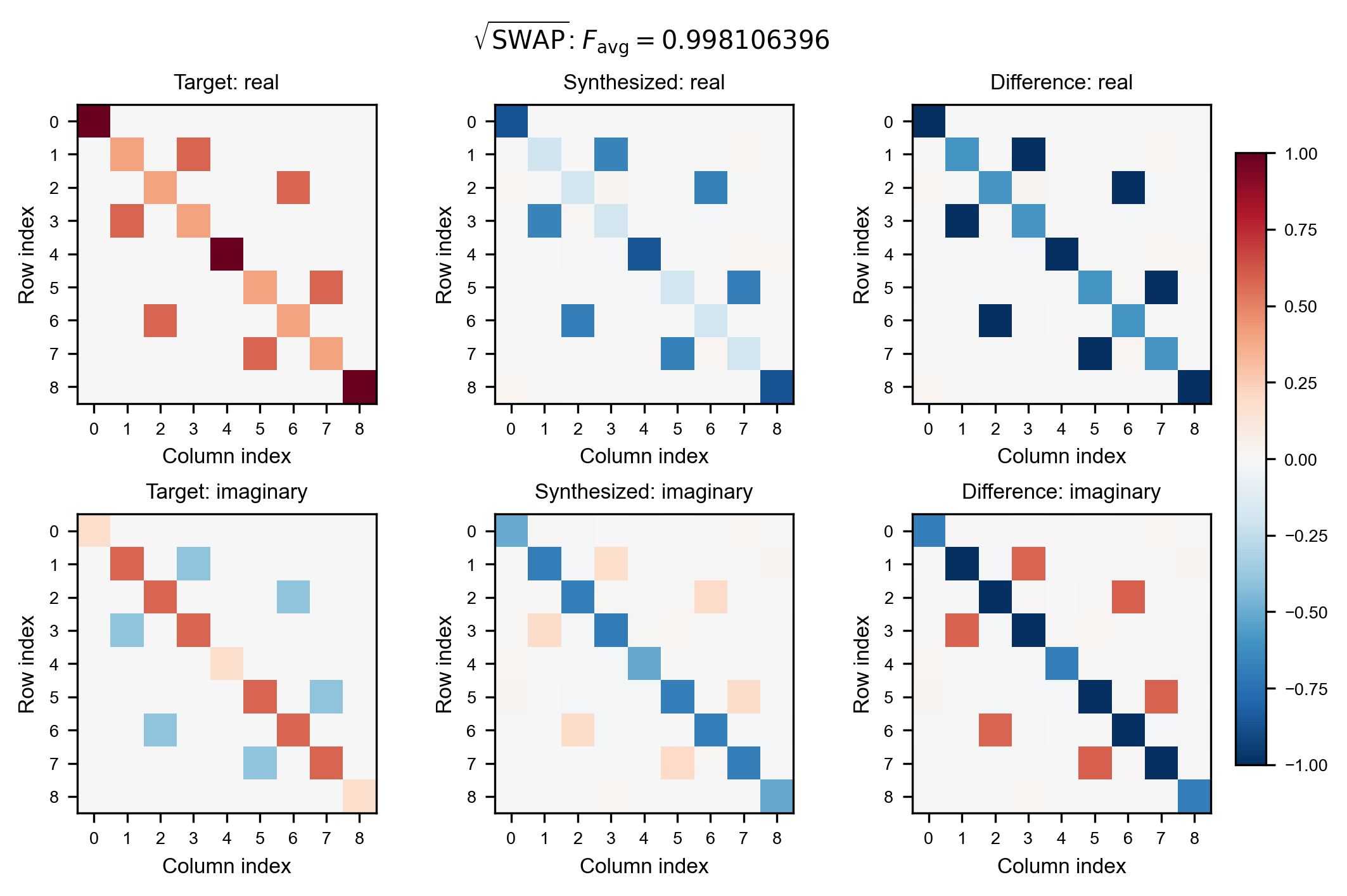}
        \vspace{-0.5em}
        \centerline{\small (d) \(\sqrt{\mathrm{SWAP}}\)}
    \end{minipage}
    \caption{
    Fixed-budget structured-target benchmarks for the representative superconducting core \(\Ksc\).
    Panel (a) compares the qutrit CZ target with its synthesized gate and reports \(F_{\rm avg}=0.999990445\).
    Panel (b) compares the qutrit CNOT target with its synthesized gate and reports \(F_{\rm avg}=0.999990511\).
    Panel (c) compares the qutrit iSWAP target with its best synthesized approximation and reports \(F_{\rm avg}=0.998129188\).
    Panel (d) compares the \(\sqrt{\mathrm{SWAP}}\) target with its best synthesized approximation and reports \(F_{\rm avg}=0.998106396\).
    Each panel shows real and imaginary parts of the target gate, synthesized gate, and elementwise difference.
    Before computing the elementwise differences, each synthesized gate is aligned to its target by the global phase that maximizes \(\operatorname{Re}\operatorname{Tr}(U_{\rm target}^{\dagger}U_{\rm synth})\).
    }
    \label{fig:app-ksc-structured-comparisons}
\end{figure*}

Selected results from the common range are summarized in Table~\ref{tab:app-robustness-summary}.
At \(\epsilon=0.20\), the Protocol-A core-wise success fractions span \(0.11\)--\(1.00\), confirming substantial direction dependence behind the median value of \(0.885\).
For Protocol B at the same amplitude, the mean target fidelity is \(0.640330\), and the minimum is \(0.108684\).
The fractions satisfying \(F_{\rm avg}\ge0.999\), \(0.99\), and \(0.95\) are \(0\), \(0\), and \(3\times10^{-5}\), respectively.
The recompilation scan was intentionally stopped after all \(20\) directions at \(\epsilon=0.20\) were complete; no recompilation results beyond this amplitude are reported.\par

The two protocols therefore support different, complementary statements.
Protocol A samples whether optimization can adapt the local layers to a perturbed but still fixed core, whereas Protocol B measures the execution sensitivity of circuits compiled for the nominal core.
The former retains a substantial population of high-performing sampled cores even at the largest reported amplitude, but the broad spread rules out any uniform numerical claim across directions.
The latter shows much stronger degradation when recompilation is forbidden.
Their logical status is the finite-precision one summarized at the end of Sec.~\ref{sec:introduction}; neither protocol yields a worst-case tolerance theorem or an analytic neighborhood guarantee.\par

The study is restricted to static offsets in the twelve-coefficient closed-system \(9\times9\) Hamiltonian model.
It does not include time-dependent noise, leakage outside the qutrit space, decoherence, finite-bandwidth distortion, readout errors, or closed-loop recalibration.
The reported quantile bands are empirical sample summaries, and individual coherent trajectories need not be monotone in \(\epsilon\).
Accordingly, Fig.~\ref{fig:ksc-robustness-protocols} should be read as a reproducible numerical comparison of recompilability and frozen-circuit sensitivity under the stated perturbation distribution and budgets.

\section{Structured-target benchmark details for \texorpdfstring{\(K_{\rm sc}\)}{Ksc}}
\label{app:sc-structured-targets}
\label{app:structured-targets}
This appendix records additional structured-target benchmark results for the representative superconducting core \(\Ksc\).
The main text reports the fixed-budget results for qutrit CZ, qutrit CNOT, qutrit iSWAP, and \(\sqrt{\mathrm{SWAP}}\).
Here we give the corresponding matrix-level comparisons.
These examples help illustrate the target dependence of the fixed-core synthesis protocol.\par
Throughout this benchmark, the superconducting core \(\Ksc\) is held fixed to the representative core defined by Table~\ref{tab:sc-representative-parameters}.
Only the five local layers in the four-core architecture are optimized.
All targets use \(100\) random restarts and \(5000\) epochs per restart, and the reported value is the best \(F_{\rm avg}\) found within this fixed budget.
The reference success threshold is \(F_{\rm avg}\ge0.999\), but threshold crossing is used only for reporting and not for early stopping.\par
The four matrix-level comparisons are collected in
Fig.~\ref{fig:app-ksc-structured-comparisons}.
The fixed-budget best values are \(0.999990445\), \(0.999990511\), \(0.998129188\), and \(0.998106396\) for qutrit CZ, qutrit CNOT, qutrit iSWAP, and \(\sqrt{\mathrm{SWAP}}\), respectively.
For the \(\sqrt{\mathrm{SWAP}}\) target, the best synthesized gate reaches
\begin{equation}
    F_{\rm avg}\!\left(\sqrt{\mathrm{SWAP}}\right)=0.998106396 .
    \label{eq:app-ksc-sqrtswap-agf}
\end{equation}
This value remains below \(F_{\rm avg}=0.999\) under the stated fixed budget.
Figure~\ref{fig:app-ksc-structured-comparisons}(d) shows the corresponding matrix comparison.
The remaining discrepancy is enough to prevent the gate from crossing the chosen success threshold.\par
Taken together, these examples show that not all structured targets are equally easy for the same fixed superconducting core.
In this run, the controlled-phase and controlled-addition targets cross the reporting threshold, whereas the two exchange-like targets do not.
This fixed-budget target dependence is a numerical observation rather than an obstruction theorem.\par
This target dependence is one concrete manifestation of the fact that the superconducting core is a hardware-motivated compromise rather than an algebraically ideal core.
The same fixed core performs strongly on the sampled Haar-random targets and on some structured gates, while the present optimizer and budget give lower fidelities for other structured targets.\par
These additional benchmarks therefore complement the main-text discussion in two ways.
First, they show that the superconducting core has nonuniform target-dependent performance even within the class of highly structured two-qutrit gates.
Second, they clarify that the limitation is not a failure on the sampled Haar-random targets, since the same core succeeds broadly on that tested ensemble, but rather target-dependent behavior under the stated fixed-core optimization protocol.
This distinction is important for interpreting the superconducting results as a hardware-motivated compromise between expressivity and implementability.

\paragraph*{Data availability.}
The numerical data and verification packages described in the appendices---the
regular-pair archive, the Haar-random benchmark outputs, the robustness scan
records, and the figure-generation inputs---are available from the authors
upon reasonable request.

\FloatBarrier
\bibliography{references}

\end{document}